\documentclass{article}

\usepackage{amsmath}
\usepackage{amsfonts}
\usepackage{amssymb}
\usepackage{graphicx}
\usepackage{amsthm}
\usepackage{thmtools}
\usepackage{cleveref}
\usepackage{autonum}
\usepackage{xcolor}
\usepackage{enumitem}
\usepackage{mathtools}
\usepackage{bbm}
\usepackage{booktabs}
\usepackage{mathrsfs}

\newtheorem{theorem}{Theorem}
\newtheorem{example}[theorem]{Example}
\newtheorem{lemma}[theorem]{Lemma}
\newtheorem{proposition}[theorem]{Proposition}
\newtheorem{remark}[theorem]{Remark}
\newtheorem{assumption}[theorem]{Assumption}

\crefname{theorem}{Theorem}{Theorems}
\crefname{assumption}{Assumption}{Assumptions}
\crefname{lemma}{Lemma}{Lemmas}
\crefname{proposition}{Proposition}{Propositions}
\crefname{figure}{Figure}{Figures}
\crefname{example}{Example}{Examples}
\crefname{section}{Section}{Sections}
\crefname{table}{Table}{Tables}

\newcommand{\R}{\mathbb{R}}
\newcommand{\dd}{\,\mathrm{d}}
\newcommand{\E}{\mathbb E}
\newcommand{\abs}[1]{\left\lvert#1\right\rvert}
\newcommand{\norm}[1]{\left\|#1\right\|}
\newcommand{\defeq}{\coloneqq}

\renewcommand{\d}{\mathrm{d}}
\renewcommand{\div}{\operatorname{div}}
\newcommand{\Rey}{\mathrm{Re}}

\begin{document}

\title{Stochastic modeling of Fourier modes \\ in two-dimensional turbulence \\ via filtered white noise}
\author{Paolo Cifani, Franco Flandoli, Andrea Zanoni \\ \\ \small{Scuola Normale Superiore, Pisa, Italy}}
\date{}
\maketitle

\begin{abstract}
Modeling turbulent flows by a random Fourier decomposition is a classical procedure in order to use simplified models of turbulence in heat transport and other applications. We investigate the Fourier time series of two-dimensional Navier--Stokes equations with friction and damping, forced at intermediate scales, and identify significant statistical structures. In particular, we find the existence of a typical time correlation length, and propose a stochastic model for the Fourier components. Finally, we compute the transport of a passive scalar under advection-diffusion dynamics by means of direct numerical simulation of the stochastic damped Navier--Stokes equation and compare it with analytical predictions of the effective diffusion produced by the stochastic model.
\end{abstract}

\section{Introduction} \label{sec:intro}

Investigating the effects of turbulent flows on passive scalars and other transported quantities, many works in the past have introduced stochastic models of fluids, which often provide analytical results not possible if the fluid is described by the solution of fluid dynamic equations \cite{ABC25,AFP21,ECV23,FlL23,MTV01}. 

Stochastic models are an idealization of fluid motion and there is no unique recipe. Nevertheless, they provide a robust mathematical framework for capturing the probabilistic nature of turbulent fluctuations. Traditionally, these models find their place in classical decorrelation turbulence theory \cite{Kra64,Pop94,Ten75,Yeu02}. As established in literature, the Taylor Eulerian scale $\tau_E$ is dominated by the advection of dissipative eddies by the most energetic eddies. Conversely, its Lagrangian counterpart $\tau_L$, which characterizes the internal dynamic straining and continuous movements, is governed by the local eddy turnover time. At observation times exceeding the temporal microscale, stochastic modelling becomes physically justifiable. However, these physical scales, while robust, provide an order of magnitude rather than a precise cutoff. In order to construct a model that accurately reproduces the statistical flow dynamics, here we follow a data-driven approach. Specifically, we derive a stochastic model which incorporates an autocorrelation function whose form closely adheres to that of the numerically measured autocorrelation. While this approach inherits the pitfall of being data-dependent rather than general, it represents a critical first step towards actual stochastic modelling of 2D turbulent flows and offers insights into where the classical Ornstein--Uhlenbeck model fails.  

We investigate this issue in the case of the two-dimensional Navier--Stokes equation, activated at small spatial scales and with a friction damping, namely described in vorticity form as
\begin{equation} \label{eq:NavierStokes_vorticity}
\partial_t\omega + \mathbf u \cdot \nabla\omega + \alpha\omega - \nu \Delta\omega = \mathcal F,
\end{equation}
where $\alpha > 0$ denotes the linear damping coefficient and $\nu > 0$ is the kinematic viscosity. The stochastic forcing term $\mathcal F$ is localized in spectral space and defined by
\begin{equation}
\widehat{\mathcal F}(\mathbf k) = \mathbbm{1}_{\{ \abs{\mathbf k} = k_f \}} \frac{F}{\sqrt{\mathsf n}} \frac{1}{\sqrt{2}} \left( \dot W_1 + i \dot W_2 \right),
\end{equation}
where $k_f$ is the forcing wavenumber, $F = 10$ sets the forcing amplitude, $\dot W_1$ and $\dot W_2$ are independent Gaussian white-noise processes in time, and $\mathsf n$ is the number of Fourier modes satisfying $|\mathbf k| = k_f$. This system develops an inverse cascade of energy and a direct enstrophy cascade \cite{BoE12,BoM10,KrM80,Lei97,Tab02}, as also shown in \cref{fig:vorticity_spectrum}. The spectral grid is discretized with $\mathsf N = 512$ modes, and the forcing is concentrated at wavenumbers $\abs{\mathbf k} \sim k_f = 64$. The inertial spectrum $S$ extends from the injection scale back to lower wavenumbers, to a level that depends on $\alpha$. Consistently, snapshots of the vorticity field in the stationary regime show coherent structures larger than the forcing scale, whose characteristic size increases as $\alpha$ decreases. In \cref{fig:vorticity_spectrum}, we report three simulations with different values of $\alpha$. For each case, we specify the Reynolds numbers associated with the inverse energy cascade, controlled by $\alpha$, and the direct enstrophy cascade, controlled by $\nu$, in \cref{tab:Reynolds}. In particular, following the derivation in \cite[Section 2]{BoE12}, we define the Reynolds numbers as
\begin{equation}
\Rey_\alpha = \frac{k_f}{\pi \alpha} \left( \sum_{\mathbf k} S_{\abs{\mathbf k}}^{(u)} \right)^{1/2} \qquad \text{and} \qquad \Rey_\nu = \frac{\pi^2}{k_f^2 \nu^{2/3}} \left( \sum_{\mathbf k} \abs{\mathbf k}^2 S_{\abs{\mathbf k}}^{(\omega)} \right)^{1/3},
\end{equation}
with viscosity $\nu = 10^{-5}$ and damping parameters $\alpha_1 = 10^{-1}$, $\alpha_2 = 1.25 \cdot 10^{-2}$, $\alpha_3 = 3 \cdot 10^{-3}$. $\Rey_\alpha$ quantifies the effect of large-scale linear damping and controls the development of the inverse energy cascade, while $\Rey_\nu$ measures the relative importance of inertial and viscous effects at small scales and characterizes the extent of the direct enstrophy cascade. We notice that the different behavior of the low-wavenumber part of the spectrum for the $\alpha_1$ case is due to the stronger linear damping. In fact, the inverse energy transfer is arrested earlier, so less energy reaches the largest scales, leading to the different spectral shape observed at low wavenumbers. In contrast, for smaller values of $\alpha$, the weaker damping allows energy to be transferred to larger scales before being dissipated.

\begin{table}
\begin{center}
\begin{tabular}{ccccc}
\toprule
&& $\alpha_1$ & $\alpha_2$ & $\alpha_3$ \\
\midrule
$\mathrm{Re}_\alpha$ && $5.03e1$ & $1.21e3$ & $1.25e4$ \\
$\mathrm{Re}_\nu$ && $5.67e2$ & $6.24e2$ & $6.18e2$ \\
\bottomrule
\end{tabular}
\end{center}
\caption{Reynolds numbers associated with the inverse energy and direct enstrophy cascades in the experiments shown in \cref{fig:vorticity_spectrum}.}
\label{tab:Reynolds}
\end{table}

\begin{figure}
\begin{center}
\begin{tabular}{ccc}
\includegraphics{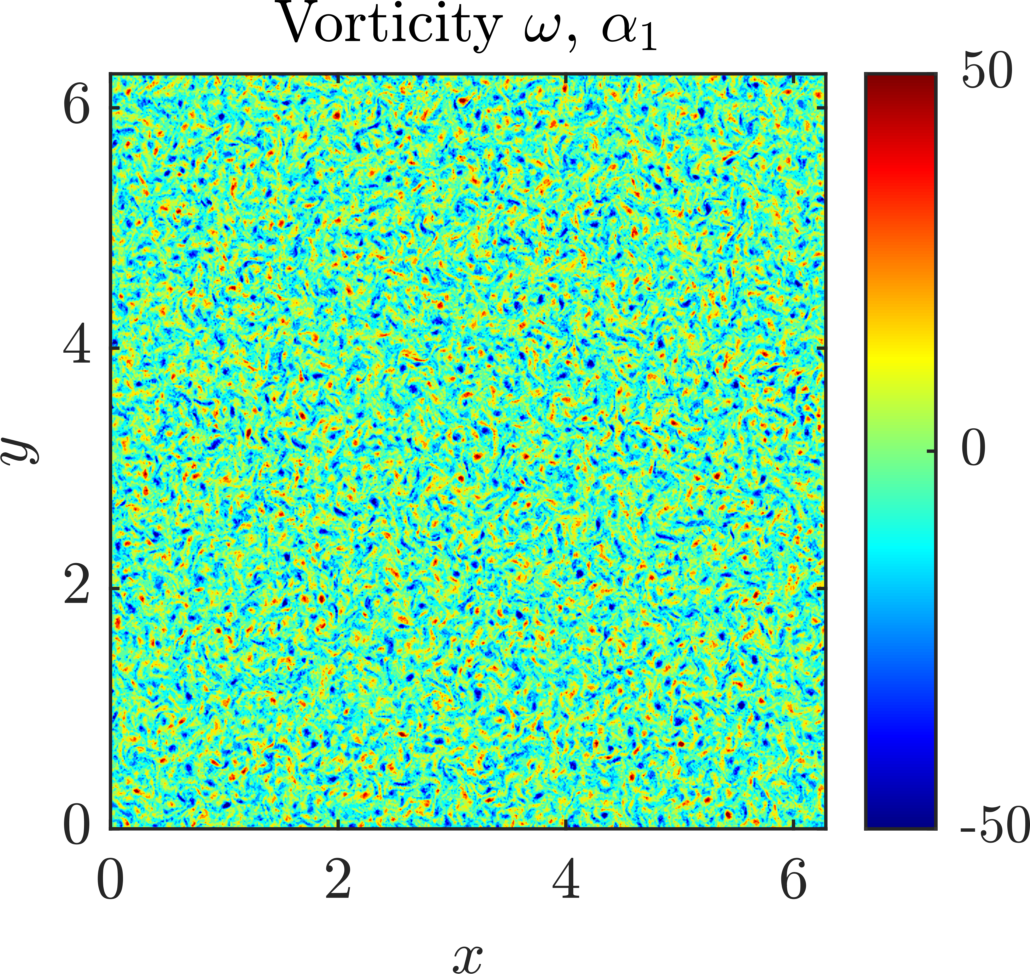} &&
\includegraphics{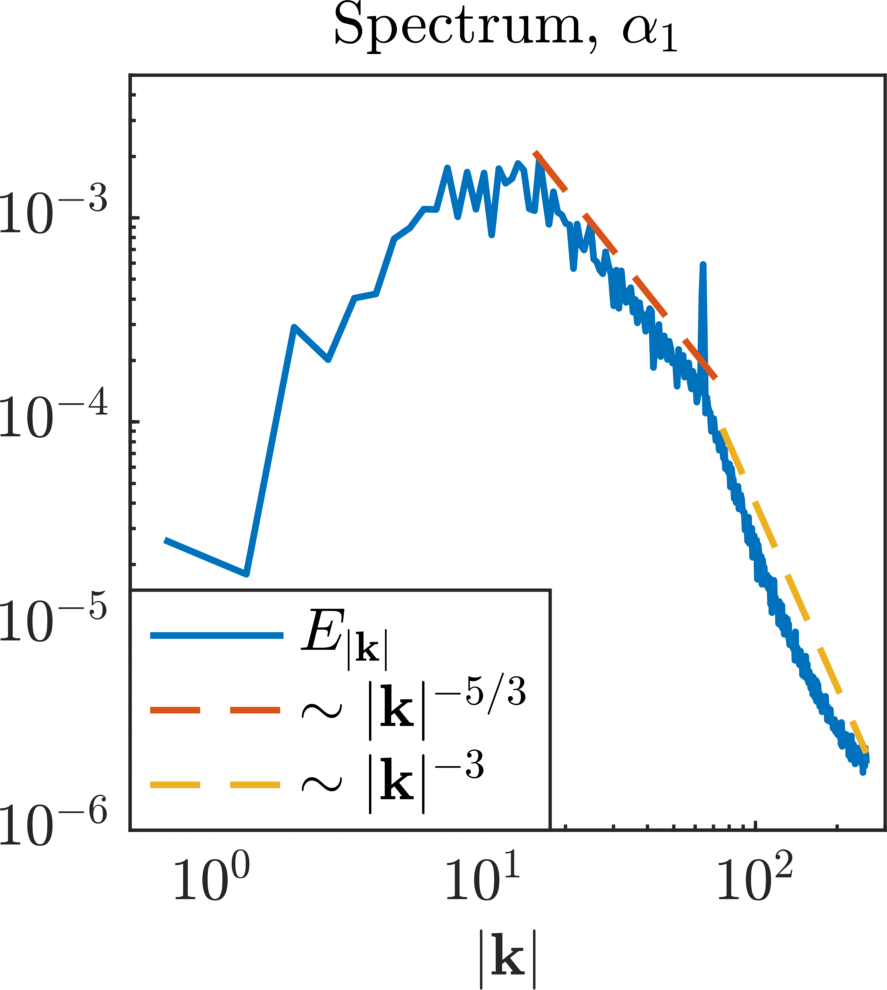} \\[5pt]
\includegraphics{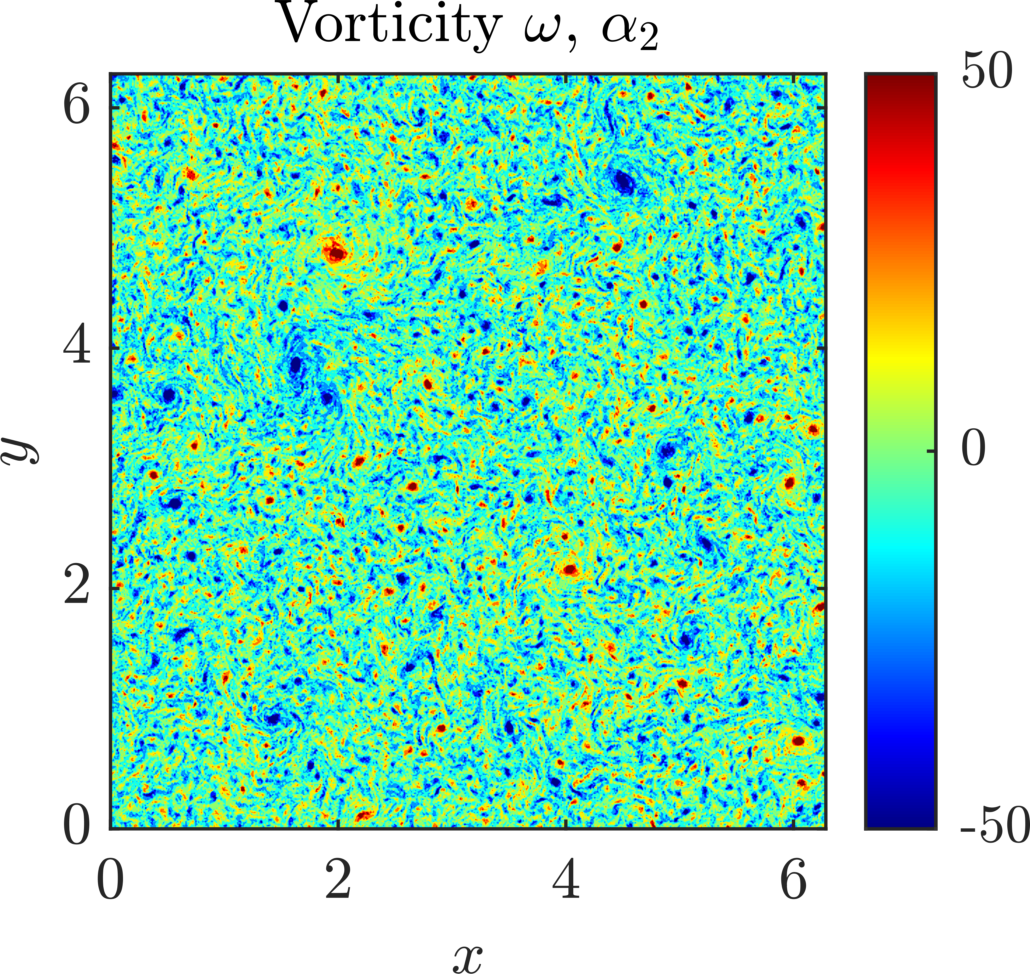} &&
\includegraphics{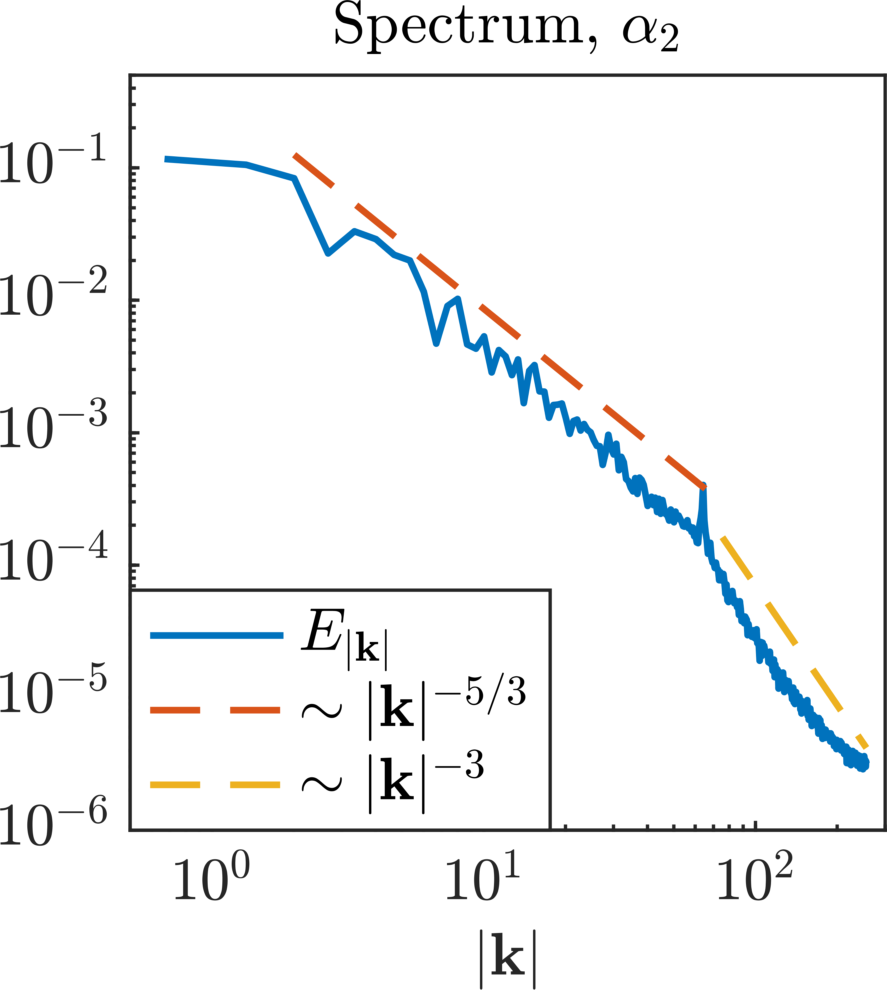} \\[5pt]
\includegraphics{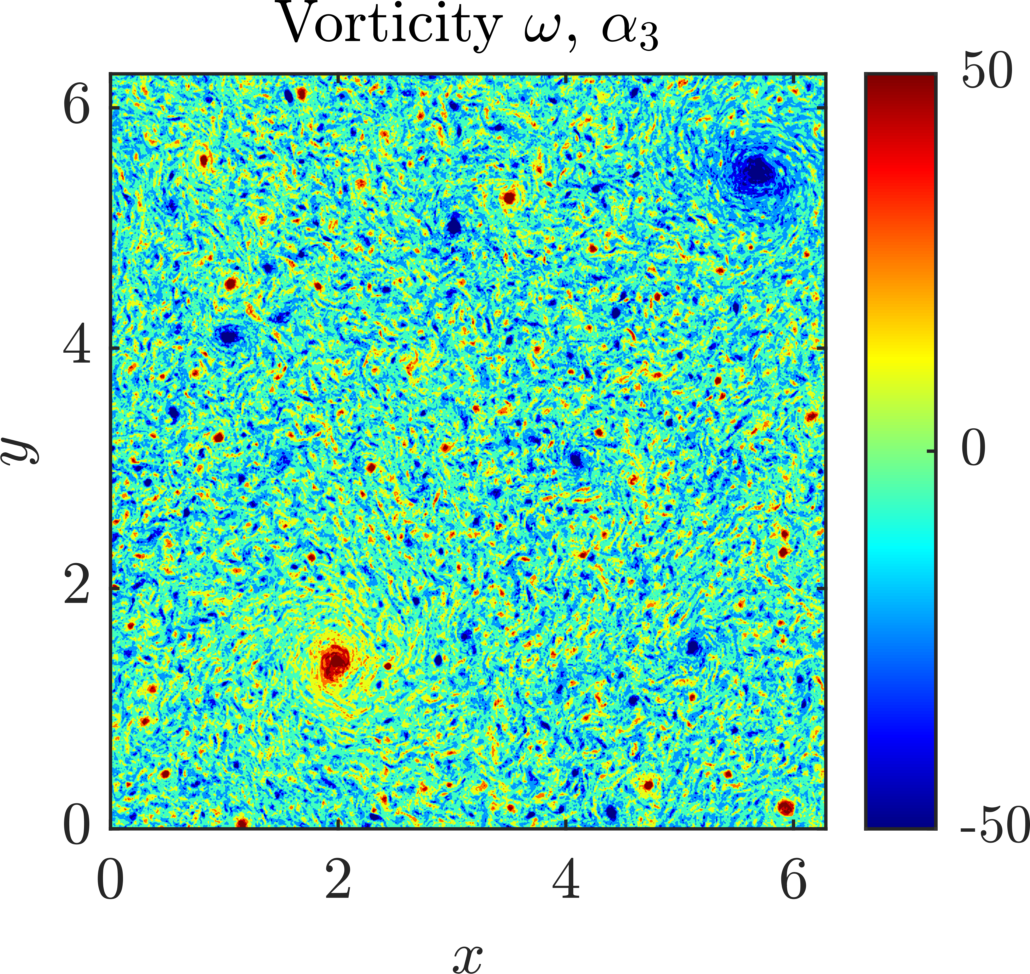} &&
\includegraphics{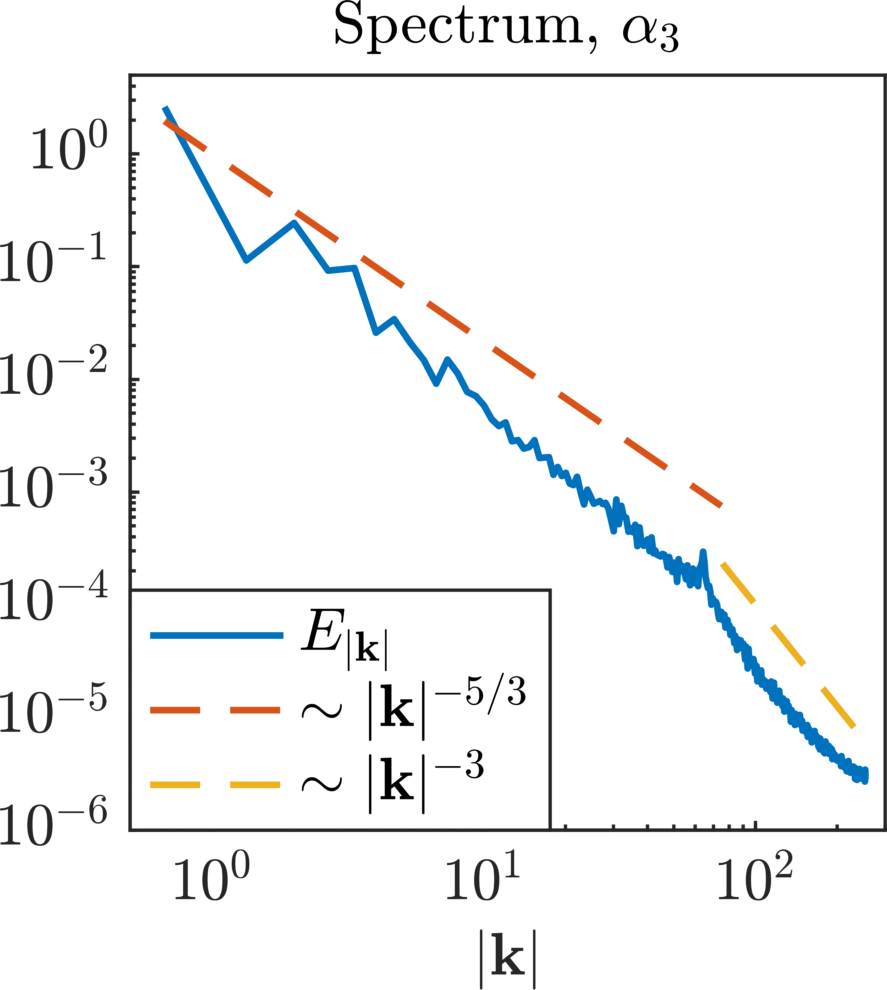}
\end{tabular}
\end{center}
\caption{Vorticity field $\omega$ and energy spectrum $E_k$ at the final time step of the simulation, for decreasing values of the damping parameter $\alpha$, i.e., $\alpha_1 > \alpha_2 > \alpha_3 > 0$.}
\label{fig:vorticity_spectrum}
\end{figure}

From the numerical simulations of this system on the torus $[0, 2\pi]^2$, we extract a velocity field, denoted below by $\mathbf u^{\text{NS}}$, that is then used in the equation
\begin{equation}
\partial_t T + \mathbf u^{\text{NS}} \cdot \nabla T = \kappa \Delta T,
\end{equation}
to examine the advection-diffusion of a passive scalar $T$. More precisely, we consider the Lagrangian dynamics for $t \ge 0$
\begin{equation}
\begin{aligned}
\frac{\d \mathbf x^{\text{NS}}(t)}{\d t} &= \mathbf u^{\text{NS}} \left(  \mathbf x^{\text{NS}}(t), t \right), \\
\mathbf x^{\text{NS}}(0) &= 0,
\end{aligned}
\end{equation}
and estimate, by Monte Carlo method, the function
\begin{equation} \label{eq:variance_true_function}
t \mapsto \E \left[ \abs{\mathbf x^{\text{NS}}(t)}^2 \right].
\end{equation}
These numerical computations aim to approximate, up to small numerical errors and the idealization of a stochastic white noise input and periodic boundary conditions, the real properties of the transport of a scalar quantity in a turbulent two-dimensional fluid. 

Our goal consists in modeling the turbulent fluid by means of a stochastic process $\mathbf u^{\text{model}}$ that is as simple as possible while remaining sufficiently realistic. Specifically, we require that the solution of
\begin{equation}
\begin{aligned}
\frac{\d \mathbf x^{\text{model}}(t)}{\d t} &= \mathbf u^{\text{model}} \left(  \mathbf x^{\text{model}}(t), t \right), \\
\mathbf x^{\text{model}}(0) &= 0,
\end{aligned}
\end{equation}
reproduces, to a good approximation, the behavior of the true system, in the sense that
\begin{equation} \label{eq:variance_model_true}
\E \left[ \abs{\mathbf x^{\text{model}}(t)}^2 \right] \simeq \E \left[ \abs{\mathbf x^{\text{NS}}(t)}^2 \right].
\end{equation}
We are interested especially on the main statistical features, rather than sophisticated ones like anomalous scalings, intermittency, and similar properties. 

We choose, a priori, to work with a model of the form
\begin{equation} \label{eq:model_u}
\mathbf u(\mathbf x, t) = \sum_{\mathbf k} \lambda_{\mathbf k} \mathbf e_{\mathbf k} (\mathbf x) \xi_t^{\mathbf k},
\end{equation}
where $\lambda_{\mathbf k}$ are real numbers, $(\mathbf e_{\mathbf k })_{\mathbf k}$ is a classical complete orthonormal system of divergence free vector fields on the two-dimensional torus, and $\xi_t^{\mathbf k}$ are real valued stationary stochastic processes \cite{CCH23,DeM25,Mem14}. The choice of $\lambda_{\mathbf k}$ follows standard approaches in turbulence modeling and is related to the energy spectrum (see, e.g., \cite[Chapter 6]{Pop00}), supposing that $\xi_t^{\mathbf k}$ are suitably normalized. The main difficulty therefore lies in the selection of the stochastic processes $\xi_t^{\mathbf k}$.

Two classical choices for $\xi_t^{\mathbf k}$ that appear in the literature are white noise (see, for instance, \cite{ECL23,FGL22,FlL22,FlT26,HaM17,Hol15,Kra94,PFH23,PPV23,RMC17}) and Ornstein--Uhlenbeck processes or, more generally, processes defined as solutions of stochastic differential equations (see, among others, \cite{ABC25,ACM22,CLR21,CiF25,DeP26,Pap22}). Even if these two classes are natural and promising, a closer analysis of the data does not confirm them. We investigate the statistical properties of the Fourier components of $\mathbf u^{\text{NS}}$ obtained from numerical simulations and first find a clear role played by a relaxation time $\tau_{\mathbf k}$, which rules out direct modeling by white noise. Second, we also observe a clear discrepancy with respect to an Ornstein--Uhlenbeck process with damping time $\tau_{\mathbf k}$. In fact, the statistical structure seems quite precise and differs from these two classes of processes. We therefore propose a general class of models that is suitable to cover the processes found in the numerical simulations. We also note that in \cite{PPB01} more complex classes of stochastic models based on systems of nonlinear stochastic differential equations are proposed and analyzed in the context of tracer dispersion.

In particular, we consider a family of stochastic processes obtained by filtering white noise with a suitable kernel $\theta$, scaled by the relaxation time $\tau_{\mathbf k}$, and defined by the Itô integral
\begin{equation} \label{eq:model_xitk}
\xi_t^{\mathbf k} = \int_{-\infty}^{+\infty} \frac1{\sqrt{\tau_{\mathbf k}}} \theta \left( \frac{t-s}{\tau_{\mathbf k}} \right) \dd W_s.
\end{equation}
This construction allows us to control both the smoothness of $\xi_t^{\mathbf k}$ and its autocorrelation structure through the choice of the kernel $\theta$. In particular, by selecting Gaussian kernels, we are able to reproduce the main statistical features, including the autocorrelation, observed in the Fourier modes extracted from the numerical simulations. This leads to our initial formulation of the model $\mathbf{u}^{\text{model}}$, obtained by substituting \eqref{eq:model_xitk} into equation \eqref{eq:model_u}.

We then analytically investigate the diffusion properties of the synthetic model, based on a Green--Kubo type formula. In both the short-time and long-time regimes, we observe agreement between model and data, as required in \eqref{eq:variance_model_true}. Moreover, these results are consistent with the classical transition from ballistic to diffusive behavior in tracer dispersion \cite{FGV01,Tay21}.

The motivation and main novelties of this work are summarized below.
\begin{itemize}[leftmargin=*]
\item We focus on the Eulerian temporal correlation structure in the two-dimensional setting, where the inverse cascade dominates. This aspect is much less investigated than in the classical three-dimensional setting, where the direct cascade dominates and studies mainly focus on spatial correlations or Lagrangian dynamics.
\item We observe that the temporal statistical structure differs from that assumed in classical stochastic models, such as the delta-correlation adopted in many idealized models \cite{Kaz68,Kra68} or the exponential decay of correlations associated with a Markovian structure, often used as an idealization of short but finite correlation times. Instead, we find that autocorrelations decay faster than exponentially and exhibit a similar temporal structure across different Fourier components.
\item We show that the Fourier-based synthetic model introduced here can be used to compute diffusion coefficients analytically, opening the door to applications in turbulent diffusion.
\end{itemize}

Finally, an important question is the range of validity of the results presented in this work. First, we remark that the analysis applies when $\alpha$ is large or moderate, namely when the structures created by the inverse cascade are still not very large and continue to contribute to a sort of chaotic motion of small-to-moderate scale structures. The turbulence and transport features indeed change when $\alpha$ is small and the largest inverse-cascade structures become very large.

Second, our analysis is based exclusively on time series obtained from direct numerical simulations of the two-dimensional Navier--Stokes equations on a periodic domain, forced at intermediate scales and subject to significant friction damping. We have not analyzed any other datasets, and therefore our conclusions are restricted to this particular, already idealized, fluid model. In particular, we cannot extrapolate these results to three-dimensional or even quasi-two-dimensional flows \cite{Ale23}, nor can we conclude that they hold for real turbulent flows.

As an example of a mechanism that could lead to deviations from the statistical behavior observed here, one may consider the presence of coherent vortical structures. In two-dimensional turbulence, such structures emerge when the friction damping is too weak to suppress their formation through the inverse cascade. In three-dimensional turbulence, analogous coherent structures arise, for example, as vortex filaments during the direct cascade. In the first case, the large size and relatively coherent motion of these structures are expected to produce much longer temporal correlations, making a simplified statistical model less appropriate, or at least requiring a more careful quantification of the deviations from the model proposed here, including the assumption that the stochastic processes associated with different Fourier components are independent. In the second case, the temporal correlations may still remain short, but their statistical structure is likely to be more complex, possibly involving intermittency effects that are neglected in the present model and are less pronounced in the idealized two-dimensional setting.

\paragraph{Outline.} The rest of the paper is organized as follows. In \cref{sec:numerics_Fourier}, we analyze the correlation properties of the time series of the Fourier components of the solution of the two-dimensional stochastic Navier--Stokes equations obtained by direct numerical simulation, identifying common features. In \cref{sec:model_process}, we summarize the numerical results by introducing a model for the time series of \cref{sec:numerics_Fourier}, investigating the properties of this model and comparing it with the data. In \cref{sec:velocity_field}, we introduce a synthetic turbulent velocity field based on the stochastic model of the Fourier components identified above, and we present a white-noise approximation obtained in the limit of small correlation time. In \cref{sec:heat_equation}, we compare the diffusion properties of this synthetic turbulent velocity field with those obtained from direct numerical simulations. Finally, in \cref{sec:conclusion}, we draw the conclusions.

\section{Numerical results on the Fourier components} \label{sec:numerics_Fourier}

The numerical results presented in this work are obtained using a fully dealiased (2/3 rule) pseudo-spectral solver for the two-dimensional Navier--Stokes equations \cite{CHQ07}. In particular, the numerical scheme conserves enstrophy up to machine precision whenever a quadratic-invariant time-integration method is employed. In practice, a third-order Runge--Kutta scheme is used with a CFL number of 0.3, providing a good balance between accuracy and computational efficiency. In this section, we consider a slightly different setup than the one used in the simulations presented in \cref{sec:intro}. In particular, the spatial discretization uses $\mathsf N = 1024$ modes, while the kinematic viscosity and damping coefficients are $\nu = 10^{-5}$ and $\alpha = 1.25 \cdot 10^{-2}$, respectively. We have compared the Fourier statistics obtained with $\mathsf N = 512$, used in \cref{sec:intro}, and $\mathsf N = 1024$, and found no significant differences in the post-processing of the time series presented here. All numerical experiments are performed for different forcing scales $k_f = 50$ and $k_f = 100$. These parameters correspond to Reynolds numbers associated with the inverse energy and direct enstrophy cascades of $\Rey_\alpha = 4.52e2$ and $\Rey_\nu = 9.78e2$ for $k_f = 50$, and $\Rey_\alpha = 5.82e2$ and $\Rey_\nu = 3.87e2$ for $k_f = 100$. The resolution is such that $\mathsf N L_\nu \approx 6.5$ for $k_f=50$ and $\mathsf N L_\nu \approx 5$ for $k_f=100$, where $L_\nu = \nu^{1/2} / \eta_nu^{1/6}$ denotes the enstrophy dissipation scale. Lagrangian trajectories are integrated by means of a second-order Adams--Bashforth scheme, and statistics are collected after the flow had reached a statistically stationary state.

In this section, we focus on the real Fourier components of the vorticity field $\omega$, which we denote by $f_t^{\mathbf k}$, where $\mathbf k$ is the corresponding wavenumber. A white-noise forcing is applied to the modes with wavenumbers $\abs{\mathbf k} \sim k_f$, where the forcing wavenumber is either $k_f = 50$ or $k_f = 100$. In particular, we consider the modes $\mathbf k = (50,0)$ and $\mathbf k = (100,0)$ for the two cases, which satisfy exactly $\abs{\mathbf k}=k_f$. In addition, we analyze the modes
\begin{equation} \label{eq:modes_used}
K = \begin{bmatrix}
3 & 18 & 34 & 71 & 177 & 354 & 512 \\
4 & 18 & 37 & 71 & 177 & 354 & 512
\end{bmatrix},
\end{equation}
whose moduli are
\begin{equation}
\abs{K} \simeq \begin{bmatrix}
5 & 25.46 & 50.25 & 100.41 & 250.32 & 500.63 & 724.08
\end{bmatrix}.
\end{equation}
The trajectories $f_t^{\mathbf k}$ are simulated up to a final time $T = 1000$ using a time step $\delta = 0.005$, which is sufficiently small to resolve the smallest scales. Consequently, all Fourier coefficient time series are adequately resolved, except for the modes corresponding to the forcing wavenumber $k_f$, where white-noise forcing is applied directly.

\begin{figure}
\begin{center}
\begin{tabular}{cc}
\includegraphics{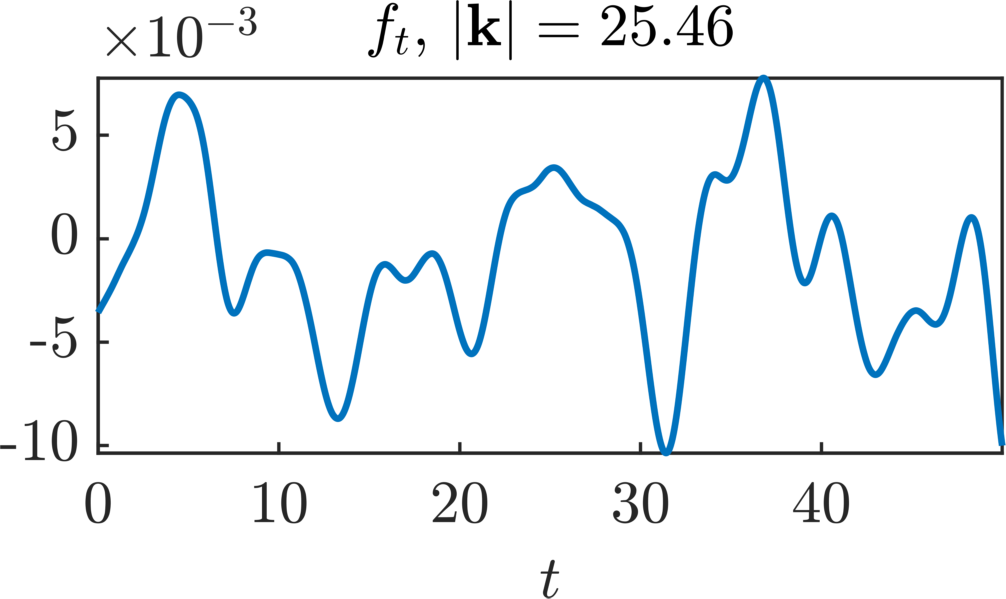} &
\includegraphics{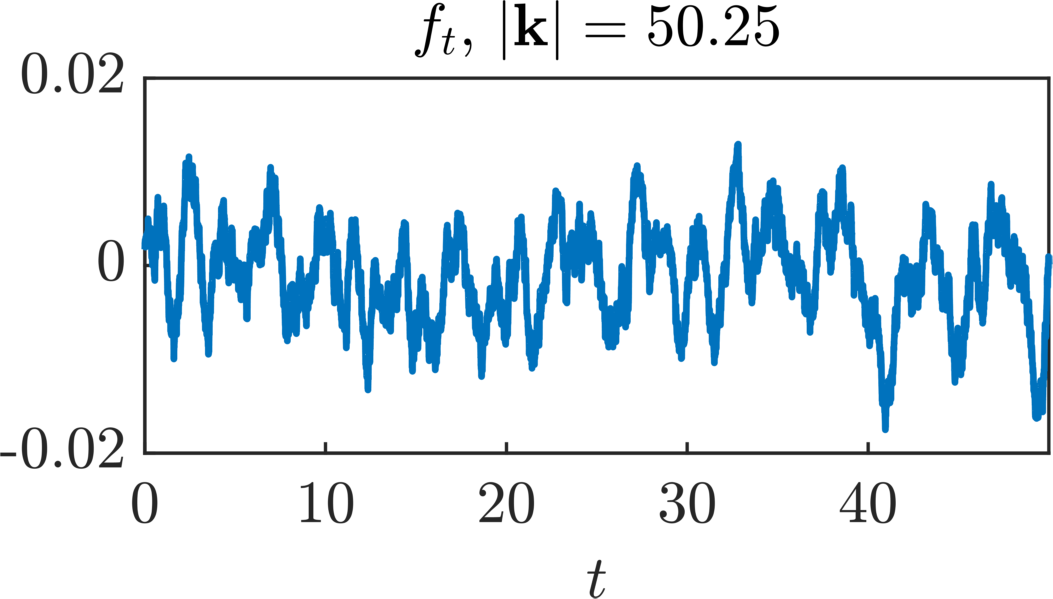} \\[5pt]
\includegraphics{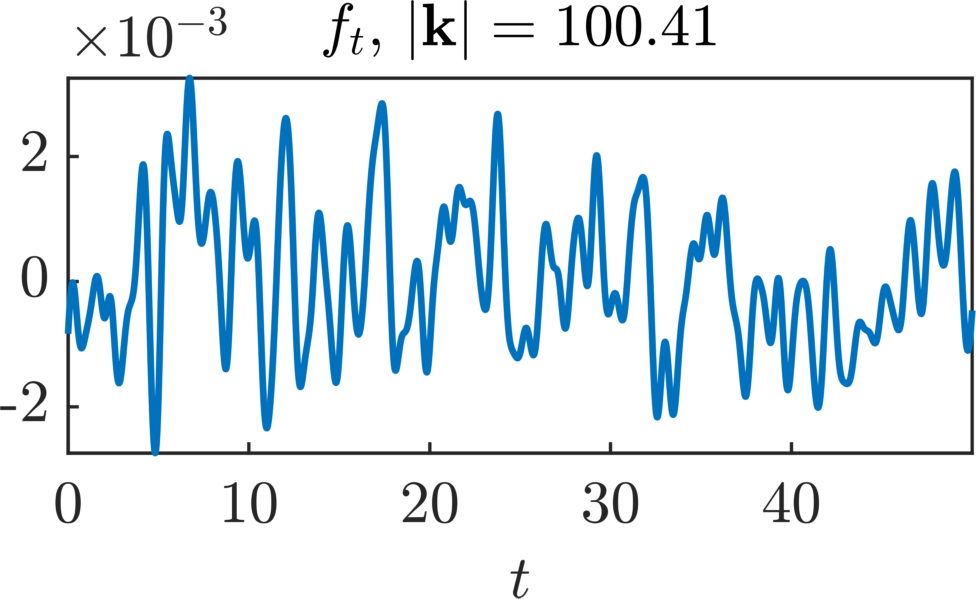} &
\includegraphics{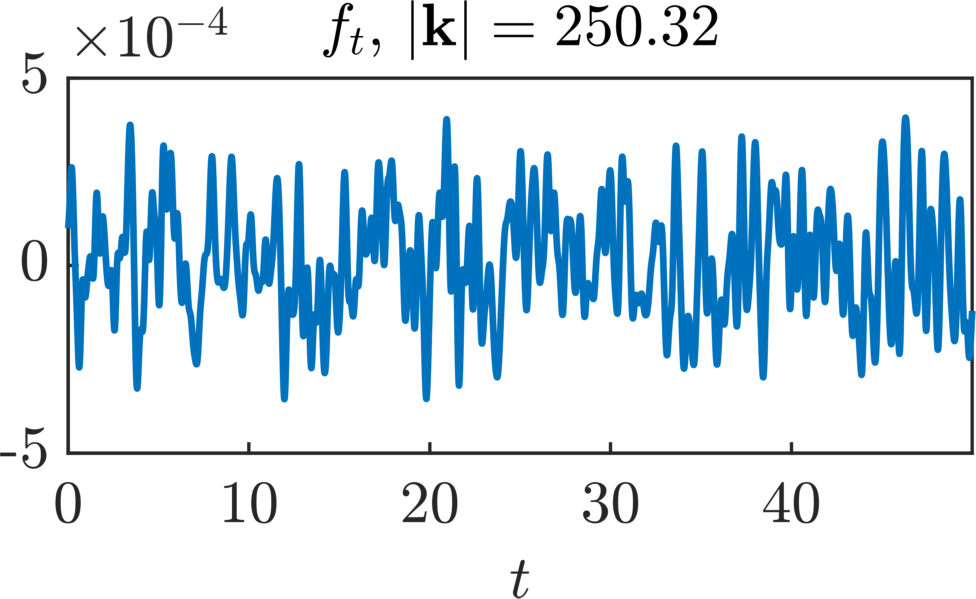}
\end{tabular}
\end{center}
\caption{Trajectories of several Fourier modes with forcing applied at $\abs{\mathbf k} \sim k_f = 50$.}
\label{fig:Fourier_trajectory_k50}
\end{figure}

\begin{figure}
\begin{center}
\begin{tabular}{cc}
\includegraphics{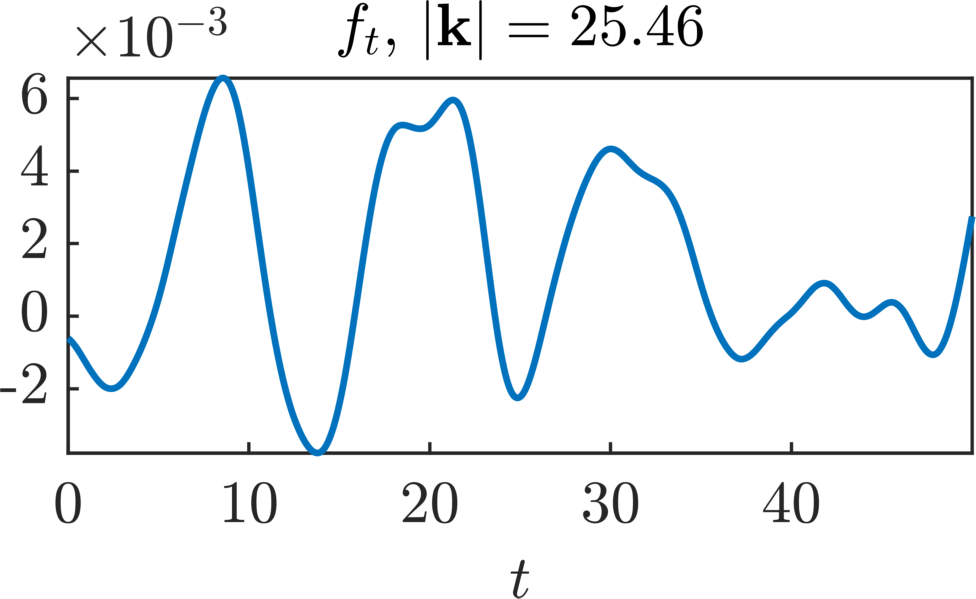} &
\includegraphics{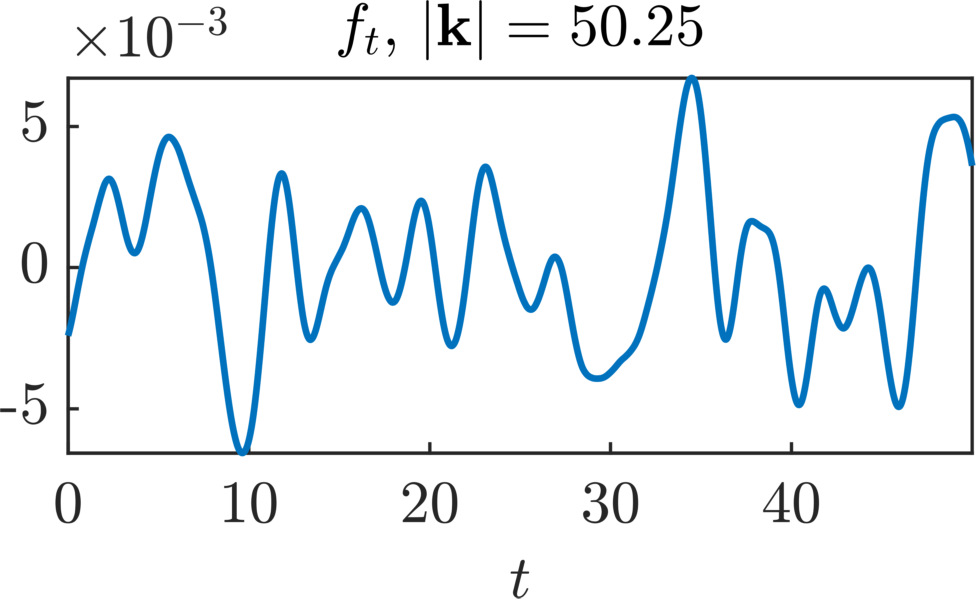} \\[5pt]
\includegraphics{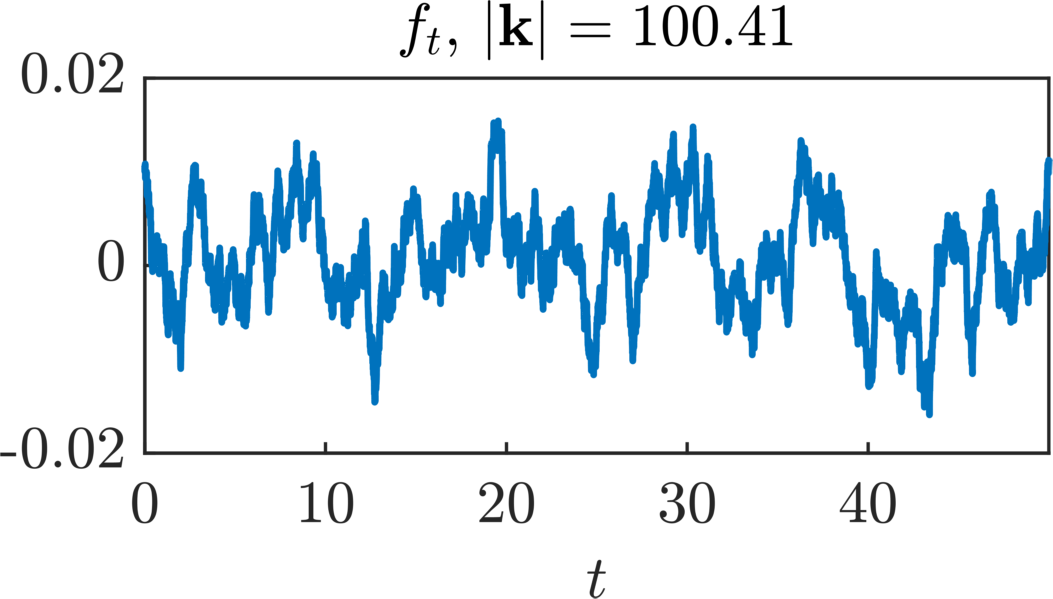} &
\includegraphics{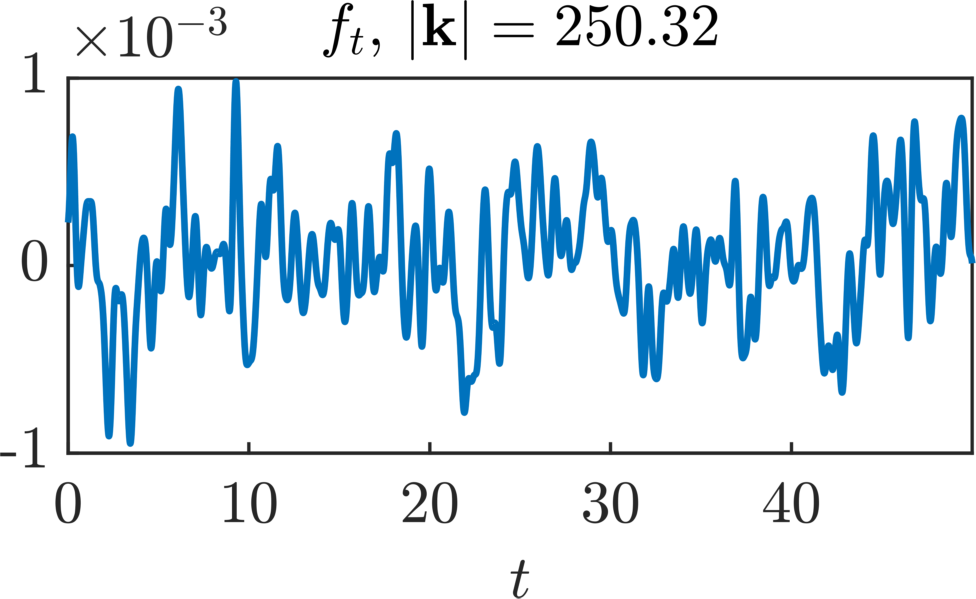}
\end{tabular}
\end{center}
\caption{Trajectories of several Fourier modes with forcing applied at $\abs{\mathbf k} \sim k_f = 100$.}
\label{fig:Fourier_trajectory_k100}
\end{figure}

We first visualize the trajectories in \cref{fig:Fourier_trajectory_k50,fig:Fourier_trajectory_k100} for a subset of the modes in \eqref{eq:modes_used}. The Fourier modes that are not directly forced appear smooth, with oscillations of increasing frequency as the wavenumber increases, corresponding to smaller spatial scales. In contrast, the forced modes exhibit trajectories that seem continuous but nowhere differentiable stochastic processes, which is consistent with the presence of white-noise forcing.

\begin{figure}
\begin{center}
\begin{tabular}{cccc}
\includegraphics{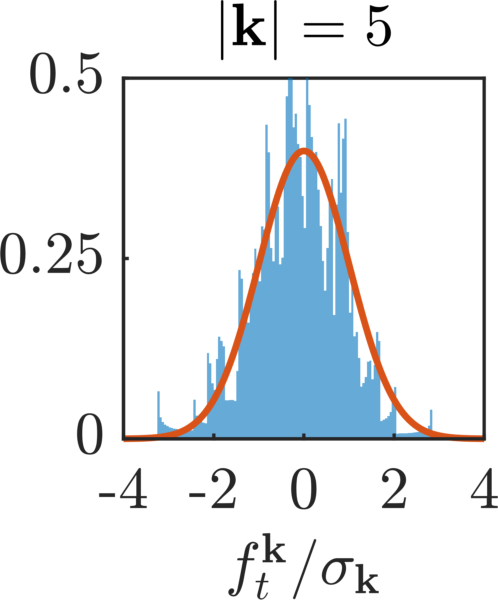} &
\includegraphics{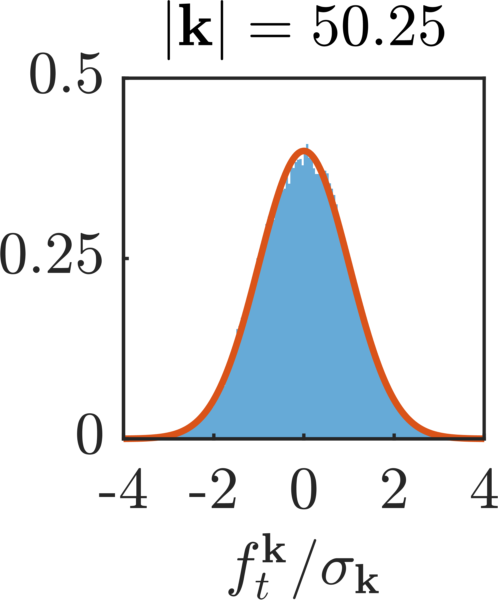} &
\includegraphics{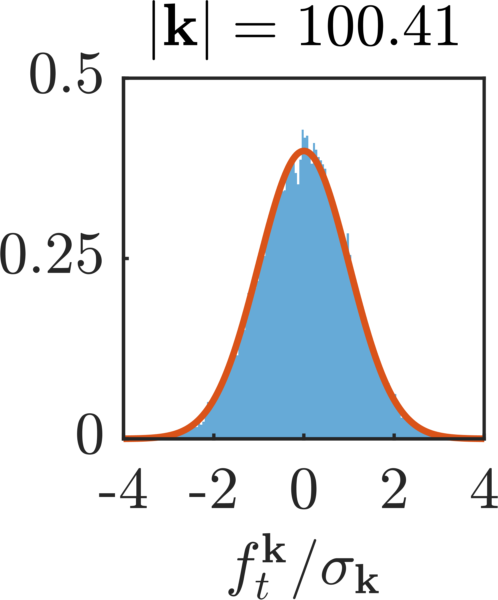} &
\includegraphics{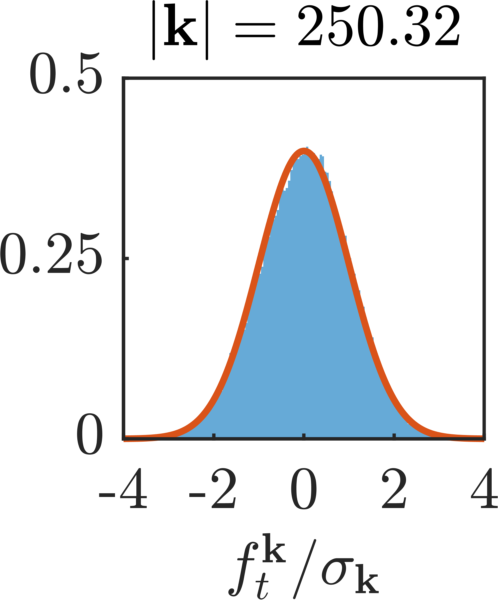} \\[0.5cm]
\includegraphics{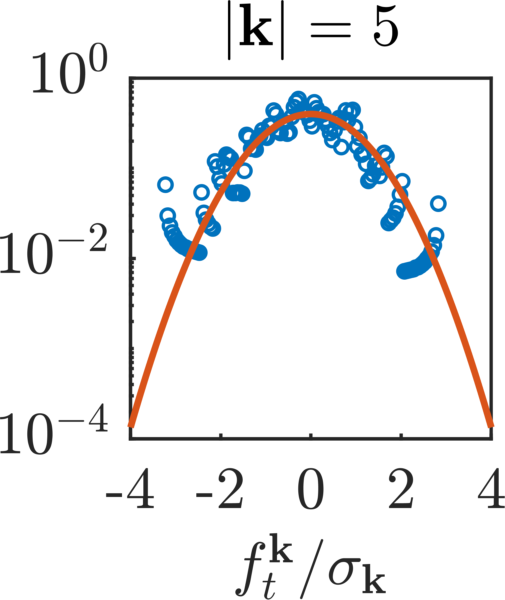} &
\includegraphics{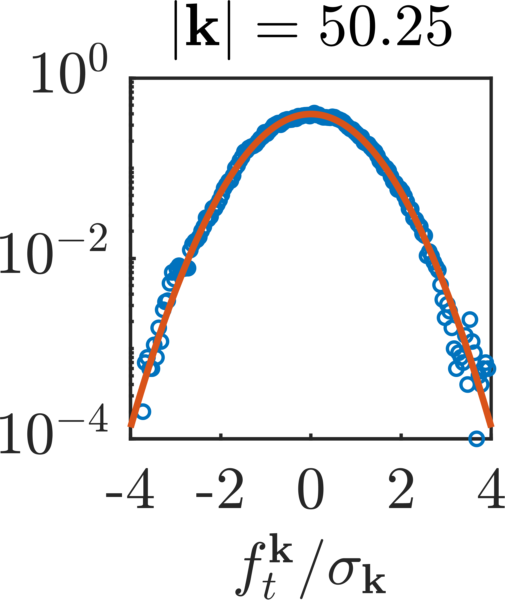} &
\includegraphics{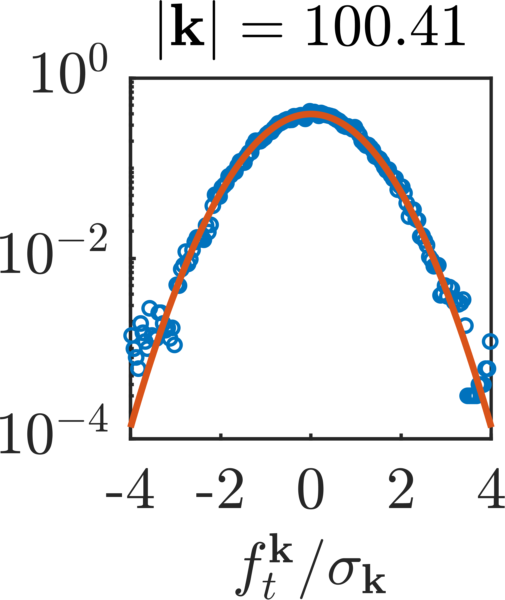} &
\includegraphics{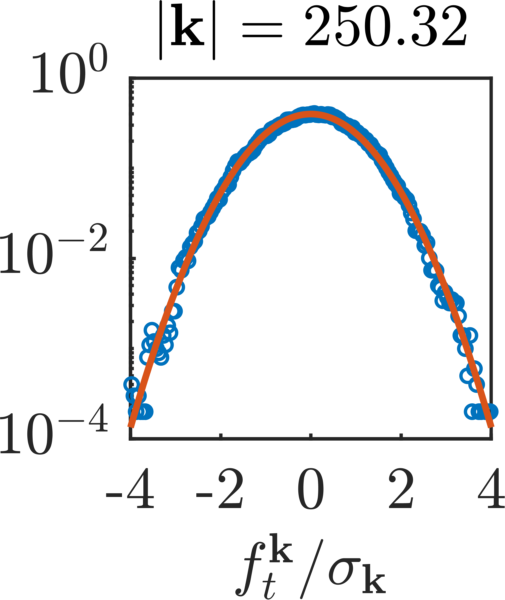}
\end{tabular}
\end{center}
\caption{Distribution of several Fourier modes with forcing applied at $\abs{\mathbf k} \sim k_f = 50$, compared with a Gaussian distribution $\mathcal N(0,\sigma_{\mathbf k}^2)$, shown in both linear (top) and semilogarithmic (bottom) scales.}
\label{fig:Fourier_invariant_distribution_k50}
\end{figure}

\begin{figure}
\begin{center}
\begin{tabular}{cccc}
\includegraphics{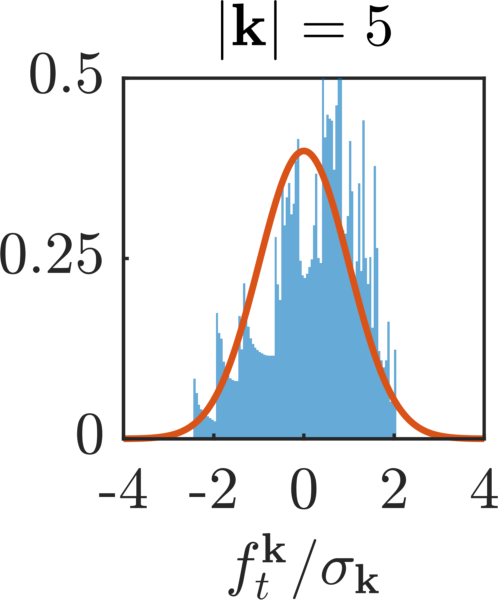} &
\includegraphics{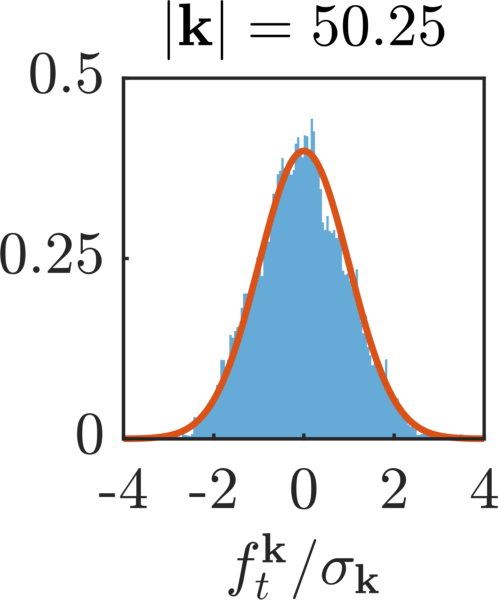} &
\includegraphics{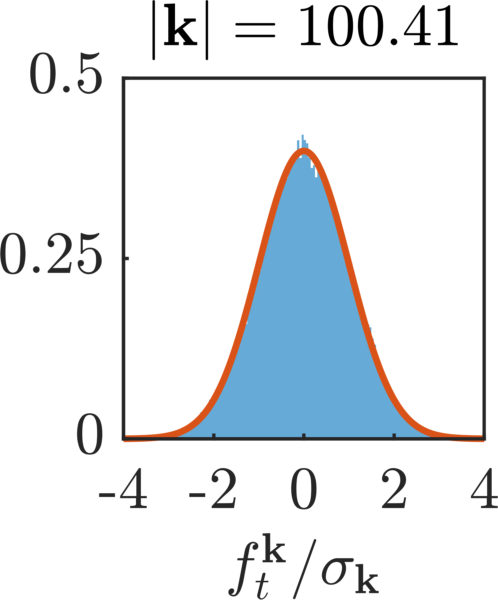} &
\includegraphics{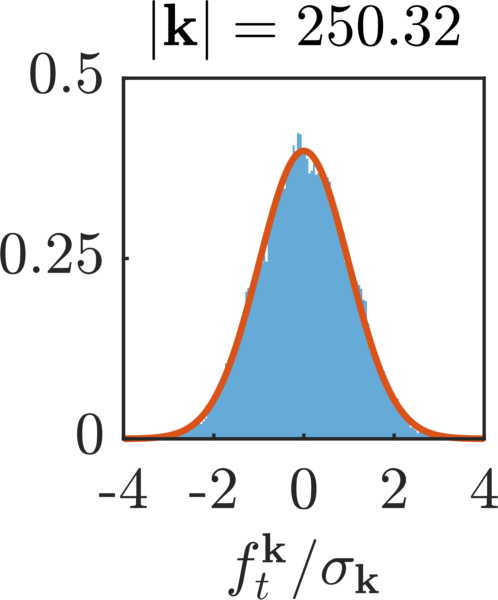}  \\[0.5cm]
\includegraphics{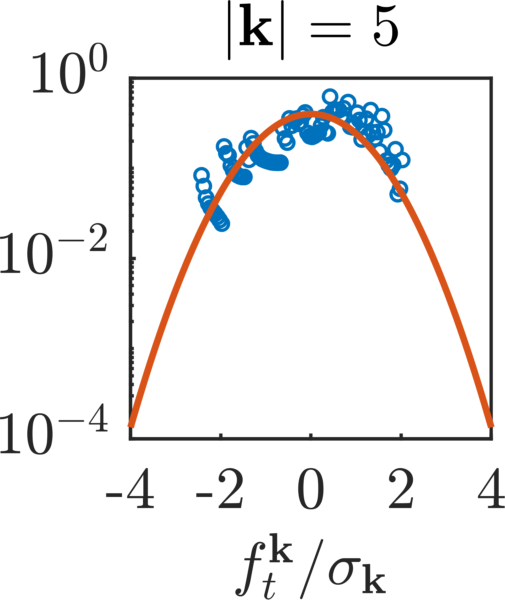} &
\includegraphics{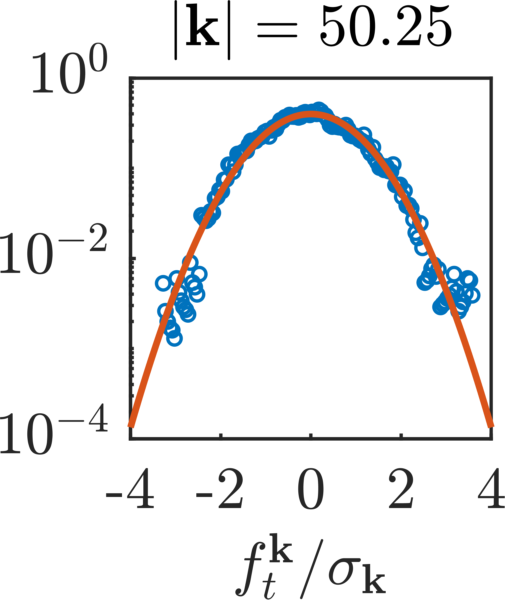} &
\includegraphics{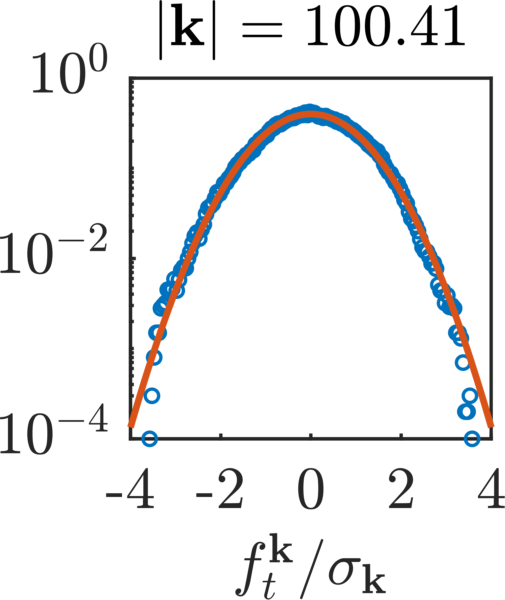} &
\includegraphics{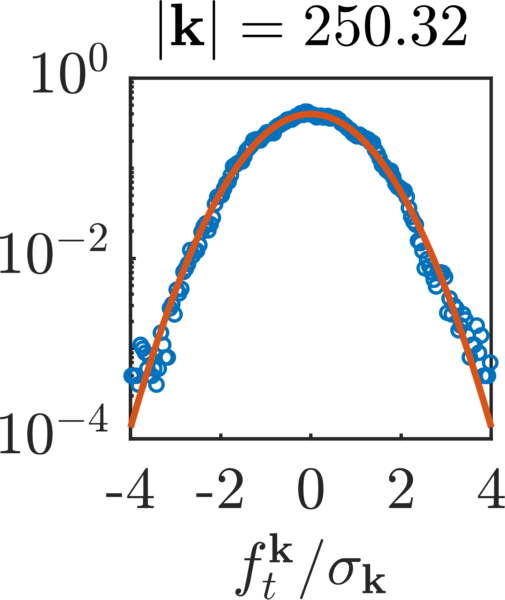}
\end{tabular}
\end{center}
\caption{Distribution of several Fourier modes with forcing applied at $\abs{\mathbf k} \sim k_f = 100$, compared with a Gaussian distribution $\mathcal N(0,\sigma_{\mathbf k}^2)$, shown in both linear (top) and semilogarithmic (bottom) scales.}
\label{fig:Fourier_invariant_distribution_k100}
\end{figure}

In \cref{fig:Fourier_invariant_distribution_k50,fig:Fourier_invariant_distribution_k100}, we also plot histograms of some of the Fourier modes in \eqref{eq:modes_used}. Except for the smallest wavenumber, which corresponds to the largest spatial scale, the empirical distributions match zero-mean Gaussian distributions whose variance is estimated from the data. This agreement support the use of Gaussian processes as a reasonable model for the Fourier modes.

\begin{figure}
\begin{center}
\includegraphics{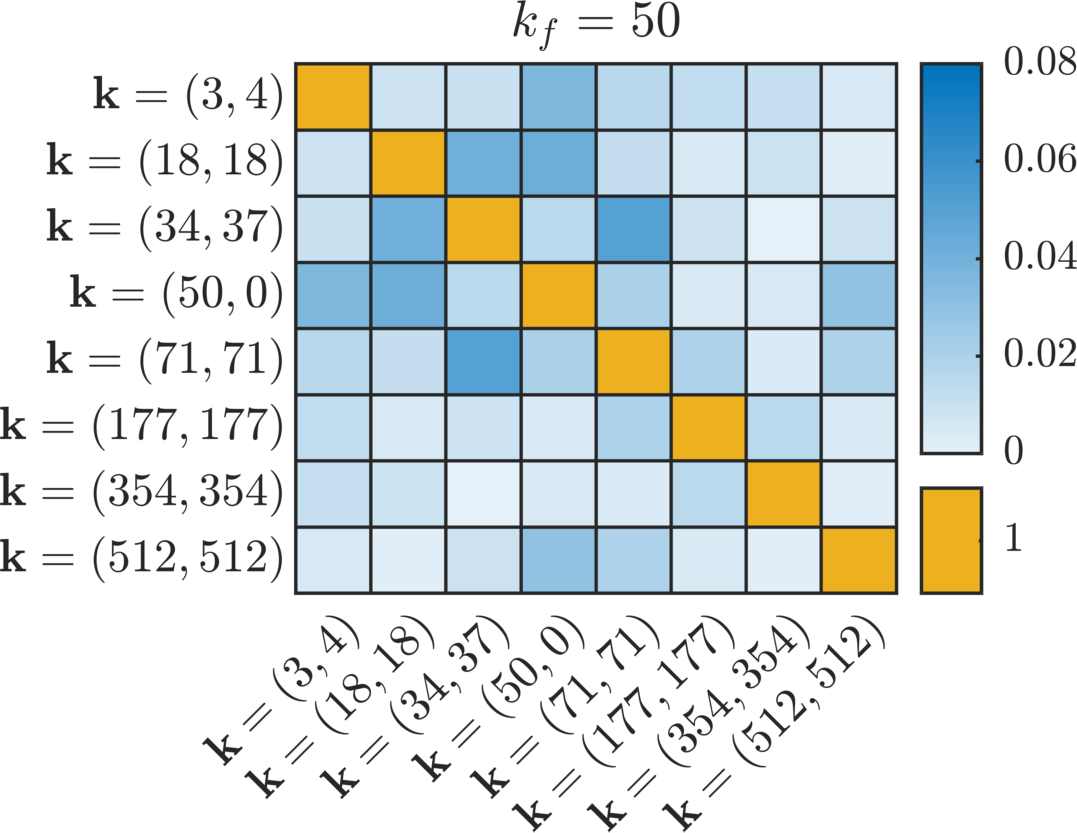} \hspace{0.5cm}
\includegraphics{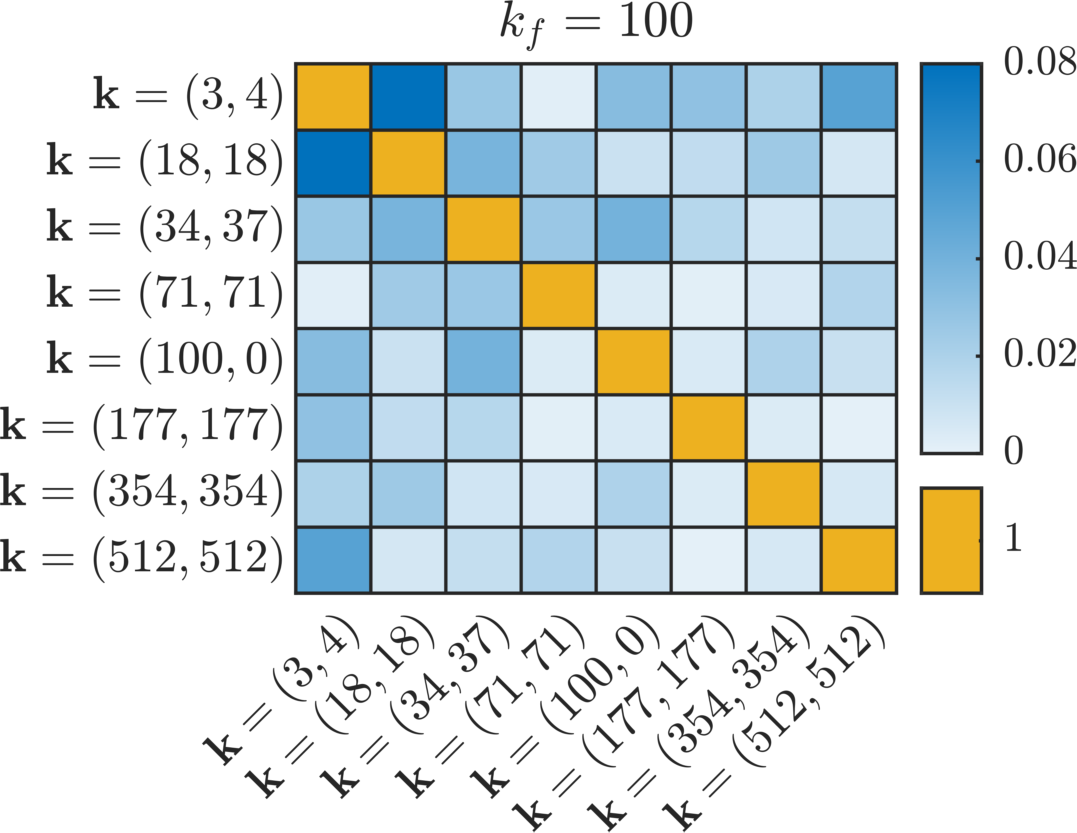}
\end{center}
\caption{Heatmap of the absolute value of the correlations between different Fourier modes, with forcing applied at $\abs{\mathbf k} \sim k_f = 50$ (left) and $\abs{\mathbf k} \sim k_f = 100$ (right).}
\label{fig:independence}
\end{figure}

We then investigate the statistical properties of the Fourier modes. We first compute the correlations between different Fourier modes and summarize the results in the heatmaps in \cref{fig:independence}. We observe that all absolute values of the correlations are smaller than $0.08$, suggesting that it may be reasonable to assume $\E [\xi_t^{\mathbf k} \xi_t^{\mathbf k'}] =0$ for $\mathbf{k}+\mathbf{k'} \ne 0$  in equation \eqref{eq:model_u}. We want to remark, at this point, that small correlation does not imply statistical or dynamical independence. Fourier modes are, in fact, coupled through triadic interactions. Modelling the Fourier coefficients as independent stochastic processes is therefore a convenient mathematical framework, which captures the covariance but distorts higher order statistics.

We now focus on the Fourier modes individually. In particular, we consider the autocorrelation function at stationarity
\begin{equation} \label{eq:Rk_Delta}
R_{\mathbf k}(\Delta) = \frac{\E[f_{t+\Delta}^{\mathbf k} f_t^{\mathbf k}]}{\E[(f_t^{\mathbf k})^2]},
\end{equation}
and the normalized variance of the increments
\begin{equation} \label{eq:Vk_Delta}
V_{\mathbf k}(\Delta) = \frac{\E[(f_{t+\Delta}^{\mathbf k} - f_t^{\mathbf k})^2]}{2 \E[(f_t^{\mathbf k})^2]} = 1 - R_{\mathbf k}(\Delta).
\end{equation}
We also define the relaxation time
\begin{equation} \label{eq:relaxation_time}
\tau_{\mathbf k} = 2 \int_0^\infty R_{\mathbf k}(\Delta) \dd \Delta,
\end{equation}
where the coefficient $2$ is justified in \cref{sec:heat_equation}. The autocorrelation function $R_{\mathbf k}(\Delta)$ is estimated from the data by temporal averages using the MATLAB function \texttt{autocorr} (Econometrics Toolbox), which provides a standard and numerically stable implementation.

\begin{figure}
\begin{center}
\includegraphics{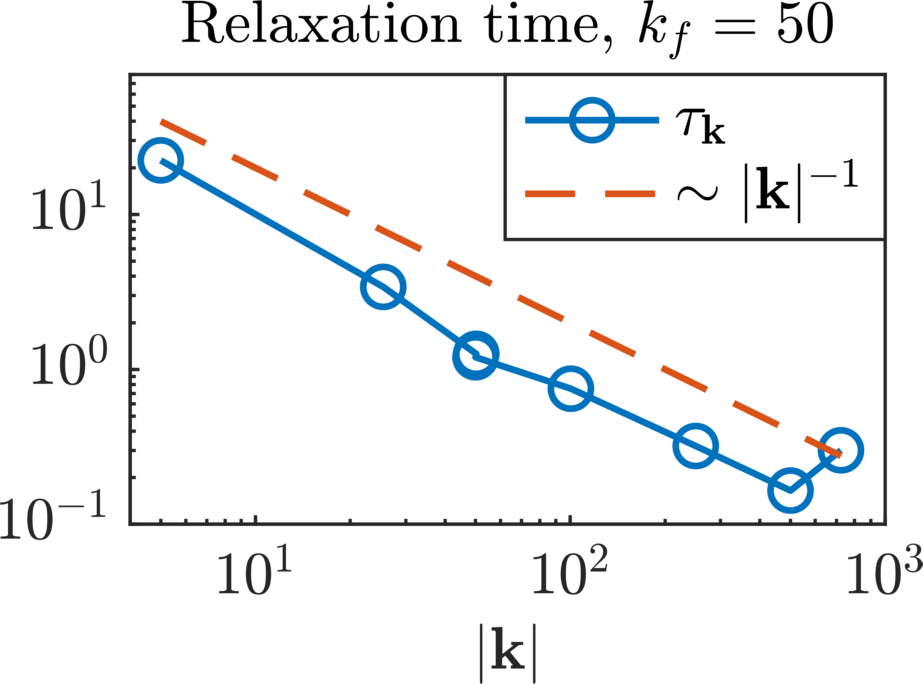} \hspace{0.5cm}
\includegraphics{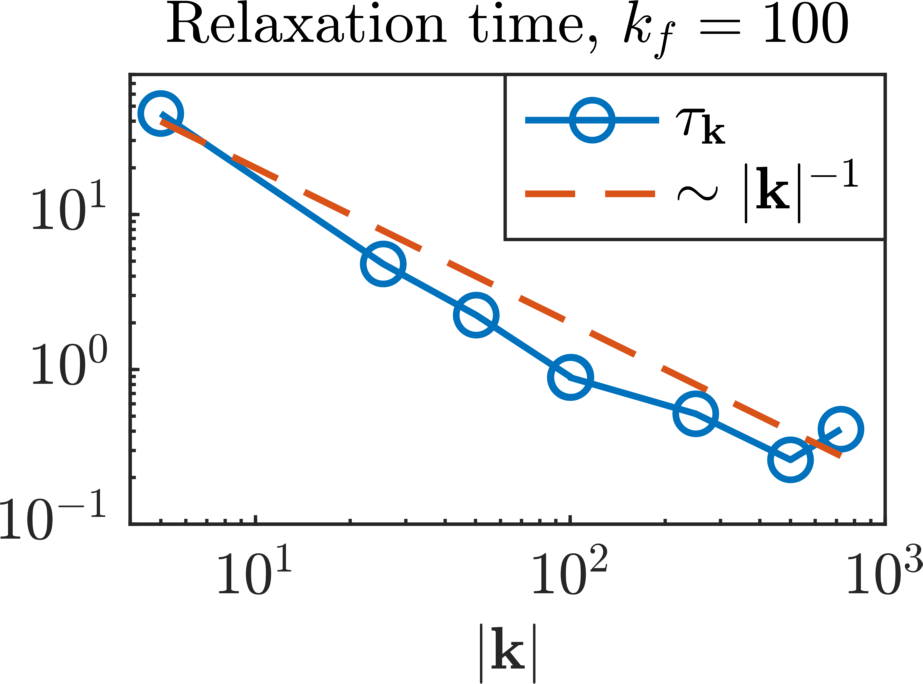}
\end{center}
\caption{Estimated relaxation time $\tau_{\mathbf k}$ for different Fourier modes, with forcing applied at $\abs{\mathbf k} \sim k_f = 50$ (left) and $\abs{\mathbf k} \sim k_f = 100$ (right).}
\label{fig:estimated_tau}
\end{figure}

The estimated relaxation times $\tau_{\mathbf k}$ are shown in \cref{fig:estimated_tau}, where $\tau_{\mathbf k}$ is plotted as a function of $\abs{\mathbf k}$. We observe an approximate inverse proportionality between the relaxation time and the Fourier wavenumber magnitude, i.e., $\tau_{\mathbf k} \propto \abs{\mathbf k}^{-1}$, except at the smallest scale. This scaling differs from the classical Kolmogorov prediction $\tau_k \sim k^{-2/3}$ associated with the inverse-cascade eddy-turnover time \cite[Chapter 6]{Pop00}. Nevertheless, we also note that a scaling of the form $\tau_k \sim k^{-1}$ has previously been reported, for example in \cite{PMC06}, where it is attributed to sweeping effects.

\begin{figure}
\begin{center}
\includegraphics{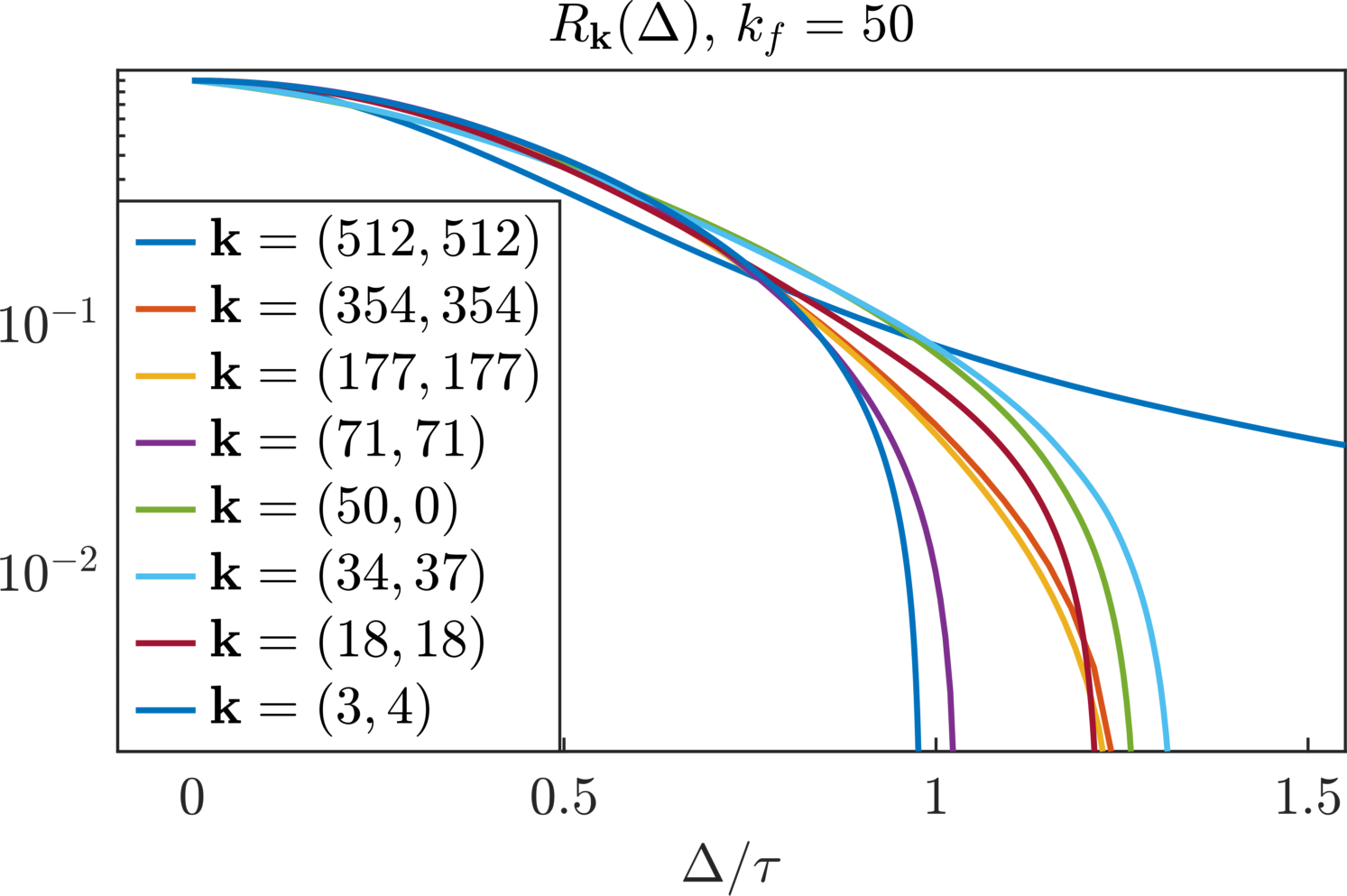} \\ \vspace{0.5cm}
\includegraphics{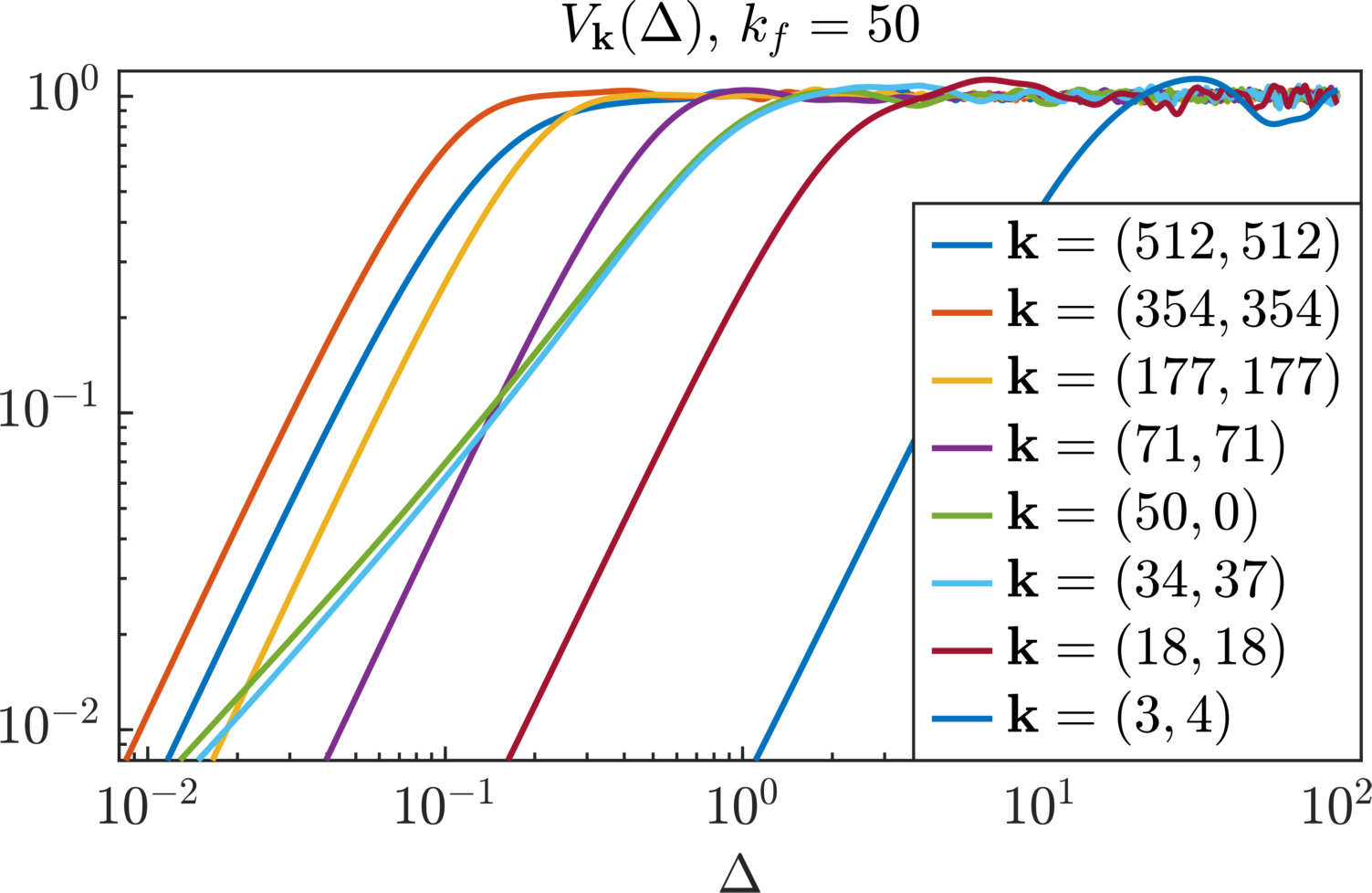}
\end{center}
\caption{Autocorrelation function $R_{\mathbf k}(\Delta)$ rescaled by the relaxation time $\tau_{\mathbf k}$ and normalized variance of the increments $V_{\mathbf k}(\Delta)$ for different Fourier modes, with forcing applied at $\abs{\mathbf k} \sim k_f = 50$.}
\label{fig:Fourier_correlation_k50}
\end{figure}

\begin{figure}
\begin{center}
\includegraphics{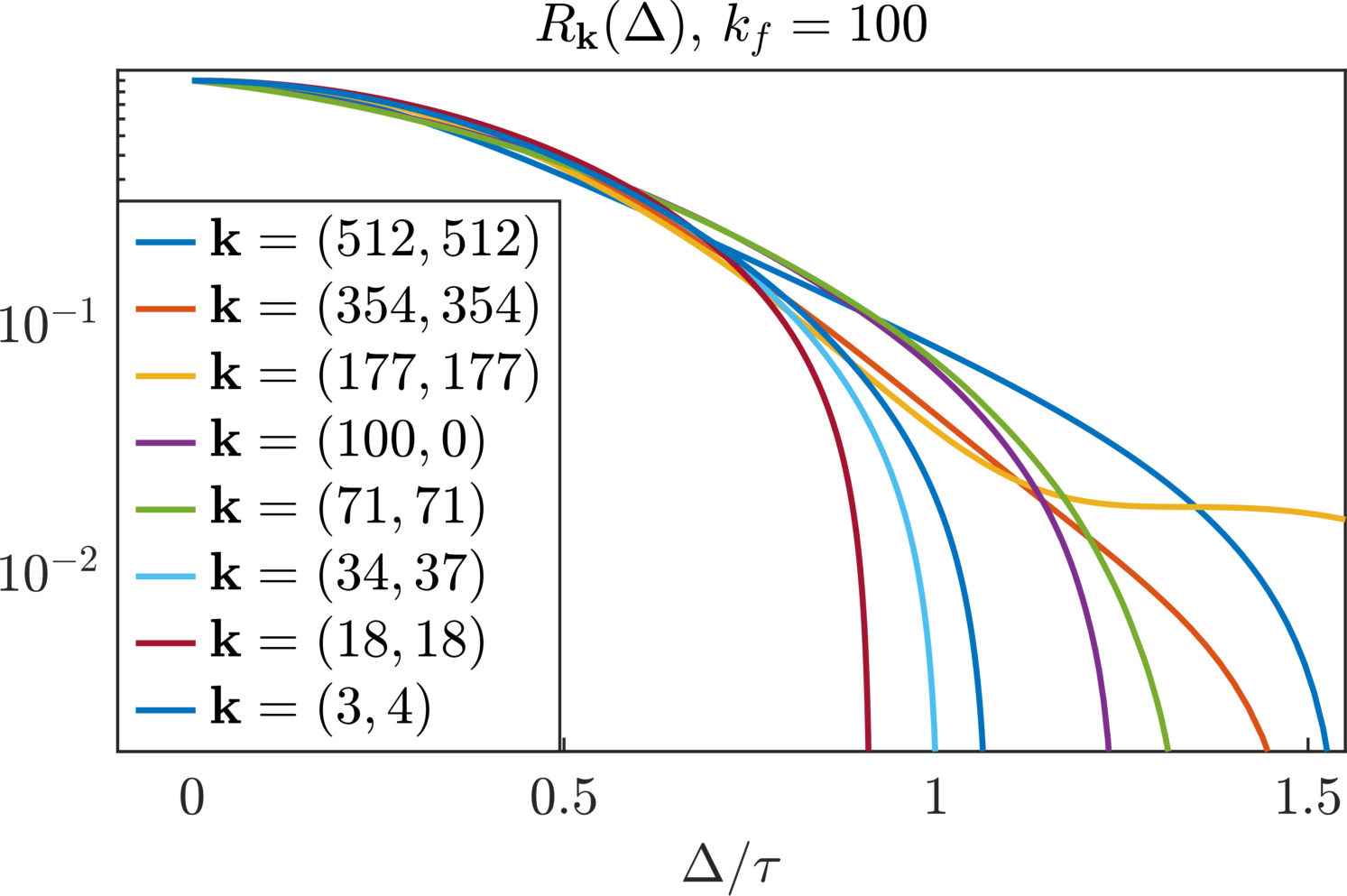} \\ \vspace{0.5cm}
\includegraphics{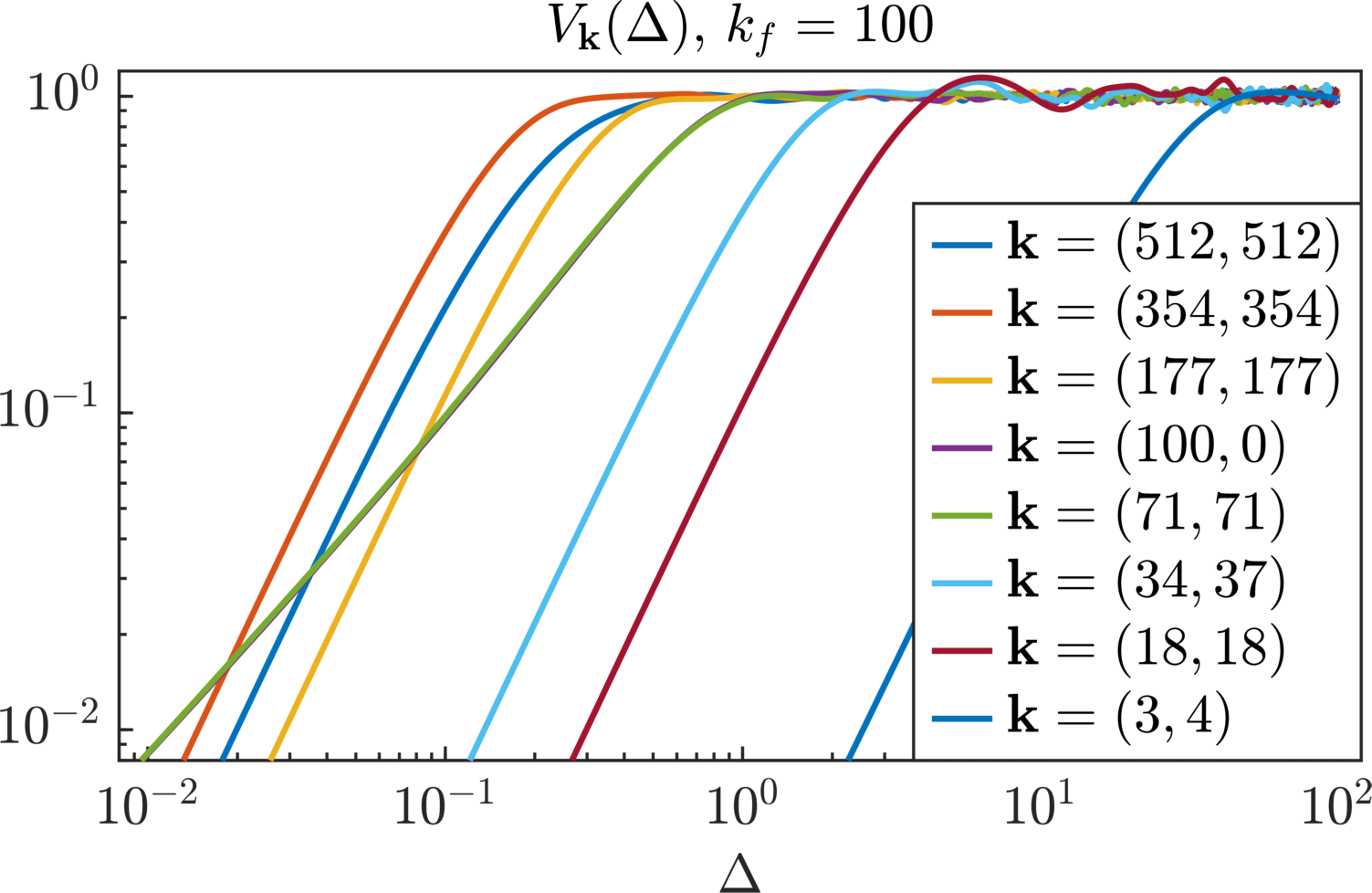}
\end{center}
\caption{Autocorrelation function $R_{\mathbf k}(\Delta)$ rescaled by the relaxation time $\tau_{\mathbf k}$ and normalized variance of the increments $V_{\mathbf k}(\Delta)$ for different Fourier modes, with forcing applied at $\abs{\mathbf k} \sim k_f = 100$.}
\label{fig:Fourier_correlation_k100}
\end{figure}

In \cref{fig:Fourier_correlation_k50,fig:Fourier_correlation_k100} we present the autocorrelation function \eqref{eq:Rk_Delta} in a semi-logarithmic plot, rescaled by the relaxation time $\tau_{\mathbf k}$, which allows for a direct comparison across the modes in \eqref{eq:modes_used}. Most modes have a similar behavior, with a decay faster than exponential. Moreover, we show the normalized variance \eqref{eq:Vk_Delta} on a logarithmic scale, and we observe two distinct behaviors. The forced modes exhibit an initial slope equal to one, consistent with the effect of the injected white noise. Conversely, the unforced modes show an initial slope equal to two. This can be explained by the fact that, although we construct stochastic models, the underlying partial differential equation and the numerical simulations are deterministic except for the forcing term. As a result, the Fourier modes are smooth in time, and the increment $f_{t+\Delta}^{\mathbf k} - f_t^{\mathbf k}$ scales linearly in $\Delta$ for small $\Delta$, which yields a quadratic scaling of $V_{\mathbf k}(\Delta)$. Nevertheless, this observation, together with the regularity of the trajectories discussed above, suggests that it may be appropriate to consider stochastic models with smoother sample paths than those of, e.g., the Ornstein–Uhlenbeck process, which is frequently used in turbulence modeling.

\begin{figure}
\begin{center}
\begin{tabular}{ccc}
\includegraphics{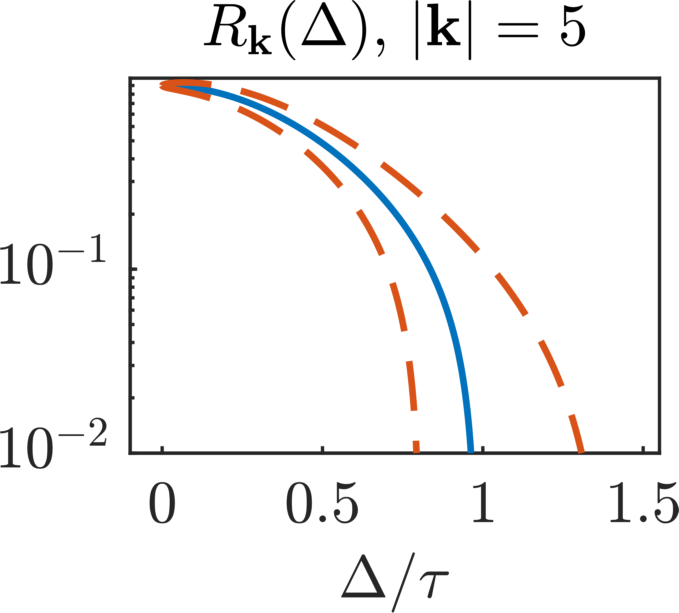} &
\includegraphics{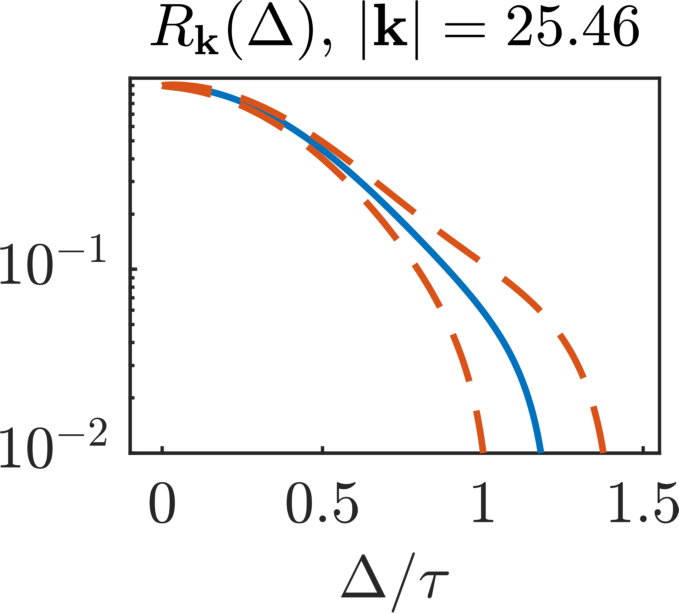} &
\includegraphics{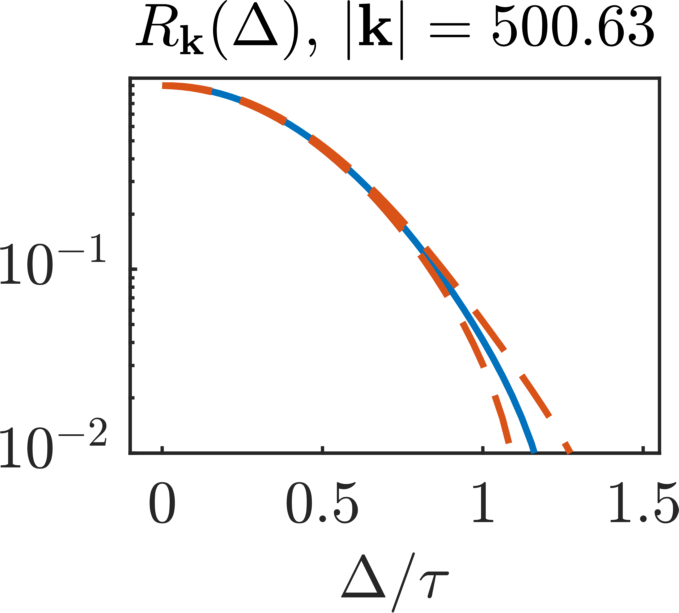} \\[5pt]
\includegraphics{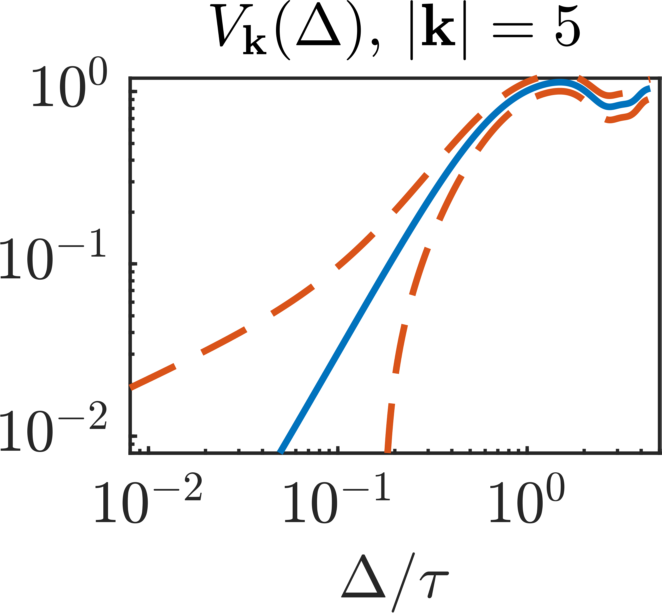} &
\includegraphics{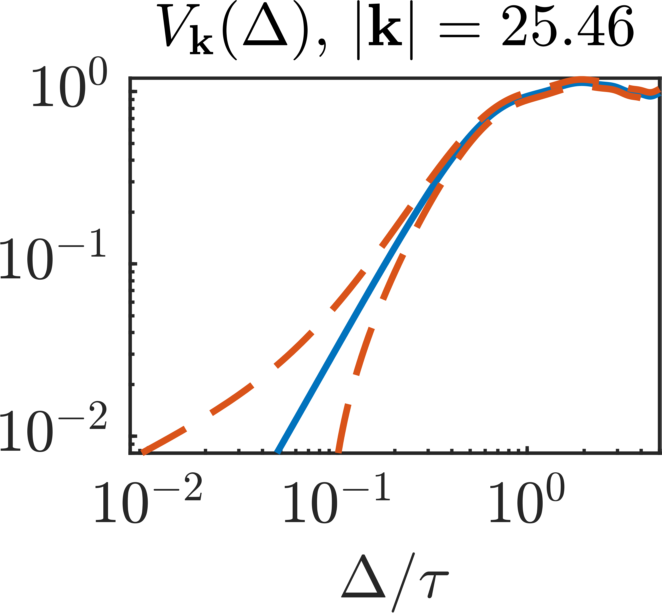} &
\includegraphics{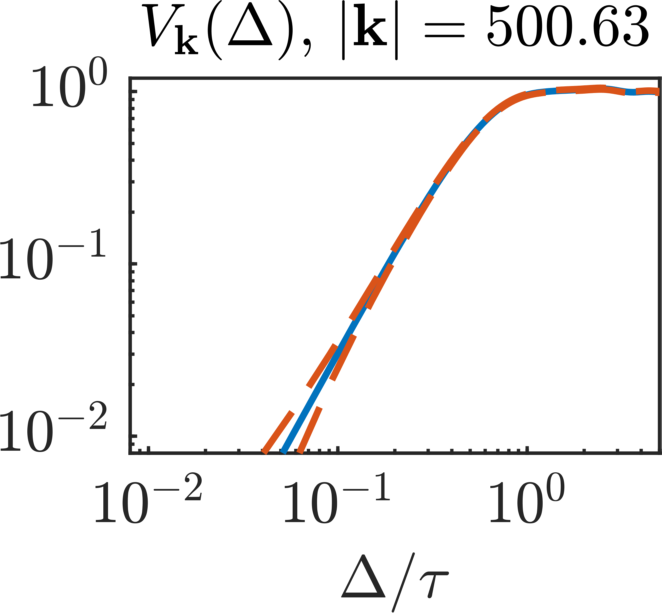}
\end{tabular}
\end{center}
\caption{Approximate 95\% confidence intervals for the autocorrelation function $R_{\mathbf k}(\Delta)$ and the normalized variance of the increments $V_{\mathbf k}(\Delta)$, both rescaled by the relaxation time $\tau_{\mathbf k}$, for different Fourier modes with forcing applied at $\abs{\mathbf k} \sim k_f = 50$.}
\label{fig:Fourier_confidence_k50}
\end{figure}

\begin{figure}
\begin{center}
\begin{tabular}{ccc}
\includegraphics{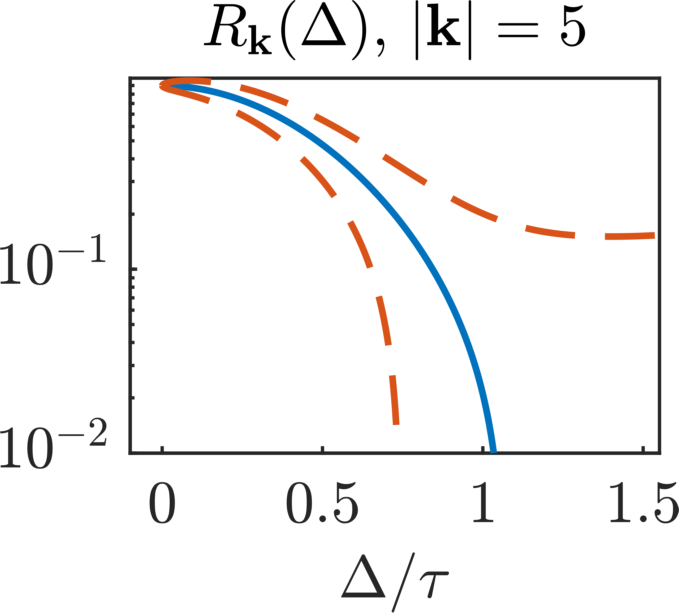} &
\includegraphics{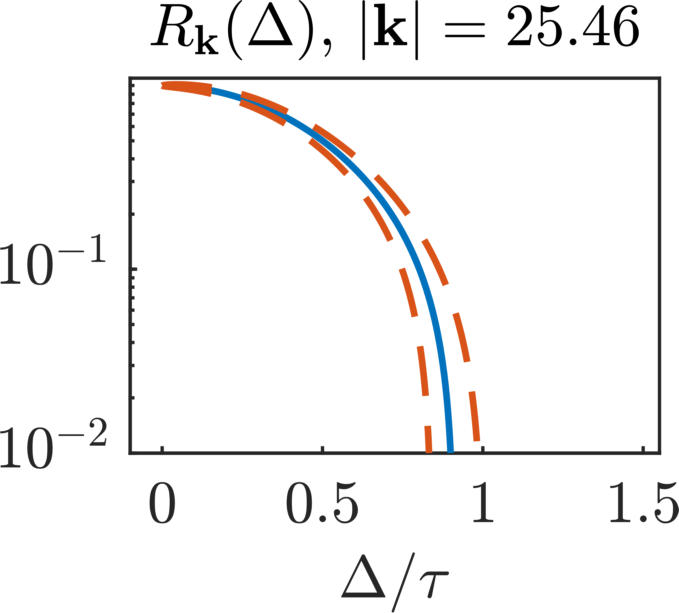} &
\includegraphics{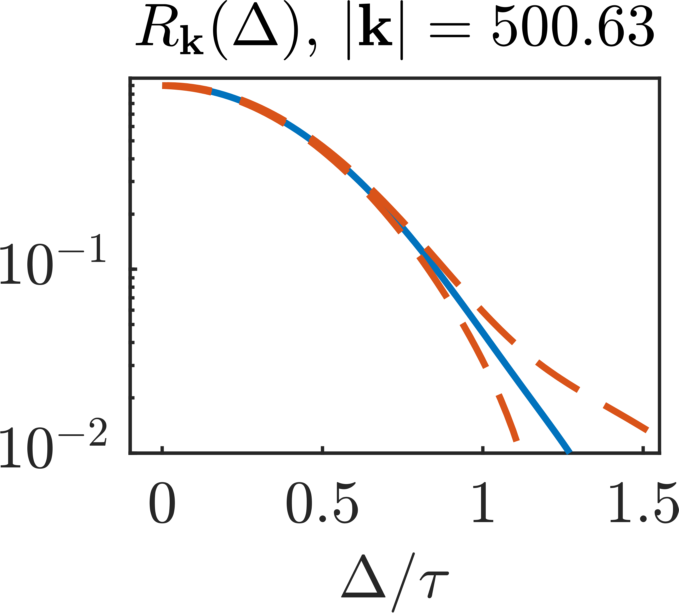} \\[5pt]
\includegraphics{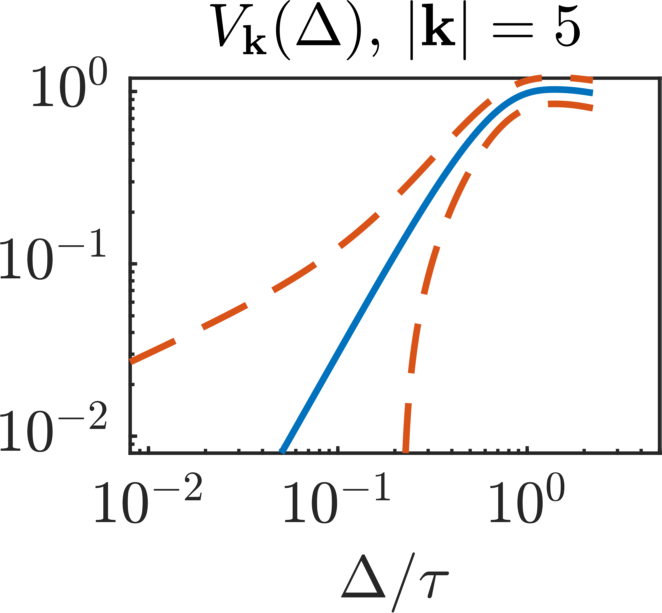} &
\includegraphics{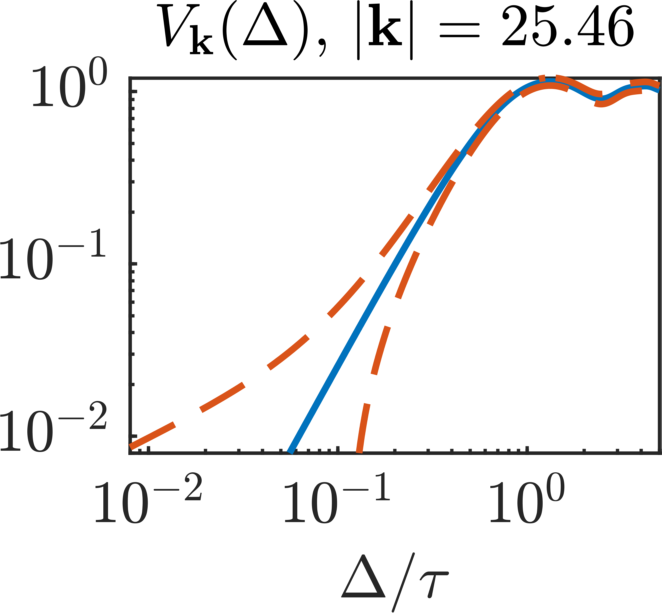} &
\includegraphics{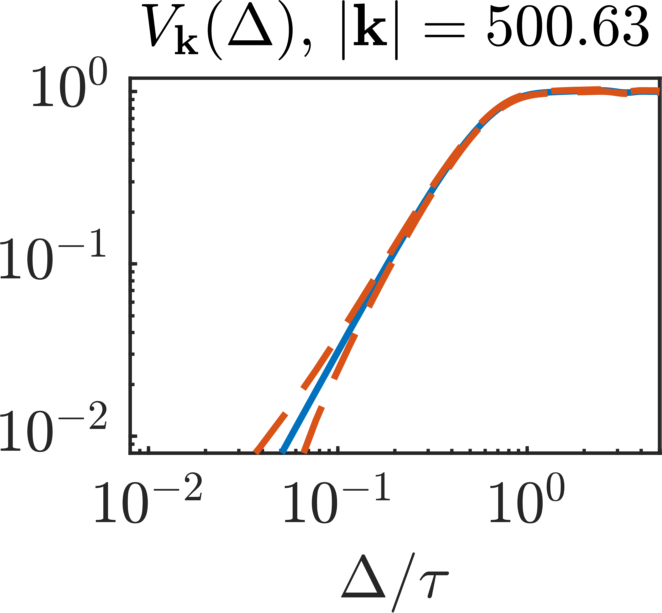}
\end{tabular}
\end{center}
\caption{Approximate 95\% confidence intervals for the autocorrelation function $R_{\mathbf k}(\Delta)$ and the normalized variance of the increments $V_{\mathbf k}(\Delta)$, both rescaled by the relaxation time $\tau_{\mathbf k}$, for different Fourier modes with forcing applied at $\abs{\mathbf k} \sim k_f = 100$.}
\label{fig:Fourier_confidence_k100}
\end{figure}

Finally, we study the reliability of the plots presented so far and derive confidence intervals for the autocorrelation function $R_{\mathbf k}(\Delta)$. We use Bartlett's formula to approximate the confidence intervals \cite{BoJ76}. In particular, let $\delta$ be the discretization step for the Fourier modes and let $N = T/\delta$ denote the number of samples from the trajectory. Then, an approximate confidence interval with confidence level $\ell$ for $R_{\mathbf k}(n\delta)$, $n = 1, \dots, N$, is
\begin{equation} \label{eq:confidence_Bartlett}
R_{\mathbf k}(n\delta) \pm z_{\frac{1+\ell}2} \sqrt{\frac1N \left( 1 + 2 \sum_{j=1}^{n-1} R_{\mathbf k}(j\delta)^2 \right)},
\end{equation}
where $z_{\frac{1+\ell}2}$ is the quantile of order $(1+\ell)/2$ of the standard Gaussian distribution. We report in \cref{fig:Fourier_confidence_k50,fig:Fourier_confidence_k100} the functions $R_{\mathbf k}(\Delta)$ and $V_{\mathbf k}(\Delta)$ for a subset of the wavenumbers in \eqref{eq:modes_used}, together with the approximate confidence intervals computed using formula \eqref{eq:confidence_Bartlett}, where we set $\ell =  0.95$. We observe that smaller wavenumbers are associated with wider confidence intervals, whereas larger wavenumbers yield more accurate estimates of the autocorrelation function. This is due to the longer relaxation time $\tau_{\mathbf k}$ for smaller wavenumbers, which implies that fewer essentially uncorrelated samples are available for the estimation of the autocorrelation function. Nevertheless, we observe that as long as $R_{\mathbf k}(\Delta)$ and $V_{\mathbf k}(\Delta)$ are larger than $\sim 0.1$, below which approximation errors begin to dominate, the measurements remain sufficiently accurate to capture the relevant trends, especially for intermediate and large wavenumbers.

\section{Model for the single processes} \label{sec:model_process}

In this section, we derive a stochastic model for the process underlying the Fourier components, with the aim of reproducing the statistical properties observed in the numerical data in \cref{sec:numerics_Fourier}. Rather than considering solutions of stochastic differential equations, such as Ornstein--Uhlenbeck processes, we focus on processes defined via convolutions of white noise with suitable kernels. This approach has the advantage that the choice of kernel directly controls both the regularity and the autocorrelation structure of the resulting process.

We first introduce the conditions that a kernel must satisfy in order to provide a meaningful model for Fourier components of turbulent flows. We then derive an explicit expression for the covariance function and show that it satisfies the required normalization properties. Finally, we present examples of admissible kernels and illustrate their flexibility in capturing different correlation structures at both short and long time scales.

Let $\theta \colon \mathbb{R} \to \mathbb{R}$ be a function that plays the role of a mollifier, and define the stationary Gaussian process
\begin{equation} \label{eq:def_xi}
\xi_t = \int_{-\infty}^{+\infty} \frac1{\sqrt\tau} \theta \left( \frac{t-s}\tau \right) \dd W_s,
\end{equation}
where $W_t$ is a two-sided Brownian motion, and $\tau > 0$ is a relaxation time. The kernel $\theta$ satisfies the conditions below, and the regularity of the process $\xi_t$ depends on the regularity of $\theta$ itself.

\begin{assumption} \label{as:condition_theta}
The function $\theta \in L^1(\R) \cap L^2(\R)$ satisfies:
\begin{enumerate}[label=(\roman*)]
\item $\theta(x) \ge 0$ for a.e. $x \in \R$;
\item $\theta(-x) = \theta(x)$ for a.e. $x \in \R$;
\item $\norm{\theta}_{L^1(\R)} = 1$;
\item $\norm{\theta}_{L^2(\R)} = 1$;
\item $M \defeq \int_\R \abs{x} \theta(x) \dd x < \infty$.
\end{enumerate}
\end{assumption}

\begin{remark}
A kernel $\theta$ satisfying \cref{as:condition_theta} can be interpreted as a probability density function of a symmetric random variable with finite first moment, together with the additional normalization $\norm{\theta}_{L^2(\mathbb{R})} = 1$. This condition can always be enforced by a suitable rescaling of a given density. Indeed, let $\widetilde{\theta}$ be a probability density function, then define
\begin{equation}
\theta(x) = \frac1A \widetilde\theta \left( \frac{x}A \right), \qquad \text{where} \qquad A = \int_{-\infty}^{+\infty} \widetilde\theta(x)^2 \dd x,
\end{equation}
which satisfies $\norm{\theta}_{L^2(\mathbb{R})} = 1$.
\end{remark}

\subsection{Autocovariance function}

The autocovariance function of $\xi_t$ depends only on the lag between two time points, as the following lemma shows.

\begin{lemma} \label{lem:autocovariance}
Let $\xi_t$ be defined in equation \eqref{eq:def_xi} and $\theta$ satisfy \cref{as:condition_theta}. Then, for all $t \ge 0$ and $\Delta \ge 0$ we have
\begin{equation}
\E[\xi_{t+\Delta} \xi_t] = C \left( \frac\Delta\tau \right),
\end{equation}
where
\begin{equation} \label{eq:def_autocovariance}
C(u) \defeq \theta^{(2)}(u) \defeq (\theta * \theta)(u) = \int_{-\infty}^{+\infty} \theta(u - r) \theta(r) \dd r.
\end{equation}
\end{lemma}
\begin{proof}
By definition of $\xi_t$ and applying Itô isometry, we have
\begin{equation}
\begin{aligned}
\E[\xi_{t+\Delta} \xi_t] &= \E \left[ \int_{-\infty}^{+\infty} \frac1{\sqrt\tau} \theta \left( \frac{t+\Delta-s}\tau \right) \dd W_s \int_{-\infty}^{+\infty} \frac1{\sqrt\tau} \theta \left( \frac{t-s}\tau \right) \dd W_s \right] \\
&= \int_{-\infty}^{+\infty} \frac1\tau \theta \left( \frac{t+\Delta-s}\tau \right) \theta \left( \frac{t-s}\tau \right) \dd s.
\end{aligned}
\end{equation}
Using the change of variables $r = -(t-s)/\tau$ and the symmetry of $\theta$, we obtain
\begin{equation}
\E[\xi_{t+\Delta} \xi_t] = \int_{-\infty}^{+\infty} \theta \left( \frac\Delta\tau - r \right) \theta(r) \dd r,
\end{equation}
which is the desired result.
\end{proof}

Therefore, by \cref{lem:autocovariance}, the variance of the increments of $\xi_t$ is given by
\begin{equation}
\E \left[ (\xi_{t+\Delta} - \xi_t)^2 \right] = 2 \left[ C(0) - C \left( \frac\Delta\tau \right) \right].
\end{equation}
In the next result, we show that \cref{as:condition_theta} ensures two normalization conditions that will be important later. Specifically, the first condition implies that the process $\xi_t$ has unit variance, while the second guarantees that the equation derived from this model is consistent with existing results in turbulence modeling.

\begin{lemma} \label{lem:normalizations}
The autocovariance function defined in equation \eqref{eq:def_autocovariance} satisfies:
\begin{enumerate}[label=(\roman*)]
\item $C(0) = 1$;
\item $\int_0^\infty C(u) \dd u = \frac12$.
\end{enumerate}
\end{lemma}
\begin{proof}
For point $(i)$, by \cref{as:condition_theta}, we have
\begin{equation}
C(0) = \int_{-\infty}^{+\infty} \theta(-r) \theta(r) \dd r = \int_{-\infty}^{+\infty} \theta(r)^2 \dd r = \norm{\theta}_{L^2(\R)}^2 = 1.
\end{equation}
For point (ii), we first note that $C$ is even. Indeed,
\begin{equation}
C(-u) = \int_{-\infty}^{+\infty} \theta(-u-r) \theta(r) \dd r = \int_{-\infty}^{+\infty} \theta(u+r) \theta(-r) \dd r,
\end{equation}
which, by the change of variables $s = -r$, implies
\begin{equation}
C(-u) = \int_{-\infty}^{+\infty} \theta(u-s) \theta(s) \dd s = C(u).
\end{equation}
Therefore, we obtain
\begin{equation}
\int_0^\infty C(u) \dd u = \frac12 \int_{-\infty}^{+\infty} C(u) \dd u = \frac12,
\end{equation}
since
\begin{equation}
\int_{-\infty}^{+\infty} C(u) \dd u = \int_{-\infty}^{+\infty} \int_{-\infty}^{+\infty} \theta(u - r) \theta(r) \dd r \dd u = \norm{\theta}_{L^1(\R)}^2 = 1,
\end{equation}
and this concludes the proof.
\end{proof}

Hence, \cref{lem:normalizations} also gives
\begin{equation}
\int_0^\infty C \left( \frac\Delta\tau \right) \dd\Delta = \frac\tau2,
\end{equation}
which is the reason for the coefficient $2$ in equation \eqref{eq:relaxation_time}.

\begin{example} \label{ex:Matern_Gaussian}
A kernel $\theta$ satisfying \cref{as:condition_theta} is
\begin{equation} \label{eq:Gaussian_kernel}
\theta(x) = \sqrt2 e^{-2\pi x^2}.
\end{equation}
In this case, the process $\xi_t$ is smooth, and its autocovariance function can be computed explicitly
\begin{equation} \label{eq:Gaussian_covariance}
C(u) = e^{-\pi u^2},
\end{equation}
which exhibits a fast decay of memory. Other possible choices for $\theta$ can be derived from Matérn-type kernels, originally introduced in \cite{Mat60}. In particular, let $\beta \ge 0$ be a smoothness parameter, and define
\begin{equation} \label{eq:Matern_kernel}
\theta_\beta(x) = \frac{2^{\beta+1}}{\sqrt\pi} \frac{\Gamma(\beta+1)}{\Gamma(2\beta+1/2)} \left( \gamma_\beta \abs{x} \right)^\beta K_\beta \left( \gamma_\beta\abs{x} \right),
\end{equation}
where $\Gamma$ denotes the Gamma function, $K_\beta$ is the modified Bessel function of the second kind, and
\begin{equation}
\gamma_\beta = 2^{2\beta+1} \frac{\Gamma(\beta+1/2) \Gamma(\beta+1)}{\Gamma(2\beta+1/2)}.
\end{equation}
The corresponding autocovariance function becomes
\begin{equation} \label{eq:Matern_covariance}
C_\beta(u) = \sqrt{\frac2\pi} \frac{\Gamma(\beta+1/2) \Gamma(\beta+1)}{\Gamma(2\beta+1/2) \Gamma(2\beta+1)} \left( \gamma_\beta \abs{u} \right)^{2\beta+1/2} K_{2\beta+1/2} \left( \gamma_\beta\abs{u} \right).
\end{equation}
We note that as $\beta \to \infty$, the kernel $\theta_\beta$ in \eqref{eq:Matern_kernel} approaches the Gaussian kernel $\theta$ in \eqref{eq:Gaussian_kernel}, and consequently the covariance $C_\beta$ in \eqref{eq:Matern_covariance} converges to $C$ in \eqref{eq:Gaussian_covariance}. Moreover, when $\beta = 0$, we recover the exponential autocovariance function
\begin{equation}
C_0(x) = e^{-2\abs{u}},
\end{equation}
which is characteristic of Ornstein--Uhlenbeck processes.
\end{example}

\subsection{Numerical validation}

We focus on the model described in \cref{ex:Matern_Gaussian}. We first illustrate its properties numerically and then compare the results with the fluid data from \cref{sec:numerics_Fourier}.

\begin{figure}
\begin{center}
\begin{tabular}{ccc}
\includegraphics{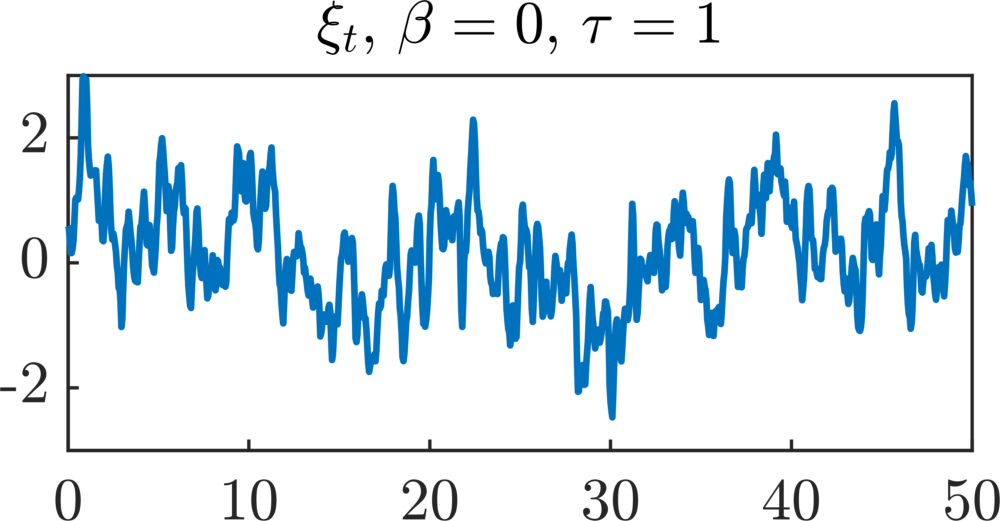} &&
\includegraphics{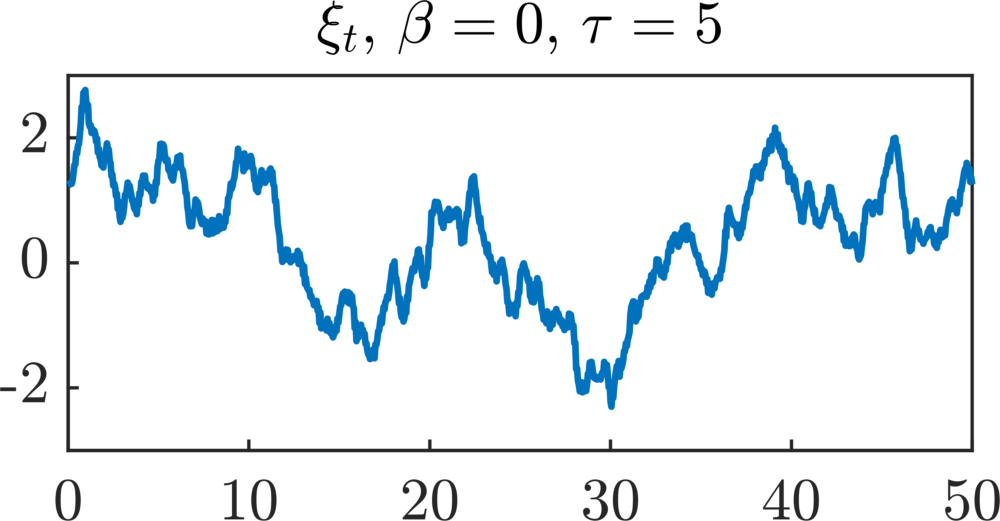} \\[5pt]
\includegraphics{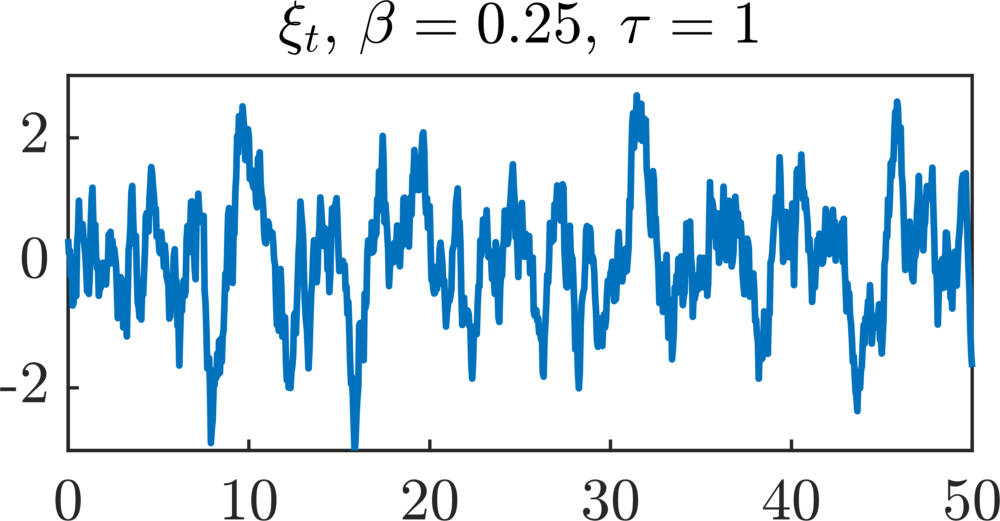} &&
\includegraphics{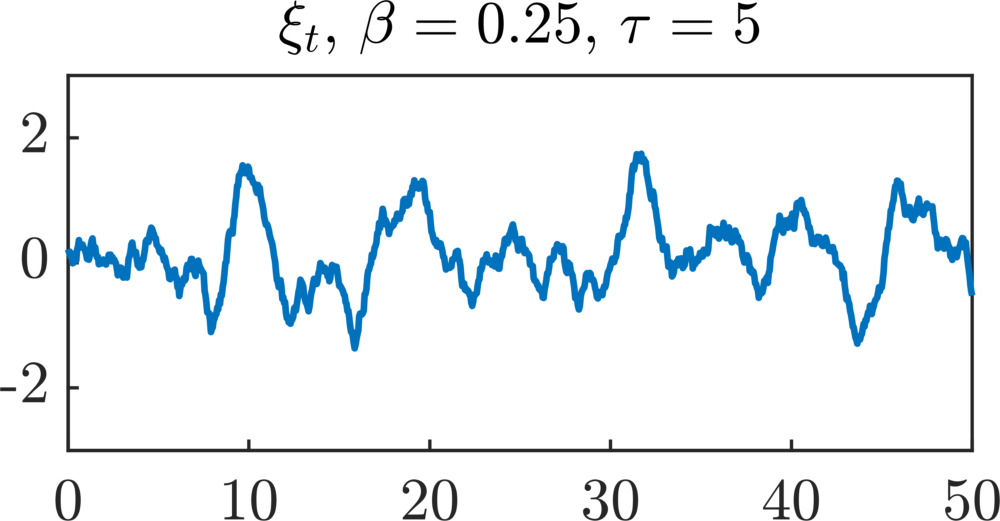} \\[5pt]
\includegraphics{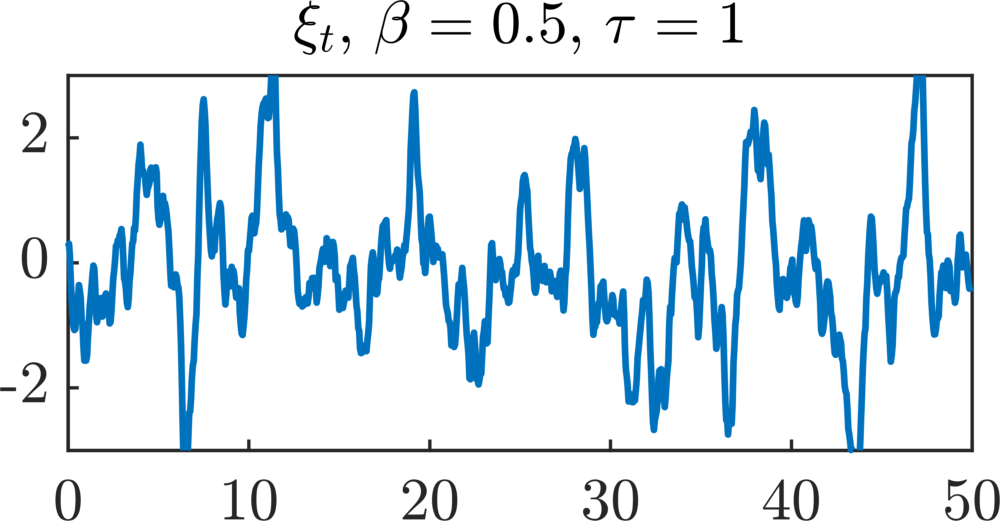} &&
\includegraphics{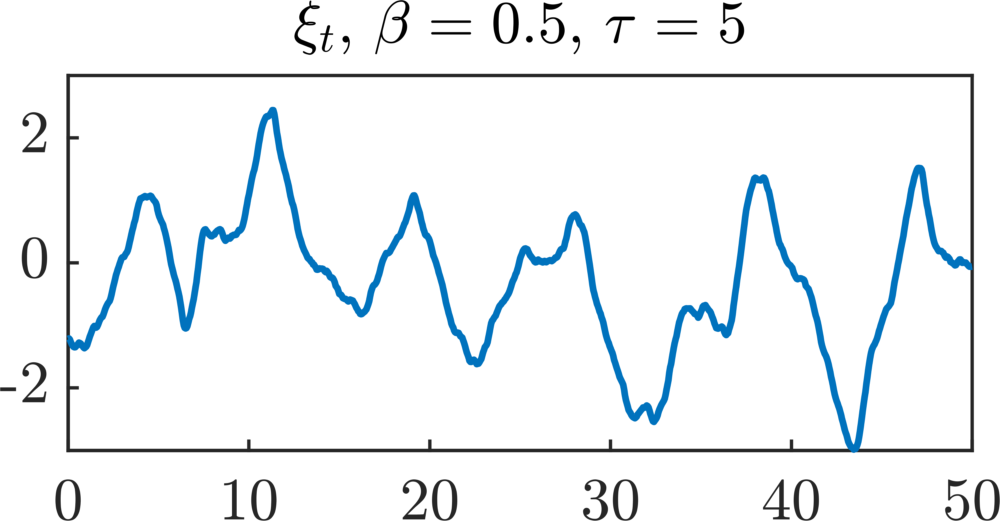} \\[5pt]
\includegraphics{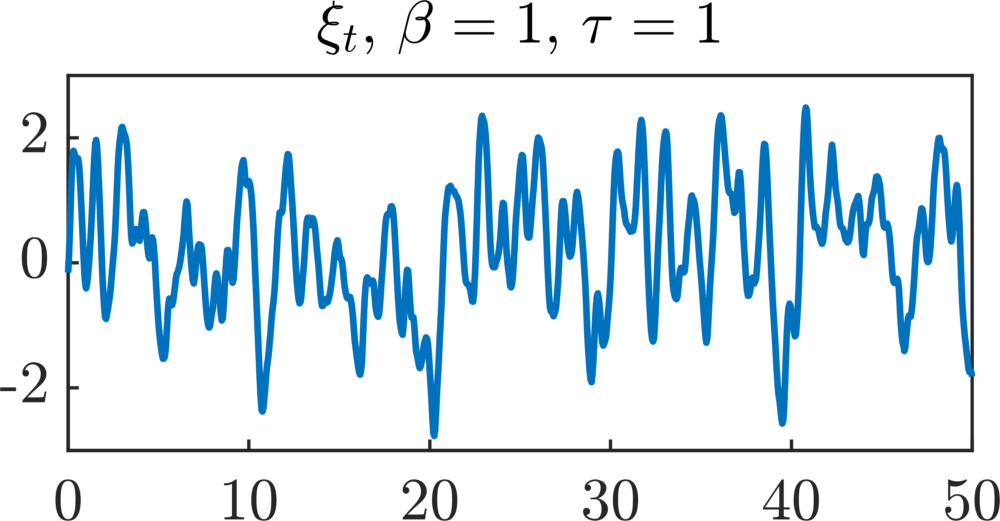} &&
\includegraphics{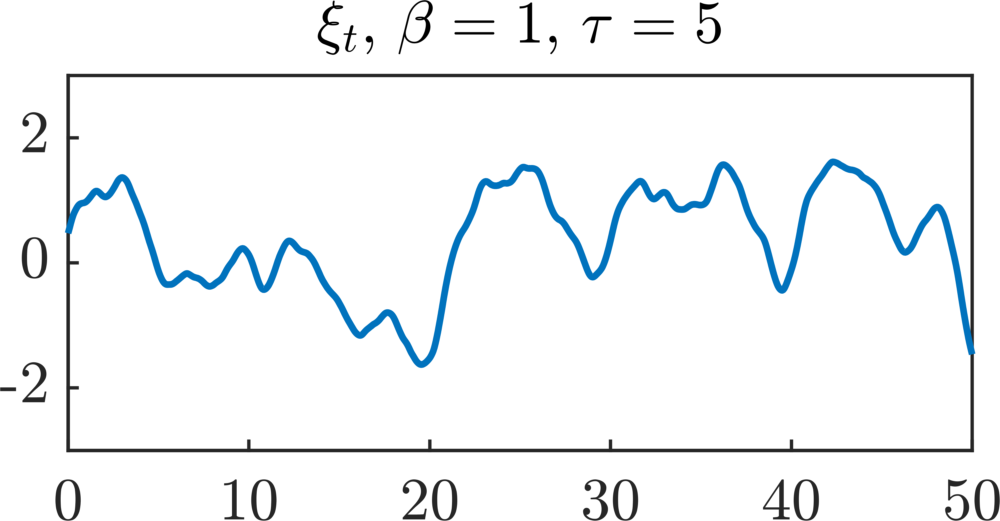} \\[5pt]
\includegraphics{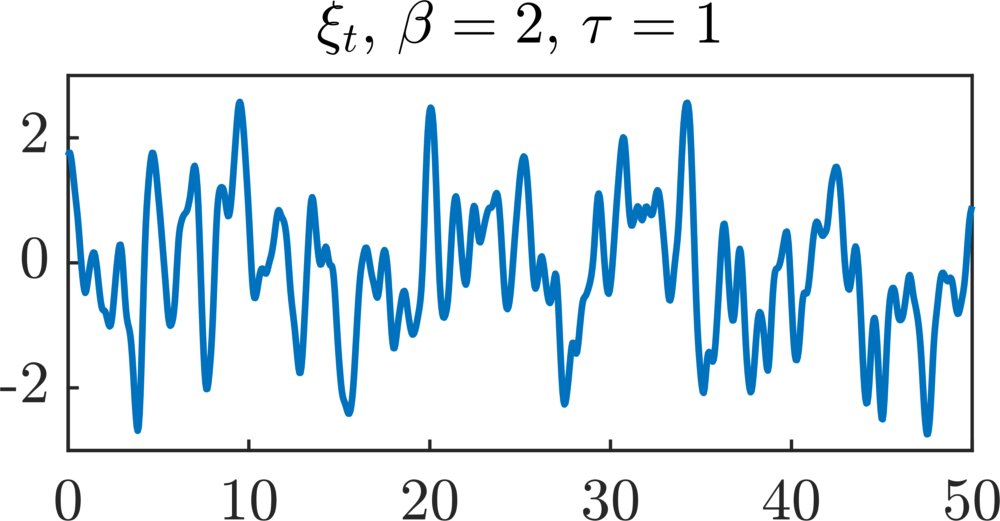} &&
\includegraphics{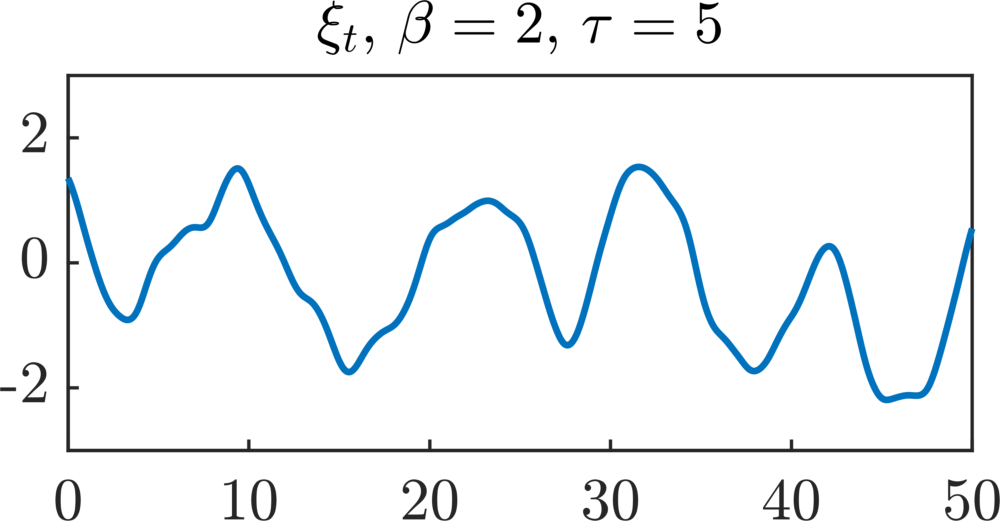} \\[5pt]
\includegraphics{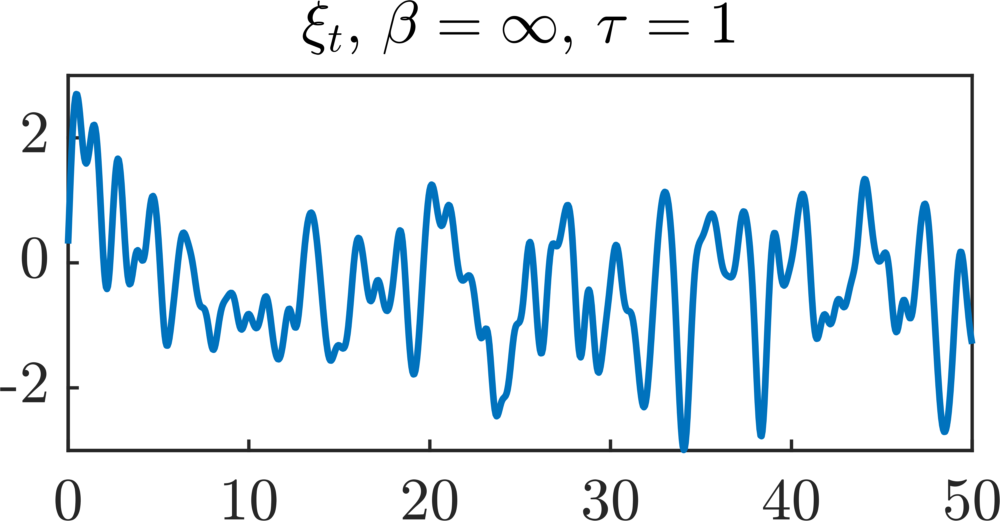} &&
\includegraphics{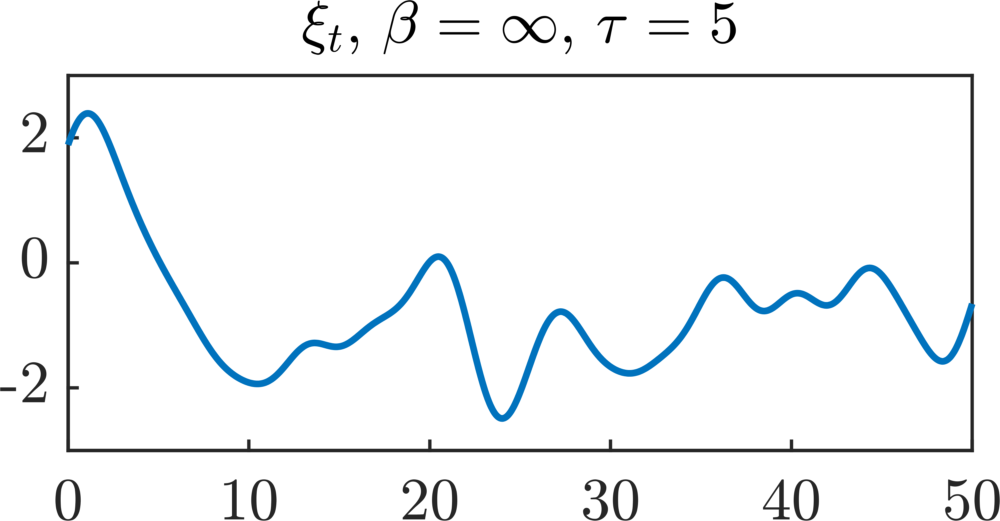}
\end{tabular}
\end{center}
\caption{Simulated trajectories of the model for different values of $\beta$ and $\tau$.}
\label{fig:model_trajectory}
\end{figure}

We begin by simulating trajectories of the stochastic process $\xi_t$ in equation \eqref{eq:def_xi}, using the kernel $\theta_\beta$ from equation \eqref{eq:Matern_kernel}. The smoothness parameter is varied as $\beta = 0, 0.25, 0.5, 1, 2$, and we also consider $\beta = \infty$, which corresponds to the Gaussian kernel in equation \eqref{eq:Gaussian_kernel}. In \cref{fig:model_trajectory}, trajectories are shown for relaxation times $\tau = 1$ and $5$. As expected, increasing $\beta$ enhances the regularity of $\xi_t$, while larger $\tau$ reduces the frequency of oscillations.

\begin{figure}
\begin{center}
\includegraphics{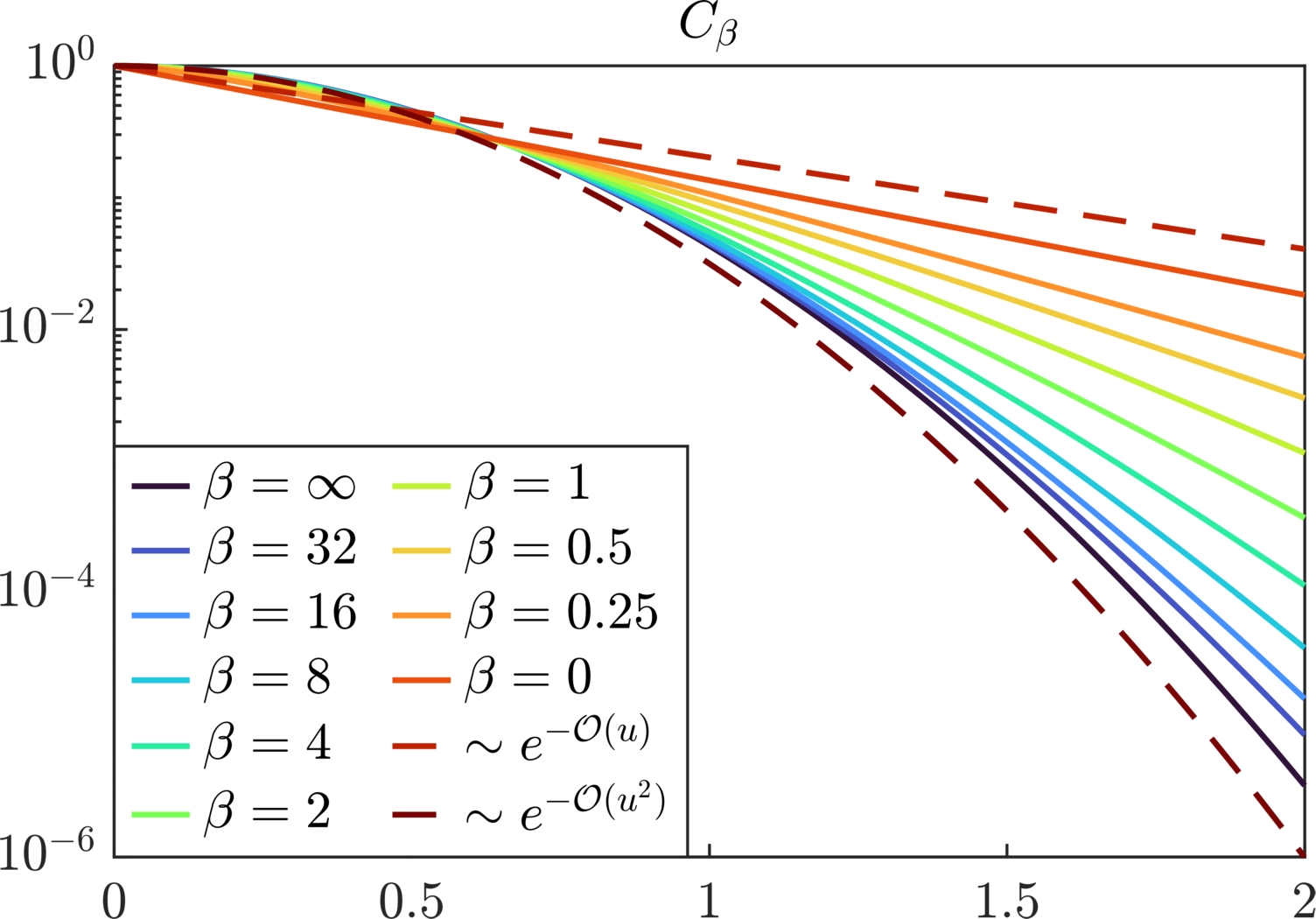} \\ \vspace{0.5cm}
\includegraphics{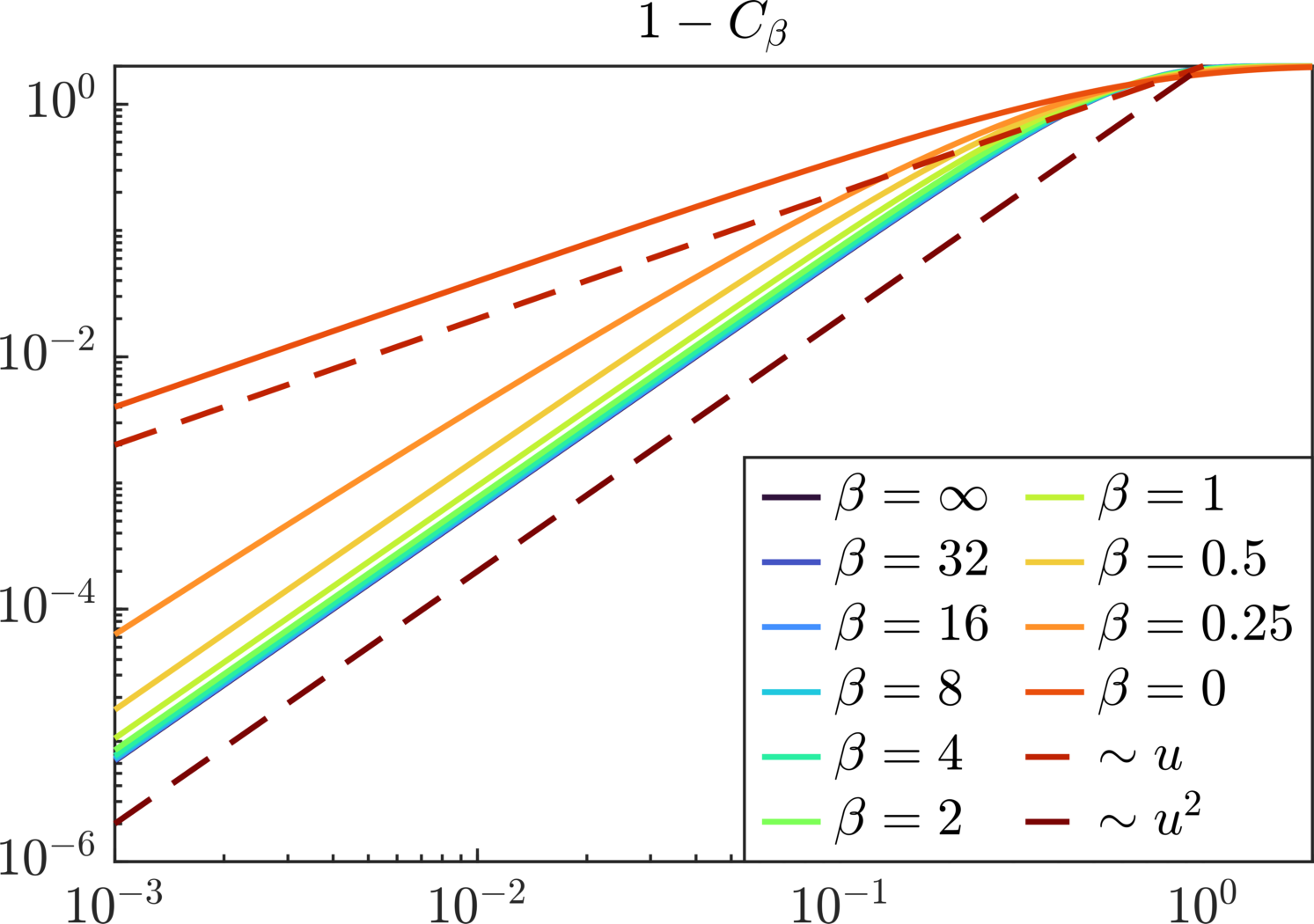}
\end{center}
\caption{Autocorrelation function $C_\beta$ from \cref{ex:Matern_Gaussian} and the function $2(1 - C_\beta)$, corresponding to the statistics in equations \eqref{eq:Rk_Delta} and \eqref{eq:Vk_Delta} for the vorticity field, respectively, for different values of the parameter $\beta$.}
\label{fig:model_correlation}
\end{figure}

Next, we examine the autocovariance function $C_\beta$. In \cref{fig:model_correlation}, we present both a semi-logarithmic plot of $C_\beta(u)$ and a logarithmic plot of $1 - C_\beta(u)$, which correspond to $R_{\mathbf k}(\Delta)$ and $V_{\mathbf k}(\Delta)$ from \cref{sec:numerics_Fourier}. Several values of $\beta$ are considered: $0, 0.25, 0.5, 1, 2, 4, 8, 16, 32, \infty$. By adjusting $\beta$, the model can capture the full spectrum of autocorrelation behaviors, ranging from the exponential decay of the Ornstein–Uhlenbeck process ($\beta = 0$) to Gaussian decay ($\beta = \infty$). In the logarithmic plot of $1 - C_\beta$, the initial slope varies continuously from $1$ for $\beta = 0$ to $2$ for $\beta = \infty$, covering the range typical of fractional Ornstein–Uhlenbeck processes with Hurst index $H \ge 1/2$.

\begin{figure}
\begin{center}
\includegraphics{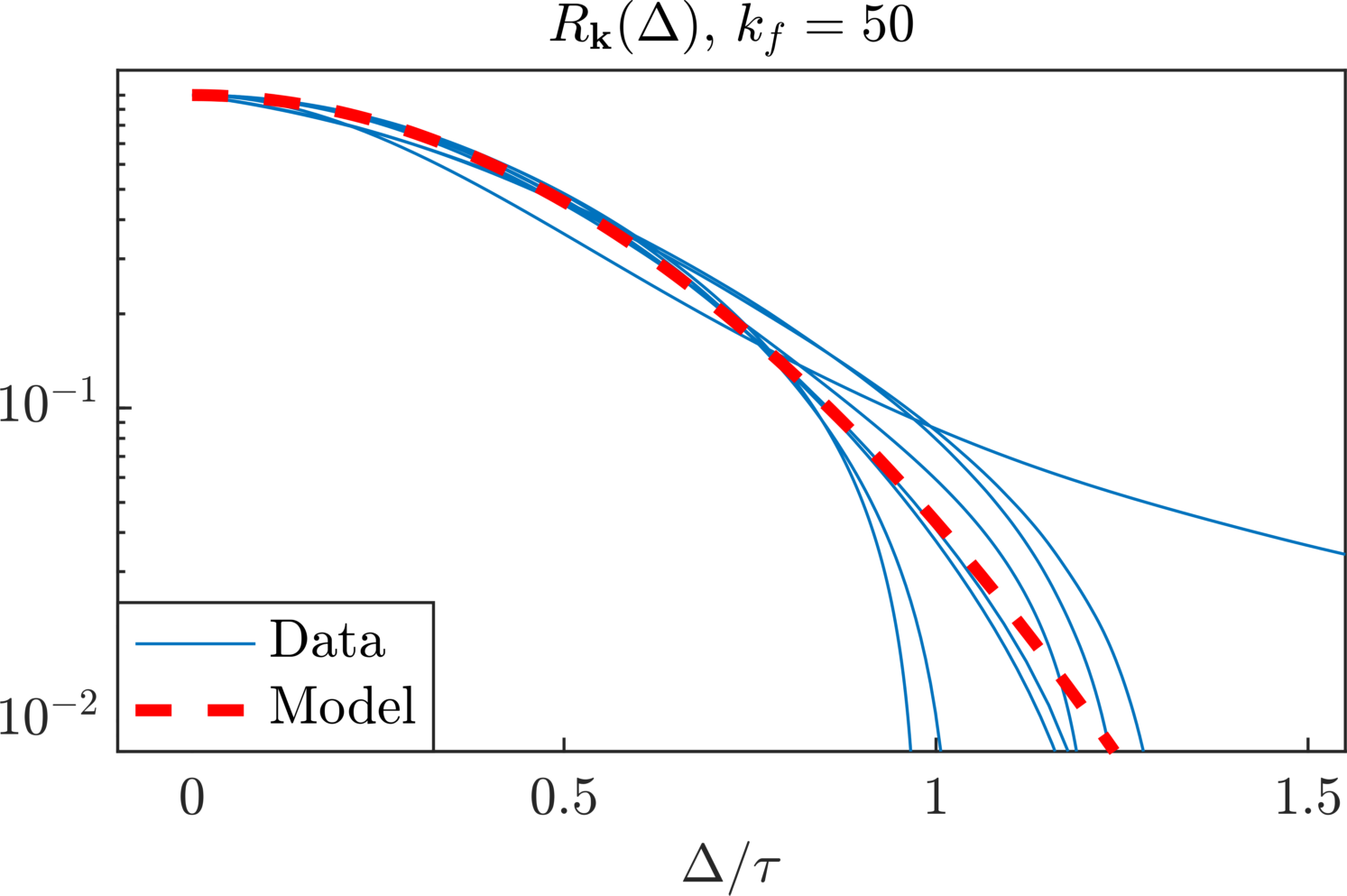} \\ \vspace{0.5cm}
\includegraphics{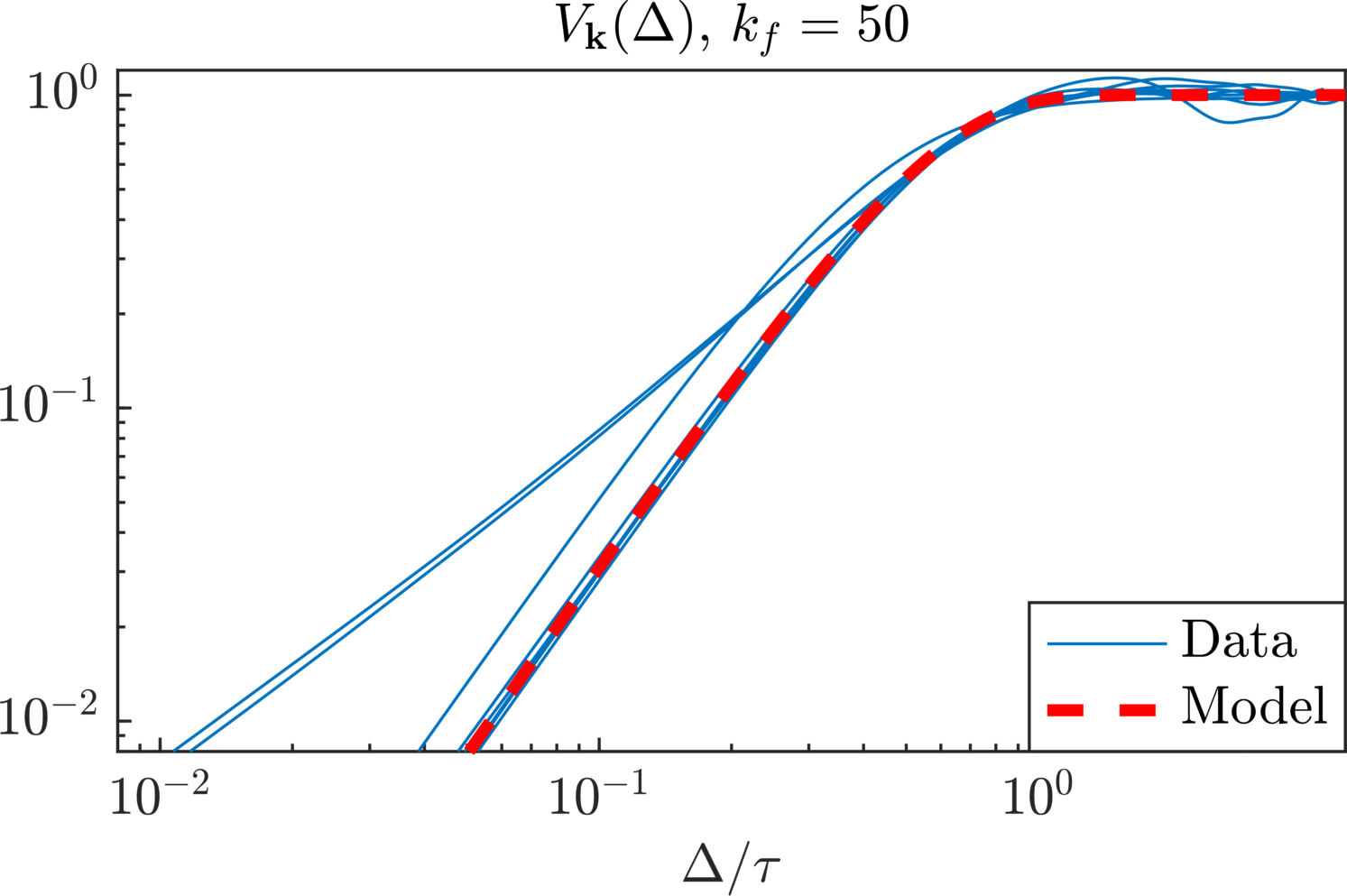}
\end{center}
\caption{Autocorrelation function $R_{\mathbf k}(\Delta)$ and normalized variance of the increments $V_{\mathbf k}(\Delta)$, both rescaled by the relaxation time $\tau_{\mathbf k}$, for different Fourier modes with forcing applied at $\abs{\mathbf k} \sim k_f = 50$, compared with the Gaussian kernel model.}
\label{fig:model_comparison_k50}
\end{figure}

\begin{figure}
\begin{center}
\includegraphics{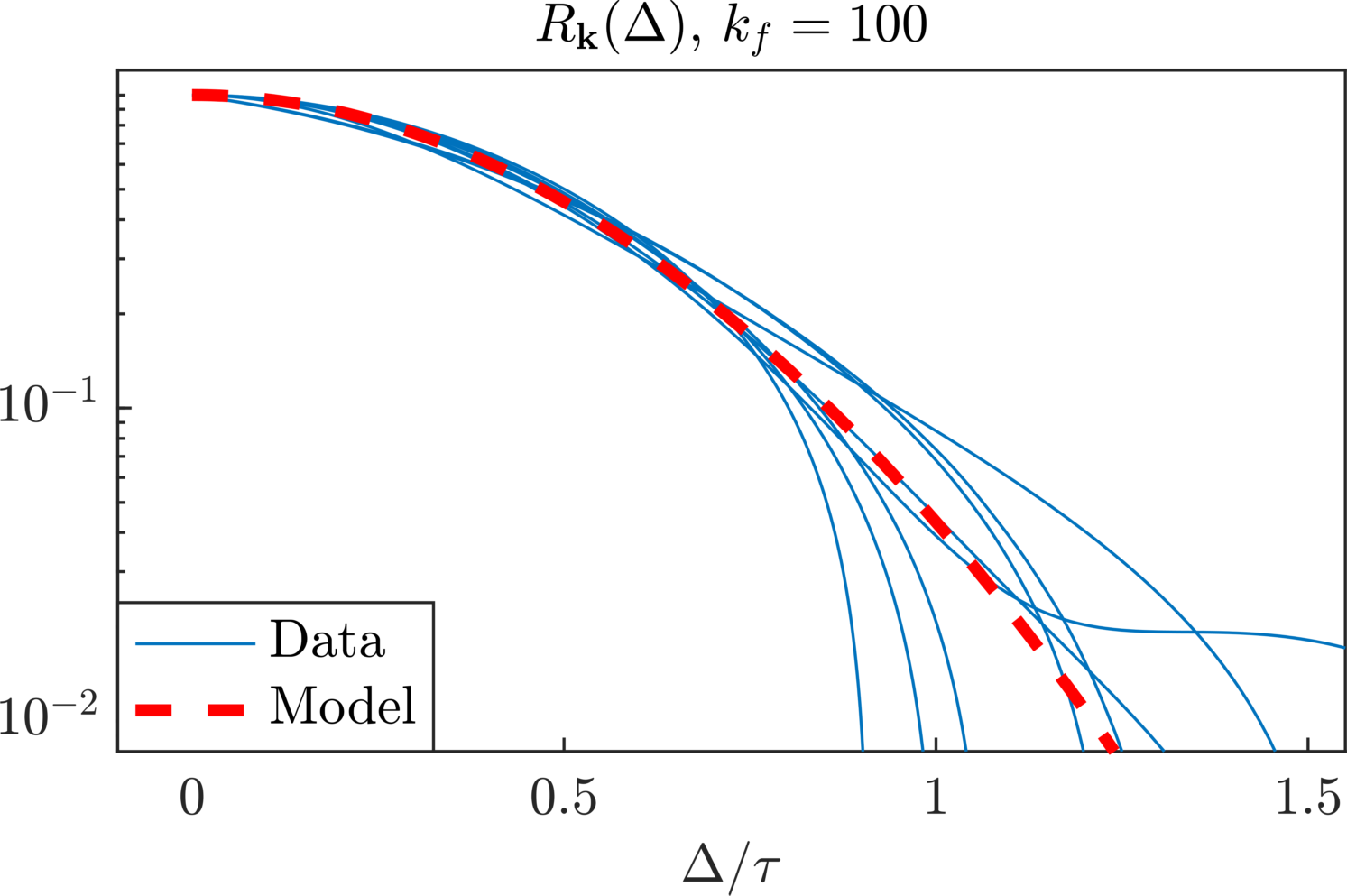} \\ \vspace{0.5cm}
\includegraphics{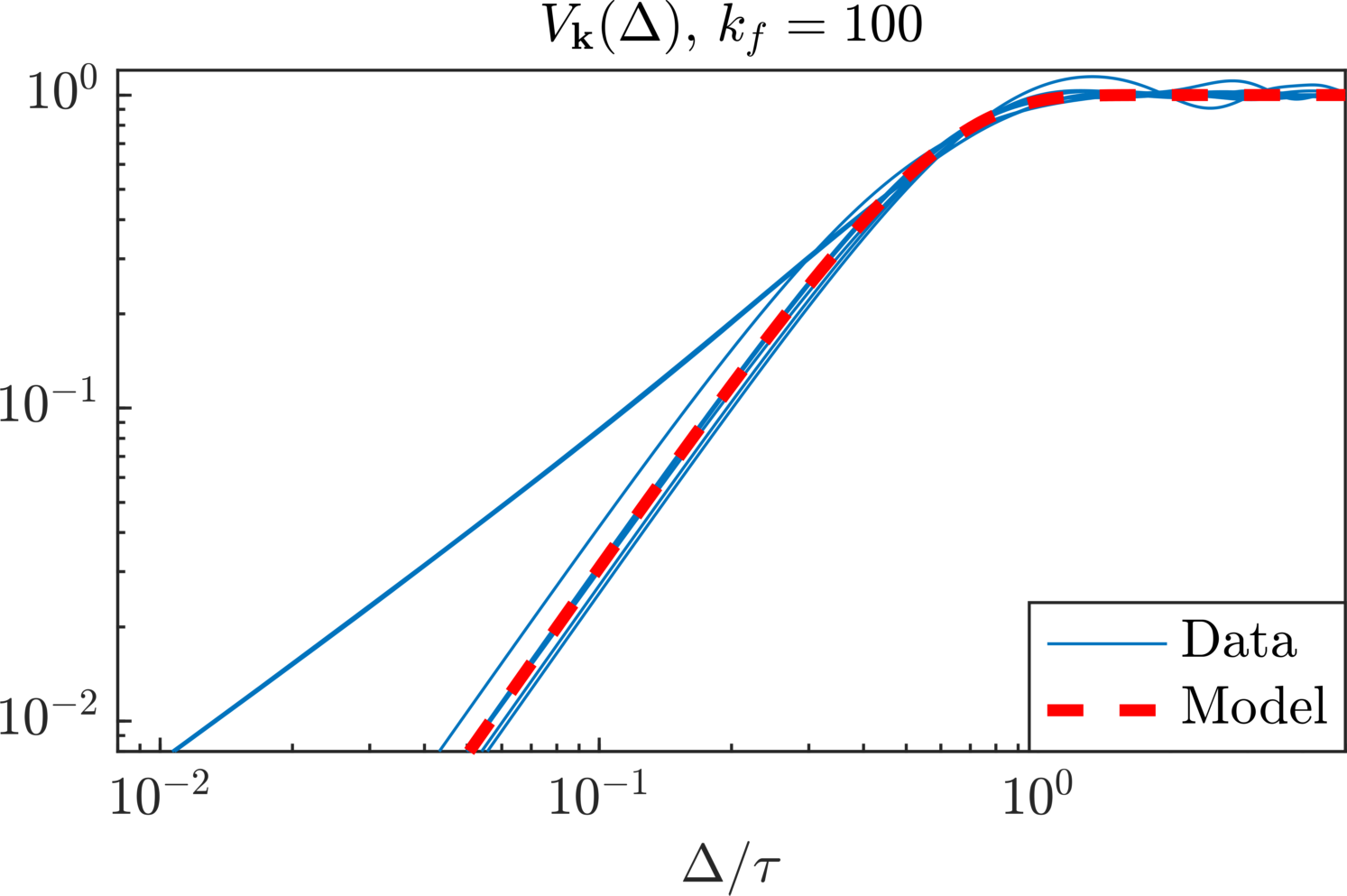}
\end{center}
\caption{Autocorrelation function $R_{\mathbf k}(\Delta)$ and normalized variance of the increments $V_{\mathbf k}(\Delta)$, both rescaled by the relaxation time $\tau_{\mathbf k}$, for different Fourier modes with forcing applied at $\abs{\mathbf k} \sim k_f = 100$, compared with the Gaussian kernel model.}
\label{fig:model_comparison_k100}
\end{figure}

To compare with the fluid data, we fit $C_\beta$ to the estimated autocorrelation $R_{\mathbf k}(\Delta)$ across all wavenumbers, determining the optimal $\beta$. For most modes, this optimal value is large, and since the difference between large $\beta$ and $\beta = \infty$ is small, we choose the simplest model and focus on the Gaussian kernel case ($\beta = \infty$). In \cref{fig:model_comparison_k50,fig:model_comparison_k100}, we modify \cref{fig:Fourier_correlation_k50,fig:Fourier_correlation_k100} by rescaling both $R_{\mathbf k}(\Delta)$ and the normalized variance of the increments $V_{\mathbf k}(\Delta)$ with respect to the relaxation time $\tau_{\mathbf k}$, while overlaying $C(u)$ and $1 - C(u)$, respectively. According to \cref{lem:autocovariance,lem:normalizations}, we expect
\begin{equation}
\begin{aligned}
R_{\mathbf k}(\Delta) &\simeq C \left( \frac{\Delta}{\tau_{\mathbf k}} \right), \\
V_{\mathbf k}(\Delta) &\simeq 1 - C \left( \frac{\Delta}{\tau_{\mathbf k}} \right),
\end{aligned}
\end{equation}
which is indeed observed numerically. The model is in agreement with the data for most Fourier modes that are not forced, particularly for values above $\sim 0.1$, where the confidence intervals in \cref{fig:Fourier_confidence_k50,fig:Fourier_confidence_k100} are smaller.

\begin{figure}
\begin{center}
\includegraphics{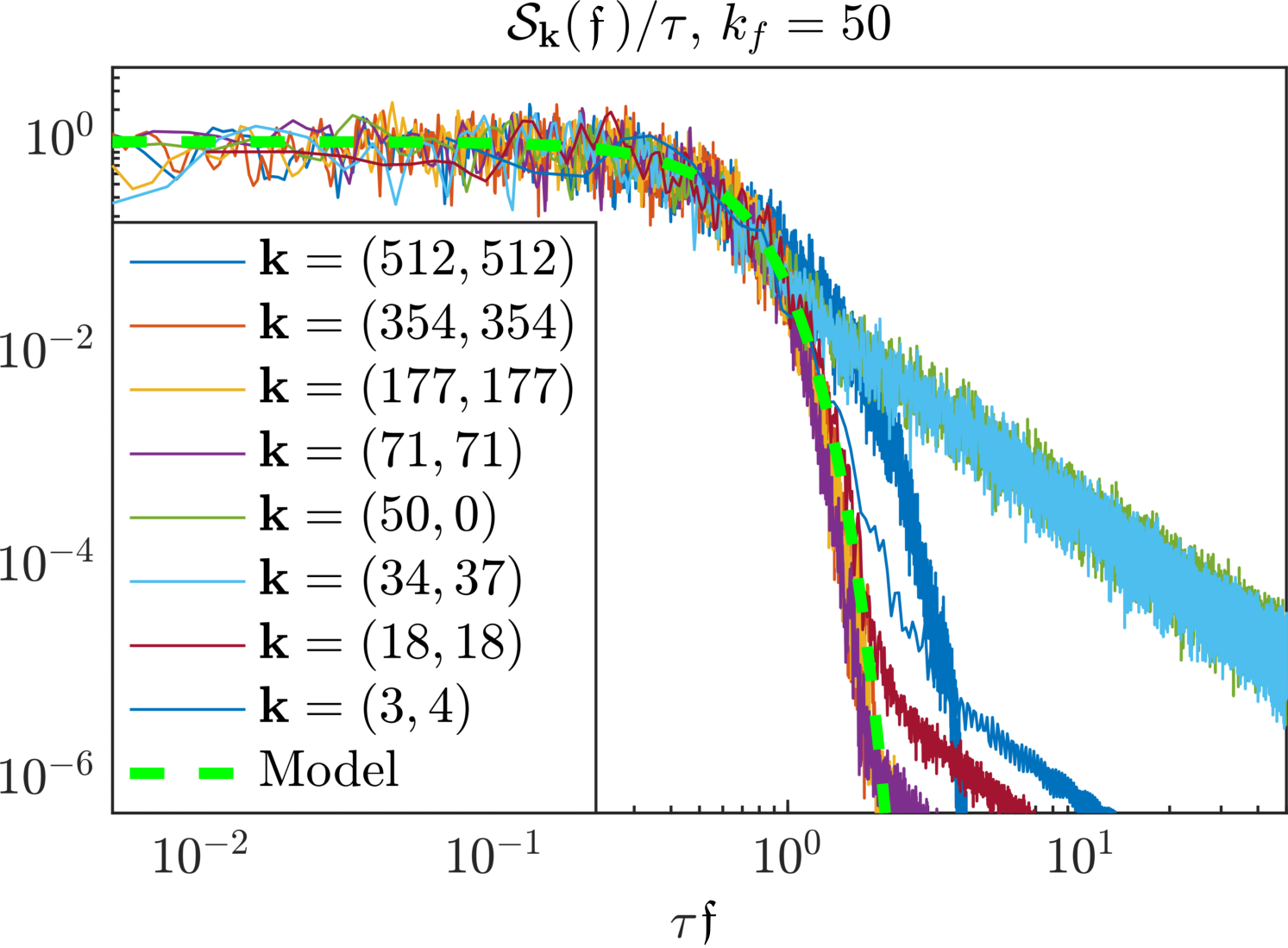} \\ \vspace{0.5cm}
\includegraphics{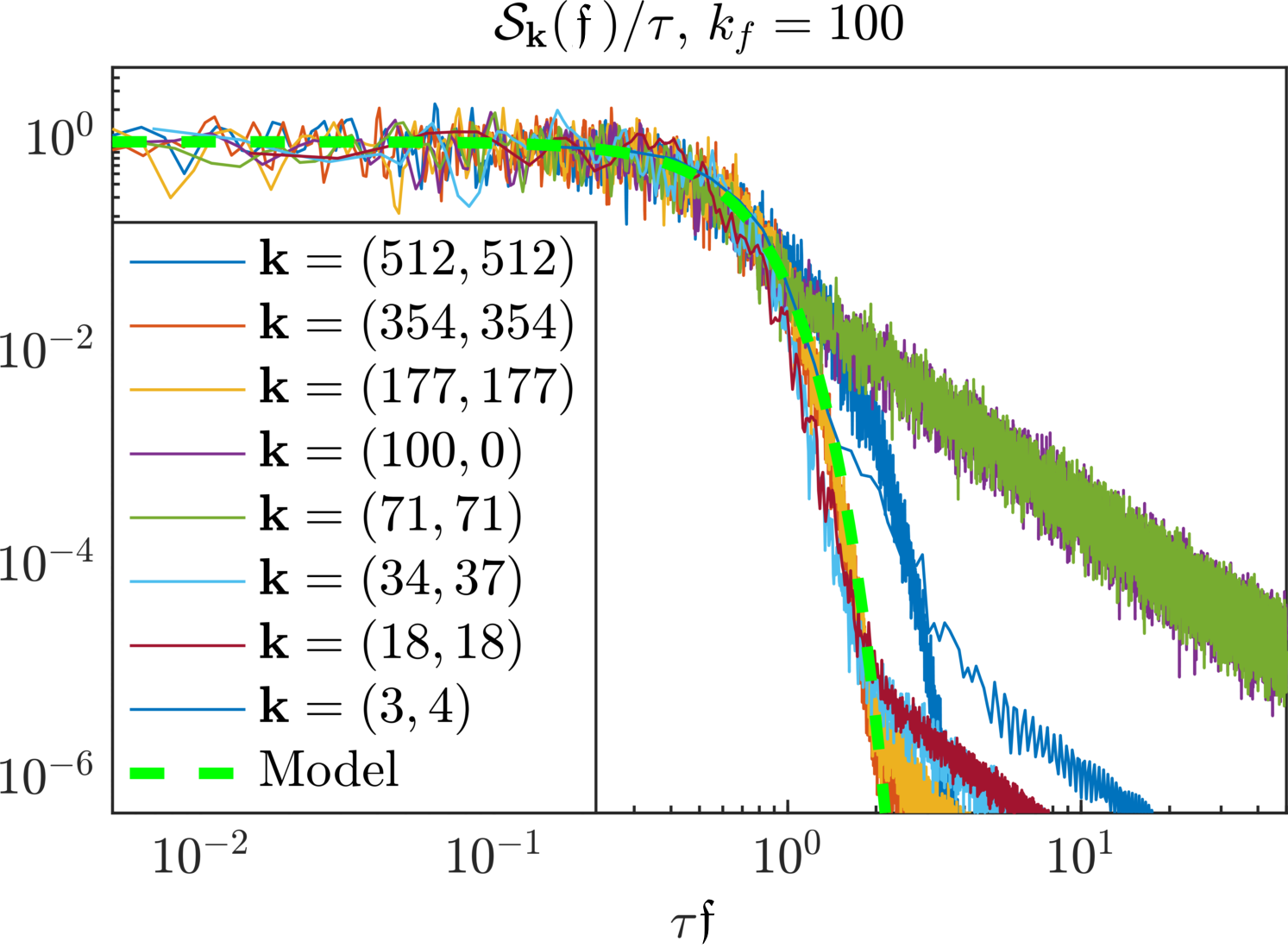}
\end{center}
\caption{Normalized power spectral density $\mathcal S_{\mathbf k}(\mathfrak f)$, rescaled by the relaxation time $\tau_{\mathbf k}$, for different Fourier modes with forcing applied at $\abs{\mathbf k} \sim k_f = 50$ (top) and $\abs{\mathbf k} \sim k_f = 100$ (bottom), compared with the Gaussian kernel model.}
\label{fig:temporal_power_spectra}
\end{figure}

Moreover, as additional validation metric, we consider the normalized power spectral density $\mathcal S_{\mathbf k}$ defined as
\begin{equation}
\mathcal S_{\mathbf k}(\mathfrak f) = \frac{\widetilde{\mathcal S}_{\mathbf k}(\mathfrak f)}{\int_{-\infty}^{+\infty} \widetilde{\mathcal S}_{\mathbf k}(\mathfrak f) \dd \mathfrak f},
\end{equation}
where $\widetilde{\mathcal S}_{\mathbf k}(\mathfrak f)$ is the power spectral density estimated from the data using Welch's method implemented in the MATLAB function \texttt{pwelch} (Signal Processing Toolbox) \cite{Wel67}. For the model associated with kernel \eqref{eq:Gaussian_kernel}, the normalized temporal power spectral density can be computed analytically and reads
\begin{equation}
\mathcal S(\mathfrak f) = \tau e^{- \pi (\tau \mathfrak f)^2}.
\end{equation}
Therefore, to compare different modes with the stochastic model, we plot the rescaled spectra $\mathcal S_{\mathbf k}(\mathfrak f / \tau_{\mathbf k})/ \tau_{\mathbf k}$ together with the function $e^{-\pi u^2}$. In \cref{fig:temporal_power_spectra}, we observe that the model captures the correct trend in most of the unforced Fourier modes, for both $k_f = 50$ and $k_f = 100$. Hence, this model is overall able to capture the statistical correlation properties of the simulated two-dimensional turbulent flow.

\section{The velocity field} \label{sec:velocity_field}

We now have all the ingredients needed to construct a stochastic model for the velocity field based on \eqref{eq:model_u}. For the temporal processes $\xi_t^{\mathbf k}$, we adopt the construction introduced in \cref{sec:model_process}. For the coefficients $\sigma_{\mathbf k}$, we follow a standard approach and determine them from the energy spectrum.

In particular, we model the velocity field $\mathbf u (\mathbf x, t)$ as a Fourier series with random coefficients
\begin{equation}
\mathbf u (\mathbf x, t)  = \sum_{\mathbf k} \sqrt{\mathcal E_{\mathbf k}} \mathbf e_{\mathbf k} (\mathbf x)  \xi_t^{\mathbf k},
\label{eq:vel_model_sec4}
\end{equation}
where each stochastic coefficient $\xi_t^{\mathbf k}$ is defined by the model in the previous section
\begin{equation}
\xi_t^{\mathbf k} = \int_{-\infty}^{+\infty} \frac1{\sqrt{\tau_{\mathbf k}}} \theta_{\mathbf k} \left( \frac{t-s}{\tau_{\mathbf k}} \right) \dd W_s^{\mathbf k},
\end{equation}
which has unit variance by \cref{lem:normalizations}, and where $W_s^{\mathbf k}$ are independent two-sided Brownian motions. The functions $\mathbf e_{\mathbf k}(\mathbf x)$ form a prescribed divergence-free Fourier basis, with the direction of each mode fixed by the wave vector $\mathbf k$ and the incompressibility constraint $\mathbf e_{\mathbf k}(\mathbf x) \cdot \mathbf{k} = 0$. In particular, these basis functions are not inferred from the simulated data and can be defined, for example, as
\begin{equation}
\mathbf e_{\mathbf k}(\mathbf x) = \frac{\mathbf k^\perp}{\abs{\mathbf k}} e^{2\pi i \mathbf k \cdot \mathbf x}.
\end{equation}
The information extracted from the simulations enters only through the energy spectrum $\mathcal E_{\mathbf k}$ and the temporal statistics of the scalar coefficients $\xi_t^{\mathbf k}$. Therefore, the resulting velocity field should be interpreted as a kinematic surrogate reproducing selected Eulerian second-order statistics, rather than as a dynamical prediction of the turbulent velocity field. Specifically, the model does not reproduce nonlinear mode interactions, inter-mode phase coupling, coherent vortex structures, or energy/enstrophy fluxes.

The energy spectrum $\mathcal E_{\mathbf k}$ is indexed by wave vectors $\mathbf k$, while the classical spectrum $E_k$, indexed by positive real numbers, is approximately defined by
\begin{equation}
E_k \sim \sum_{\abs{\mathbf k} \sim k} \mathcal E_{\mathbf k}.
\end{equation}
Assuming that the spectrum $\mathcal E_{\mathbf k}$ is roughly invariant under rotations, one obtains
\begin{equation}
\mathcal E_{\mathbf k} \sim \frac{E_{\abs{\mathbf k}}}{2\pi \abs{\mathbf k}}.
\end{equation}
This representation provides a convenient decomposition of the velocity field into spatial Fourier modes with independent stochastic temporal coefficients, each scaled consistently with the energy spectrum. In particular, the velocity field can be written as
\begin{equation}
\mathbf u (\mathbf x, t)  = \sum_{\mathbf k} \sqrt{\frac{E_{\abs{\mathbf k}}}{2\pi \abs{\mathbf k}}} \mathbf e_{\mathbf k} (\mathbf x)  \xi_t^{\mathbf k}.
\end{equation}

\subsection{White noise approximation of the model} \label{sec:approximation_model}

The model derived so far involves a large number of Fourier modes and, consequently, a large collection of stochastic processes, which limits its practical applicability. In this section, we show how, for small relaxation times, the model can be further approximated by a white noise model.

In the regime where all relaxation times $\tau_{\mathbf k}$ are small, which is the case when $\alpha$ is large, we approximate each process $\xi_t^{\mathbf k}$ by a rescaled white noise
\begin{equation}
\xi_t^{\mathbf k} \simeq \sqrt{\tau_{\mathbf k}} \frac{\d W_t^{\mathbf k}}{\d t}.
\end{equation}
This approximation is made rigorous in the following result, which shows that, in the limit of vanishing relaxation time, the integral of the process $\xi_t$ converges to a suitably rescaled Brownian motion. This result is analogous to \cite[Lemma 6]{CiF25}, which establishes the corresponding result for fractional Ornstein--Uhlenbeck processes.

\begin{proposition} \label{pro:approximation_white_noise}
Let $\xi_t$ be defined in equation \eqref{eq:def_xi}. Under \cref{as:condition_theta}, it holds
\begin{equation}
\left( \E \left[ \left( \int_0^t \xi_s \dd s - \sqrt\tau W_t \right)^2 \right] \right)^{\frac12} \le \sqrt{2M} \tau,
\end{equation}
where $M > 0$ is defined in \cref{as:condition_theta}.
\end{proposition}
\begin{proof}
Let us first rewrite the two terms in the left-hand side of the statement. Denoting $\vartheta_\tau = \theta(\cdot/\tau)/\tau$, we have
\begin{equation}
\begin{aligned}
\int_0^t \xi_s \dd s &= \sqrt\tau \int_{-\infty}^{+\infty} (\mathbbm{1}_{[0,t]} * \vartheta_\tau )(r) \dd W_r, \\
\sqrt\tau W_t &= \sqrt\tau \int_{-\infty}^{+\infty} \mathbbm{1}_{[0,t]}(r) \dd W_r,
\end{aligned}
\end{equation}
which, due to the Itô isometry, implies
\begin{equation}
\begin{aligned}
\E \left[ \left( \int_0^t \xi_s \dd s - \sqrt\tau W_t \right)^2 \right] &= \tau \int_{-\infty}^{+\infty} \left( (\mathbbm{1}_{[0,t]} * \vartheta_\tau )(r) - \mathbbm{1}_{[0,t]}(r) \right)^2 \dd r \\
&\le \tau \int_{-\infty}^{+\infty} \abs{(\mathbbm{1}_{[0,t]} * \vartheta_\tau )(r) - \mathbbm{1}_{[0,t]}(r)} \dd r,
\end{aligned}
\end{equation}
where in the last row we used that 
\begin{equation}
\abs{(\mathbbm{1}_{[0,t]} * \vartheta_\tau )(r) - \mathbbm{1}_{[0,t]}(r)} \le 1.
\end{equation}
Notice that
\begin{equation}
\int_{-\infty}^{+\infty} \abs{\mathbbm{1}_{[0,t]}(r-s) - \mathbbm{1}_{[0,t]}(r)} \dd r \le 2\abs{s},
\end{equation}
which implies
\begin{equation}
\begin{aligned}
\E \left[ \left( \int_0^t \xi_s \dd s - \sqrt\tau W_t \right)^2 \right] &\le \tau \int_{-\infty}^{+\infty} \int_{-\infty}^{+\infty} \abs{\mathbbm{1}_{[0,t]}(r-s) - \mathbbm{1}_{[0,t]}(r)} \dd r \vartheta_\tau(s) \dd s \\
&\le 2\tau \int_{-\infty}^{+\infty} \abs{s} \vartheta_\tau(s) \dd s.
\end{aligned}
\end{equation}
Finally, by definition of $\vartheta_\tau$ and doing the change of variables $x = s/\tau$, we deduce
\begin{equation}
\E \left[ \left( \int_0^t \xi_s \dd s - \sqrt\tau W_t \right)^2 \right] \le 2 \tau^2 \int_{-\infty}^{+\infty} \abs{x} \theta(x) \dd x = 2M\tau^2,
\end{equation}
which gives the desired result.
\end{proof}

Therefore, assuming all $\tau_{\mathbf k}$ relatively small, \cref{pro:approximation_white_noise} gives the approximation 
\begin{equation}
\mathbf u (\mathbf x, t) \simeq \sum_{\mathbf k} \sqrt{\frac{E_{\abs{\mathbf k}}}{2\pi\abs{\mathbf k}}} \tau_{\mathbf k} \mathbf e_{\mathbf k} (\mathbf x) \frac{\d W_t^{\mathbf k}}{\d t}.
\end{equation}
This demonstrates that our model is consistent with the white noise approximation, an idea that has been extensively studied in the literature on turbulence modeling \cite{FlL23}. However, this approximation deteriorates when large scale structures, for which $\tau_{\mathbf k}$ is not small, are present.

\section{Stochastic transport and effective diffusion} \label{sec:heat_equation}

As anticipated in \cref{sec:intro}, there are many reasons to introduce over-simplified models of turbulent velocity fields. One, which is not our purpose here but deserves to be mentioned, is the supplementation of numerical codes for large-scale systems with synthetic models of small scales that cannot be resolved; see, for example, the STUOD project \cite{CCH23}, based on synthetic models of turbulent transport, which have some analogy with those discussed here.

For us, the motivation is to build a tool for \emph{analytical} investigations. The most typical example—but there are many others; see, for instance, \cite{FlT26,PFH23,PPV23}, dealing with polymers and particle aggregation—is understanding the transport of a passive temperature field by a turbulent fluid, modeled by the advection--diffusion equation
\begin{equation} \label{heat_eq}
\partial_t T + \mathbf u \cdot \nabla T = \kappa \Delta T.
\end{equation}
Here, $T$ is treated as a passive scalar, so thermodynamic coupling is neglected. Obviously, especially in two dimensions, direct numerical simulation of the Navier--Stokes equations coupled with the passive-scalar advection--diffusion equation provides the most accurate description of the temperature field and the associated heat flux within this approximation. However, numerical simulations alone may not help to \emph{explain} behaviors, for instance, the dependence of certain quantities on others (like the Nusselt number on the Rayleigh number in Rayleigh--Bénard convection, or the energy confinement time on the magnetic field strength in confined fusion plasma). In order to increase our understanding of certain phenomena and their links, \emph{simplified} models may help. Using the Navier--Stokes equations directly in a turbulent regime for \emph{analytical} computations is almost impossible nowadays. Therefore, replacing their solution in a turbulent regime with a stochastic model, even if extremely idealized and simplified, may provide some qualitative information.

This is the rationale behind the very strong approximations made so far and below in this section. Let us explain how a stochastic model may provide information on the transport of a passive scalar. First, recall that the average flux (up to geometrical details of the domain, which we do not discuss here) is an average of the form $\left\langle \mathbf u T - \kappa \nabla T \right\rangle$, which requires complicated knowledge of $\mathbf u$. But if we have a rational way to replace the turbulent dispersion produced by $\mathbf u$ with a turbulent diffusivity, then the average flux takes the form $\left\langle (\kappa + \kappa_T) \nabla T \right\rangle$ for some appropriate $\kappa_T$, which increases our ability to perform analytical computations.

Assume the velocity is given by a sum of components
\begin{equation} \label{general_velocity}
\mathbf u(\mathbf x, t) = \sum_{\mathbf k} \sigma_{\mathbf k}(\mathbf x) \xi_t^{\mathbf k}.
\end{equation}
In a number of cases (but not always), the average quantity $\overline T(\mathbf x, t) = \E[T(\mathbf x,t)]$ solves an equation of the form
\begin{equation} \label{heat_eq_averaged}
\partial_t \overline T = \kappa \Delta \overline T + \div \left(Q(\mathbf x, t)  \nabla \overline T \right),
\end{equation}
where the matrix-valued function $Q(\mathbf x, t)$ is given by the diagonal $Q(\mathbf x, t) = Q(\mathbf x, \mathbf x, t)$ of the covariance matrix
\begin{equation}
Q(\mathbf x, \mathbf y, t) = \sum_{\mathbf k} q_{\mathbf k}(t) \sigma_{\mathbf k}(\mathbf x) \otimes \sigma_{\mathbf k}(\mathbf y),
\end{equation}
and the functions $q_{\mathbf k}(t)$ are given in terms of the auto-covariance function $\E \left[ \xi_t^{\mathbf k} \xi_0^{\mathbf k} \right]$ of the noise by
\begin{equation}
q_{\mathbf k}(t) = \int_0^t \E \left[ \xi_s^{\mathbf k} \xi_0^{\mathbf k} \right] \dd s.
\end{equation}
As said, there is no general proof of this ``Green--Kubo'' formula, and numerically we have also discovered cases where it is not accurate \cite{CiF25}. However, it is rigorously true in a number of cases, such as the case where the stochastic processes $\xi_t^{\mathbf k}$ are independent Gaussian white noises (delta-correlated), as in the Kraichnan model \cite{FlL23,KrM80}, where $q_{\mathbf k}(t) = 1$, and, more generally, in the case $\kappa = 0$, when the processes $\xi_t^{\mathbf k}$ are independent Gaussian processes with finite second moments and the vector fields $\sigma_{\mathbf k}(\mathbf x) \cdot \nabla$ commute \cite{FlR26}.

Let us assume the validity of this result, at least approximately, in our case, although, strictly speaking, our processes $\xi_t^{\mathbf k}$ have a finite time correlation — compared to the white noise case — and the vector fields do not commute — compared to \cite{FlR26}. There is intense mathematical research on generalizations of this result (see, for instance, \cite{PoV03}), and establishing a Green--Kubo formula for the class of stochastic processes introduced in this work would be an interesting direction for future research.

Let us also stress that, in the case of white noise, there are ``quenched'' results stating that single realizations of the stochastic system given by equations \eqref{heat_eq} and \eqref{general_velocity} are close to the average equation \eqref{heat_eq_averaged}, see \cite{FGL22,FlL22,FlL23}. Namely, several independent arguments point in the direction of the equation \eqref{heat_eq_averaged}, although a rigorous proof in full generality is not known, and modifications may be relevant in particular cases.

Now, in order to simplify the exposition, assume that the stochastic processes $\xi_t^{\mathbf k}$ have identical auto-correlation functions, up to a time rescaling depending on the norm of the wave vector $\abs{\mathbf k}$, of the form
\begin{equation}
\E \left[ \xi_t^{\mathbf k} \xi_0^{\mathbf k} \right] = h \left( \frac{t}{\tau_{\abs{\mathbf k}}} \right),
\end{equation}
for some function $h$ satisfying $h(0) = 1$. This is approximately the case in our simulations above. Under this assumption,
\begin{equation}
q_{\mathbf k}(t) = \int_0^t h \left( \frac{s}{\tau_{\abs{\mathbf k}}} \right) \dd s = \tau_{\abs{\mathbf k}} \int_0^{t/\tau_{\abs{\mathbf k}}} h(r) \dd r.
\end{equation}
Assume also that we are in two dimensions, on the torus $[0,2\pi]^2$, with the vector fields $\sigma_{\mathbf k}(\mathbf x)$ of the noise proportional to 
\begin{equation}
\frac{\mathbf k^{\perp}}{\abs{\mathbf k}} \cos \left(\mathbf k \cdot \mathbf x \right) \qquad \text{and} \qquad \frac{\mathbf k^\perp}{\abs{\mathbf k}} \sin \left( \mathbf k \cdot \mathbf x\right),
\end{equation}
where $\mathbf k^\perp = \left( -k_2, k_1 \right)$ if $\mathbf k = \left(  k_1, k_2 \right)$, in such a way that they appear in groups of four fields associated with the wave vectors $\mathbf k, - \mathbf k, \mathbf k^\perp, - \mathbf k^\perp$, more precisely, quartets of the form
\begin{equation}
\begin{aligned}
\mathbf e_{\mathbf k}^{(1)}(\mathbf x) &= \frac{\mathbf k^\perp}{\abs{\mathbf k}} \cos \left(\mathbf k \cdot \mathbf x \right), & \mathbf e_{\mathbf{k}}^{(2)}(\mathbf x) &= \frac{\left(\mathbf k^\perp\right)^\perp}{\abs{\mathbf k^\perp}} \cos \left( \mathbf k^\perp \cdot \mathbf x \right), \\
\mathbf e_{\mathbf k}^{(3)}(\mathbf x) &= \frac{\left(-\mathbf k \right)^\perp}{\abs{-\mathbf k}} \sin \left(-\mathbf k \cdot \mathbf x \right), \qquad& \mathbf e_{\mathbf k}^{(4)}(\mathbf x) &= \frac{\left( -\mathbf k^\perp \right)^\perp}{\abs{-\mathbf k^\perp}} \sin \left( -\mathbf k^\perp \cdot \mathbf x \right),
\end{aligned}
\end{equation}
with the obvious algebraic simplifications. An easy computation gives
\begin{equation}
\sum_{a=1}^4 \mathbf e_{\mathbf k}^{(a)}(\mathbf x) \otimes \mathbf e_{\mathbf k}^{(a)}(\mathbf x) = I_2,
\end{equation}
where $I_2$ is the two-dimensional identity matrix. Of course, when there are many quartets of this form, the identity matrix is multiplied by a suitable coefficient.

Summarizing, if we accept the general result above, even though our assumptions are somewhat borderline with respect to the currently available rigorous theorems, and we adopt the notational simplifications just described, we obtain the equation
\begin{equation} \label{heat_equation_simplified}
\partial_t \overline T = \left( \kappa + \kappa_T(t) \right) \Delta \overline T,
\end{equation}
where the ``eddy diffusivity'' $\kappa_T(t)$ is given by an expression of the form
\begin{equation}
\kappa_T(t) = \sum_{\mathbf k} \mu_{\abs{\mathbf k}} \tau_{\abs{\mathbf k}} \int_0^{t/\tau_{\abs{\mathbf k}}} h(r),\dd r,
\end{equation}
for certain coefficients $\mu_{\abs{\mathbf k}}$ that can be made explicit in terms of the precise synthetic model.

We see the main advantage of a synthetic model of the turbulent velocity field: the possibility of obtaining interpretable results for the diffusion produced by turbulence. In the next sections, we derive approximate analytical large-time and small-time behaviors and compare them with direct numerical simulations of a passive tracer advected by the solution of the stochastic Navier--Stokes equation.

\subsection{Large time diffusion} \label{sec:diffusion_large_time}

In this and the next two sections, we idealize the stochastic processes $\xi_t^{\mathbf k}$ introduced above by fitting the numerical data by assuming that they have an auto-correlation function of the form
\begin{equation}
\E \left[ \xi_t^{\mathbf k} \xi_0^{\mathbf k} \right] \sim e^{- \left( \abs{\mathbf k}^2 t^2 \right)} = e^{- \left( t/\tau_{\abs{\mathbf k}} \right)^2}, \qquad \text{where} \qquad \tau_{\abs{\mathbf k}} = \frac1{\abs{\mathbf k}},
\end{equation}
hence
\begin{equation}
\kappa_T(t) = \sum_{\mathbf k} \mu_{\abs{\mathbf k}} \tau_{\abs{\mathbf k}} \int_0^{t/\tau_{\abs{\mathbf k}}} e^{-r^2} \dd r.
\end{equation}
If $t \gg \max_{\mathbf k} \tau_{\abs{\mathbf k}}$, setting
\begin{equation}
\kappa_T \coloneq \sum_{\mathbf k} \mu_{\abs{\mathbf k}} \tau_{\abs{\mathbf k}} \int_0^\infty e^{-r^2} \dd r = \frac{\sqrt\pi}2 \sum_{\mathbf k} \mu_{\abs{\mathbf k}} \tau_{\abs{\mathbf k}},
\end{equation}
we obtain the asymptotic behavior $\kappa_T(t) \sim \kappa_T$. In other words, for large times the diffusion is classical, namely Brownian, corresponding to the fact that, on a larger time scale, the passive tracers experience random displacements due to the relatively fast vortex structures (with respect to the chosen time scale).

\subsection{Small time behavior} \label{sec:diffusion_small_time}

More complex is the small-time behavior. We divide the analysis into two parts: first, we assume for simplicity that only the largest spatial scale is relevant, then we study the modifications due to the addition of smaller scales.

Assume first that only the smallest values of $\abs{\mathbf k}$ still carrying energy matter; let those be of scale $k_{\min}$. In the sum above, they have the largest values of $\tau_{\abs{\mathbf k}}$ and also of $\mu_{\abs{\mathbf k}}$, which, by inspection, depends linearly on the energy spectrum $E_k$. We therefore approximate
\begin{equation}
\kappa_T(t) \sim \sum_{\abs{\mathbf k} \sim k_{\min}} \mu_{k_{\min}} \tau_{k_{\min}} \int_0^{t/\tau_{k_{\min}}} e^{-r^2} \dd r.
\end{equation}
Using the approximation $\int_0^t e^{-r^2} \dd r \sim t -\frac{t^3}{3}$ (indeed, we shall see that already the effect of the term $\frac{t^3}{3}$ will disappear), we get
\begin{equation}
\begin{aligned}
\kappa_T(t) &\sim \sum_{\abs{\mathbf k} \sim k_{\min}} \mu_{k_{\min}} \tau_{k_{\min}} \left( \frac{t}{\tau_{k_{\min}}} - \frac13 \left( \frac{t}{\tau_{k_{\min}}} \right)^3 \right)  \\
&= \widetilde\mu \left( t - \frac{t^3}{3 \tau_{k_{\min}}^2} \right)
\sim \widetilde\mu t
\qquad \text{for } t \ll \tau_{k_{\min}}.
\end{aligned}
\end{equation}
A linearly increasing diffusion coefficient corresponds to a ballistic behavior: the equation $\partial_t \overline T = \widetilde\mu t \Delta \overline T$ is the Fokker--Planck equation associated with the diffusion $\d X_t = \sqrt{2\widetilde\mu t} \dd W(t)$, where $W_t$ is a two-dimensional Brownian motion, and
\begin{equation} \label{eq:slope_short_time}
\E \left[ \abs{X_t}^2\right] = 2 \widetilde\mu t^2,
\end{equation}
namely $\abs{X_t}$ ``grows as $t$'' (in an average sense). This behavior is reasonable in our model for short times, when tracers move adhering to the vortices of largest size, hence as in a laminar flow, locally in time and space.

There is, however, the issue that smaller vortex structures, with shorter correlation times, may act as random perturbations and produce fluctuations on a shorter time scale. The terms we have neglected have the form
\begin{equation}
\mu_{\abs{\mathbf k}} \tau_{\abs{\mathbf k}} \int_0^{t/\tau_{\abs{\mathbf k}}} e^{-r^2} \dd r,
\end{equation}
with $\abs{\mathbf k} \gg k_{\min}$. If the chosen time $t \ll \tau_{k_{\min}}$ satisfies at the same time $t \gg \tau_{\abs{\mathbf k}}$, the neglected term behaves like
\begin{equation}
\mu_{\abs{\mathbf k}} \tau_{\abs{\mathbf k}} \int_0^{t/\tau_{\abs{\mathbf k}}} e^{-r^2} \dd r \sim \frac{\sqrt\pi}{2} \mu_{\abs{\mathbf k}} \tau_{\abs{\mathbf k}},
\end{equation}
which perturbs the linear growth $\widetilde\mu t$ found above. Thus the terms in the global sum over $\mathbf k$ divide into three categories, for a chosen value $t \ll \tau_{k_{\min}}$: those that satisfy $\tau_{\abs{\mathbf k}} \ll \tau_{k_{\min}}$, which accumulate into ballistic behavior, with a smaller weight than those of $k_{\min}$; those that satisfy $t \gg \tau_{\abs{\mathbf k}}$, which produce a Brownian perturbation, although small in amplitude; and the intermediate ones, even more difficult to characterize, but with intermediate behavior and strength.

Summarizing, for small times there is some dominance of ballistic behavior, although not a clear one. In the next section we compare these theoretical predictions, based analytically on the synthetic model of turbulent velocity together with the theoretical arguments leading to equation \eqref{heat_equation_simplified}, with direct numerical simulations of tracer diffusion.

\subsection{Numerical verification}

\begin{figure}
\begin{center}
\begin{tabular}{ccc}
\includegraphics{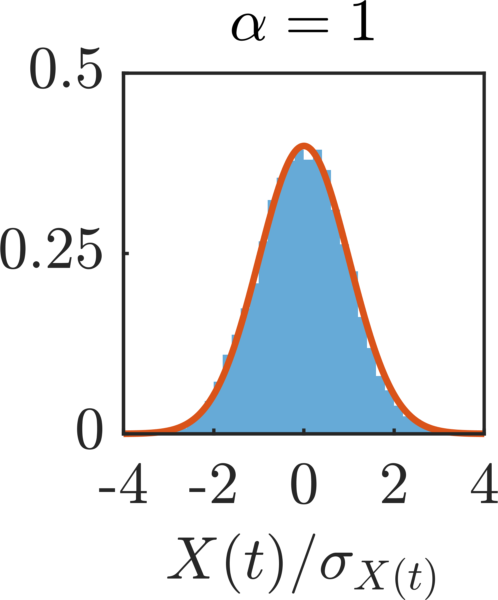} &
\includegraphics{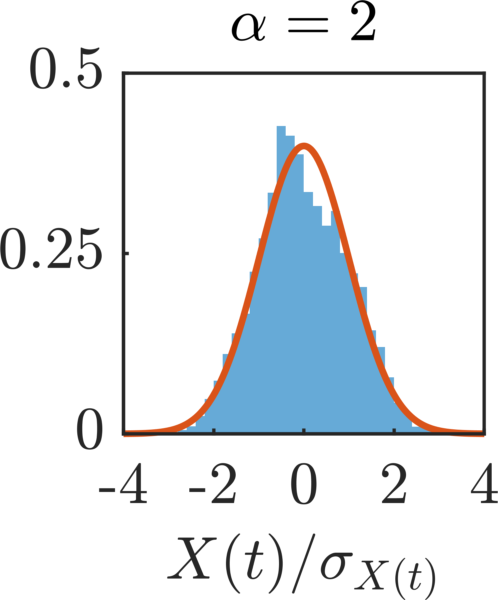} &
\includegraphics{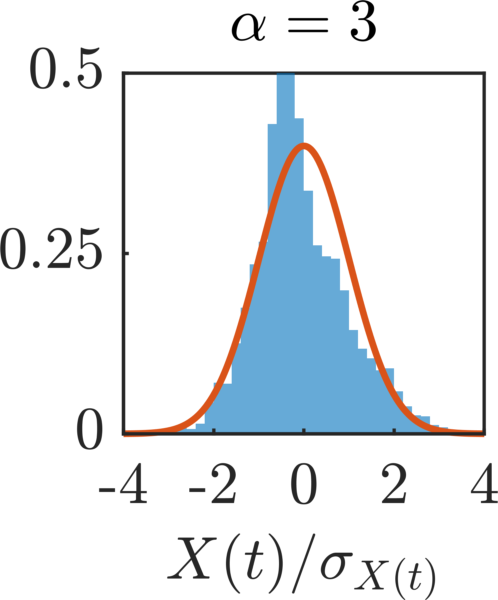} \\[0.5cm]
\includegraphics{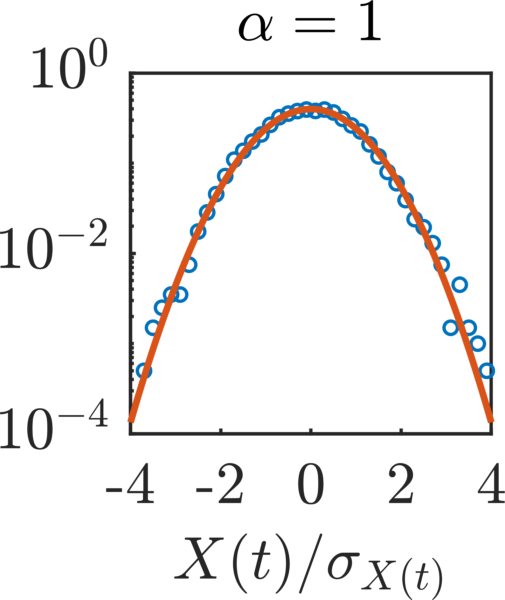} &
\includegraphics{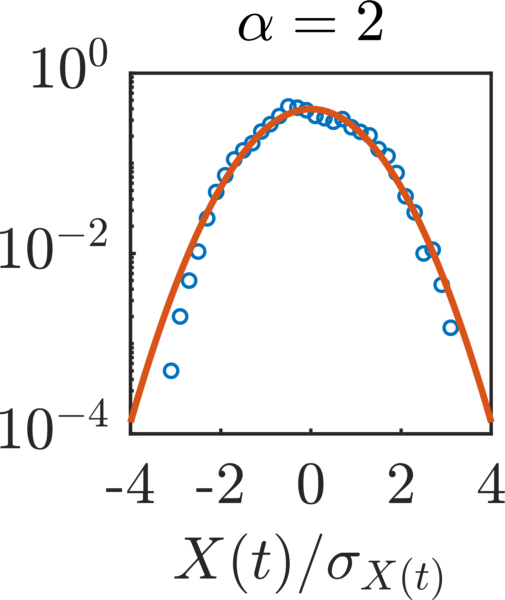} &
\includegraphics{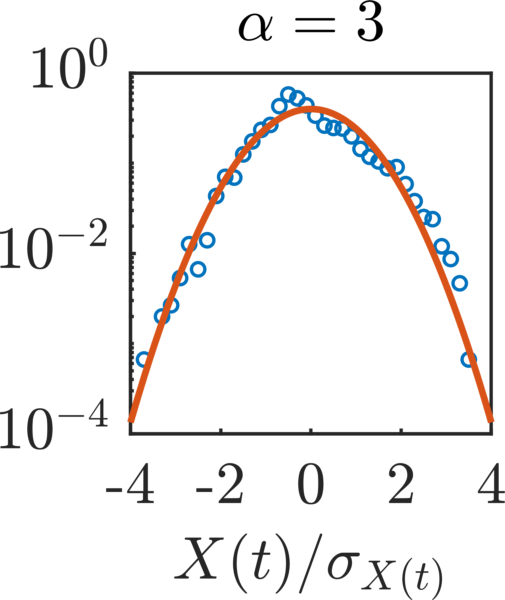}
\end{tabular}
\end{center}
\caption{Probability density function of the displacement $X(t)$ at $t = 10$, for decreasing values of the damping parameter $\alpha$, i.e., $\alpha_1 > \alpha_2 > \alpha_3 > 0$, compared with a Gaussian distribution $\mathcal N(0,\sigma_{X(t)}^2)$, shown in both linear (top) and semilogarithmic (bottom) scales.}
\label{fig:diffusion_tracer_PDF}
\end{figure}

\begin{figure}
\begin{center}
\includegraphics{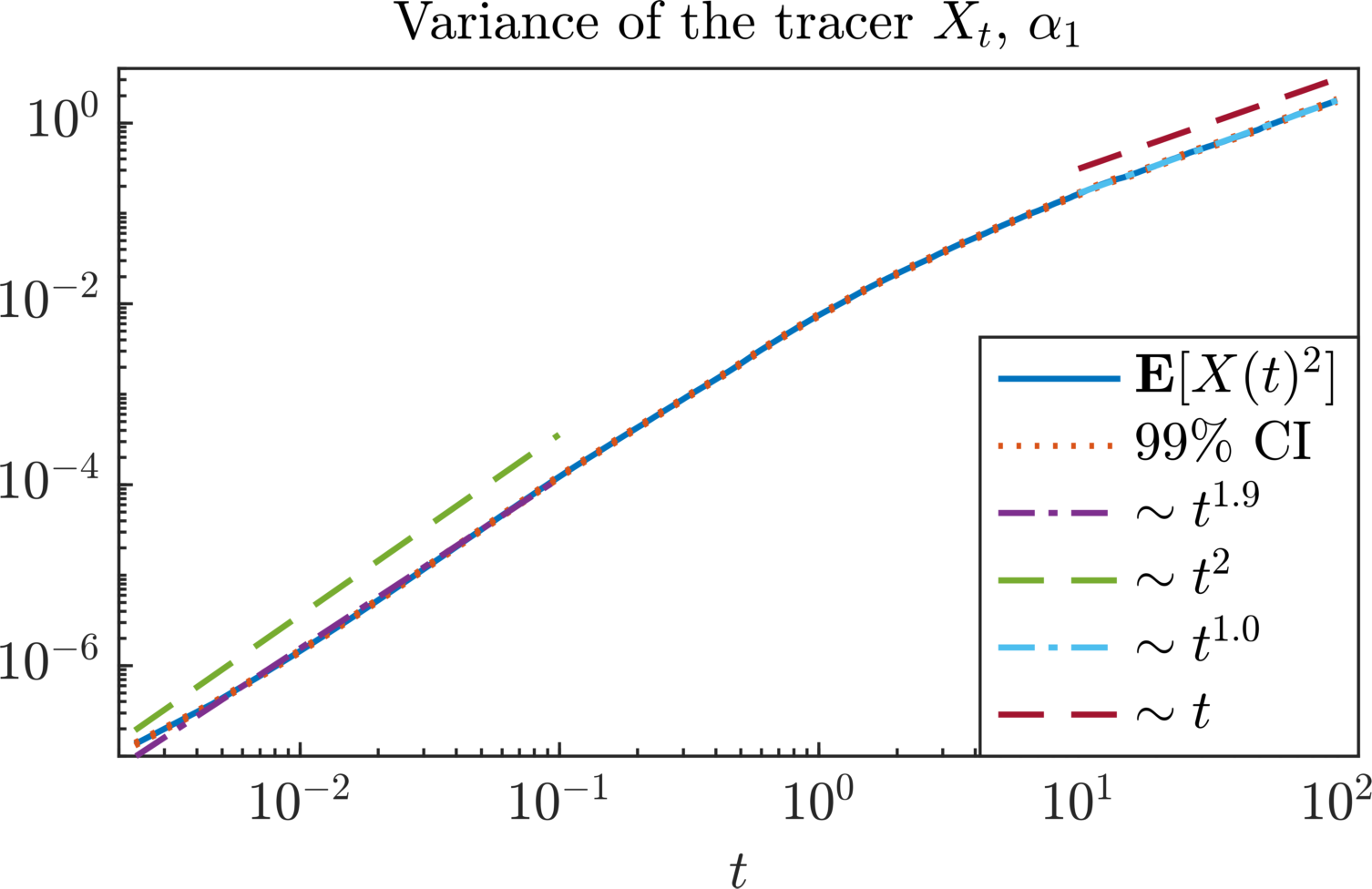} \\ \vspace{0.5cm}
\includegraphics{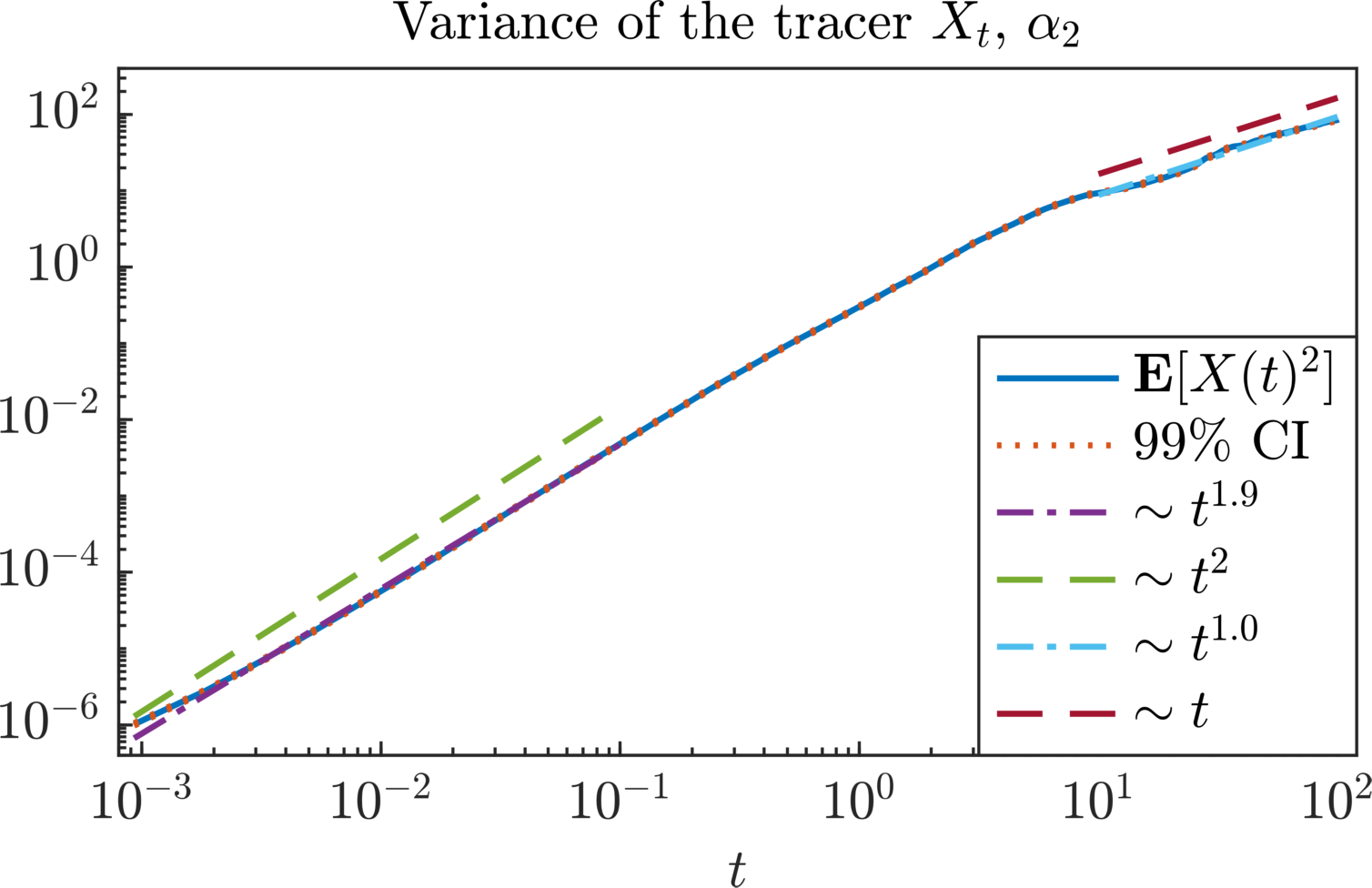} \\ \vspace{0.5cm}
\includegraphics{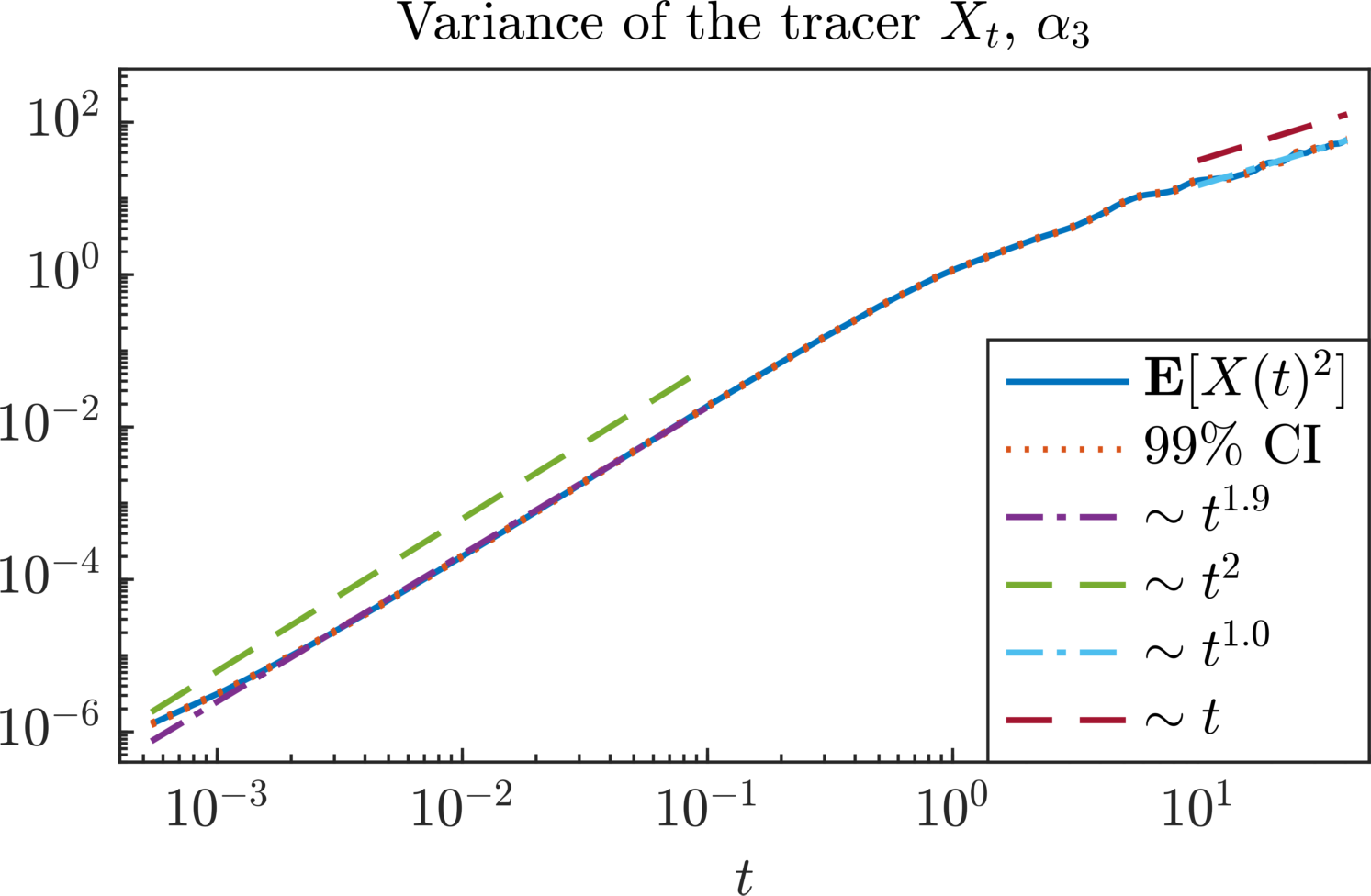}
\end{center}
\caption{Variance of the tracer $X_t$, for decreasing values of the damping parameter $\alpha$, i.e., $\alpha_1 > \alpha_2 > \alpha_3 > 0$.}
\label{fig:diffusion_tracer}
\end{figure}

We numerically verify the behavior predicted in \cref{sec:diffusion_large_time,sec:diffusion_small_time}. Although we introduced the general advection-diffusion equation with diffusion coefficient $\kappa$, the effective diffusion discussed in this section refers to the turbulent contribution $\kappa_T$. Therefore, the tracer simulations considered here correspond to the purely advective case $\kappa=0$, so that the observed dispersion is entirely due to the velocity field. For three decreasing values of the damping parameter $\alpha$, we simulate the Navier--Stokes equations in vorticity form \eqref{eq:NavierStokes_vorticity} and track the corresponding Lagrangian trajectories. For each trajectory, the displacement $X(t)$ along the $x$-axis is measured relative to its initial position, which is taken as the origin. In total, $M = 10^4$ trajectories are simulated up to the final time $T = 100$, with a time-step equal to that of the flow, corresponding to a CFL number of $0.3$. Particles are initially on a uniformly spaced grid. One flow realization is simulated to collect Lagrangian statistics. A bilinear interpolation is used to transfer the Eulerian velocity onto the particles, while a second-order Adams--Bashforth scheme is used for the Lagrangian tracking.

We first analyze the probability density function of the displacement $X(t)$, shown in \cref{fig:diffusion_tracer_PDF} at time $t = 10$, which is expected to be Gaussian based on the arguments above. This is confirmed for large values of the damping parameter $\alpha$, while deviations from Gaussianity emerge as $\alpha$ decreases.

Then, the variance of the displacement, $\E[X(t)^2]$, is estimated by means of a Monte Carlo average over the ensemble of trajectories. The numerical results are presented in \cref{fig:diffusion_tracer}, where we also report a 99\% confidence interval for $\E[X(t)^2]$. The confidence interval is not clearly visible in the figures because it overlaps the estimated mean, reflecting the high precision of the Monte Carlo estimate. Although the initial fitted slope $\sim t^{1.9}$ is slightly smaller than the $\sim t^2$ scaling predicted by equation \eqref{eq:slope_short_time}, with a relative error of approximately 5\%, the overall behavior is in agreement with the theoretical prediction and the discussion at the end of \cref{sec:diffusion_small_time}. In particular, the fitted and predicted long-time slopes $\sim t$ coincide.

\section{Conclusion} \label{sec:conclusion}

We have investigated, by direct numerical simulations, the classical problem of a Fourier synthetic description of a two-dimensional forced-damped turbulent velocity field in terms of random processes, and its applicability to capturing the transport of a passive scalar. We identify the existence of a characteristic time correlation length. The resulting structure exhibits short memory, in contrast with the natural conjecture that long memory might arise due to large-scale structures produced by the inverse cascade. However, the usual short-memory paradigm, namely the Ornstein--Uhlenbeck process or exponential decay of correlations, is not confirmed by the numerical experiments. The observed structure is different, and we propose a model in \cref{sec:model_process} that is able to reproduce the selected marginal and single-mode temporal statistics. In particular, we consider convolutions of white noise with appropriate kernels as the underlying stochastic processes for the Fourier components of turbulent flows.

We remark that the proposed model is a data-driven model constructed from the observed statistical properties of the numerical simulations performed for the range of Reynolds numbers and forcing configurations considered in this study. It does not rely on a first-principles derivation of the temporal decorrelation. Hence, establishing a physical interpretation of the model and determining the extent of its validity beyond the present parameter regime are important directions for future work. In particular, the present study does not allow us to claim that the results extend to three-dimensional turbulence, quasi-two-dimensional flows, or real turbulent flows.

Our proposal of a synthetic turbulent velocity field, made in \cref{sec:velocity_field}, is a finite sum of Fourier modes multiplied by processes of the form described in \cref{sec:model_process}. Furthermore, at longer time scales, the short memory becomes a second-order correction, and it is therefore meaningful to approximate the processes by white noises. The full procedure thus leads to a model of the velocity field in the long-time/small-correlation approximation.

In order to check the validity of the model, we have compared the direct numerical simulation of the transport of a passive scalar governed by the two-dimensional Navier--Stokes equations with analytical predictions based on the synthetic model. In the long-time regime, the approximation agrees well with the direct numerical simulation results, and in the short-time regime, it captures the nearly quadratic growth observed in the simulations.

We have mentioned above that we have not found significant memory effects in the time series investigated in this paper that could justify the introduction of fractional Gaussian noise or other long-memory processes. However, it is possible that such behavior may arise in special two-dimensional flows or in cases where very large structures emerge from the inverse cascade and related mechanisms. This topic is left for further investigation in future work.

\subsection*{Acknowledgments}

The research of PC and FF is funded by the European Union (ERC, NoisyFluid, No. 101053472). AZ is supported by ``Centro di Ricerca Matematica Ennio De Giorgi'' and the ``Emma e Giovanni Sansone'' Foundation, and is member of INdAM-GNCS.

\bibliographystyle{siamnodash}
\bibliography{biblio}

\end{document}